\documentclass[a4paper]{amsart}
\usepackage{fullpage}

	\usepackage{amsfonts,amssymb,amsmath,amsthm,gensymb,color,epic,eepic,float}
	\usepackage[mathscr]{euscript}
	\allowdisplaybreaks 

			\newtheorem{myTheorem}{Theorem}
			
			\newtheorem{myLemma}{Lemma}
			
\def\psivec{|\Psi^{(2)}_{\alpha_1\ldots\alpha_{\ell}}\rangle}
\def\vac{|0\rangle}
\def\avac{|\Uparrow_{\alpha}\rangle}

\def\li{\lambda_{\mbox{\rm\tiny I}}}
\def\lii{\lambda_{\mbox{\rm\tiny II}}}
\def\mi{\mu_{\mbox{\rm\tiny I}}}
\def\mii{\mu_{\mbox{\rm\tiny II}}}
\def\lc{\lambda^{C}}
\def\lci{\lambda^{C}_{\mbox{\rm\tiny I}}}
\def\lcii{\lambda^{C}_{\mbox{\rm\tiny II}}}

\def\lb{\lambda^{B}}
\def\lbi{\lambda^{B}_{\mbox{\rm\tiny I}}}
\def\lbii{\lambda^{B}_{\mbox{\rm\tiny II}}}

\def\mc{\mu^{C}}
\def\mci{\mu^{C}_{\mbox{\rm\tiny I}}}
\def\mcii{\mu^{C}_{\mbox{\rm\tiny II}}}

\def\mb{\mu^{B}}
\def\mbi{\mu^{B}_{\mbox{\rm\tiny I}}}
\def\mbii{\mu^{B}_{\mbox{\rm\tiny II}}}

\begin{document}

\title{Scalar products in generalized models with $SU(3)$-symmetry}

\author{M Wheeler}

\address{Laboratoire de Physique Th\'eorique et Hautes Energies, CNRS UMR 7589 and Universit\'e Pierre et Marie Curie (Paris 6), 4 place Jussieu, 75252 Paris cedex 05, France}
\email{mwheeler@lpthe.jussieu.fr}


\begin{abstract}
We consider a generalized model with $SU(3)$-invariant $R$-matrix, and review the nested Bethe Ansatz for constructing eigenvectors of the transfer matrix. A sum formula for the scalar product between generic Bethe vectors, originally obtained by Reshetikhin, is discussed. This formula depends on a certain partition function $Z(\{\lambda\},\{\mu\}|\{w\},\{v\})$, which we evaluate explicitly. In the limit when the variables $\{\mu\}$ or $\{v\} \rightarrow \infty$, this object reduces to the domain wall partition function of the six-vertex model $Z(\{\lambda\}|\{w\})$. Using this fact, we obtain a new expression for the off-shell scalar product (between a generic Bethe vector and a Bethe eigenvector), in the case when one set of Bethe variables tends to infinity. The expression obtained is a product of determinants, one of which is the Slavnov determinant from $SU(2)$ theory. 
\end{abstract} 

\maketitle

\setcounter{section}{0}

\section{Introduction}
\label{s:intro}

The calculation of scalar products between generic Bethe states is an extremely important area of study in models solvable by the Bethe Ansatz. On the one hand, the scalar product reduces to the norm-squared of a Bethe eigenvector in the limit where the states become on-shell ({\it i.e.}\/ when both states are parametrized by the same set of roots of the Bethe equations). On the other hand, off-shell scalar products (which in this work always means a scalar product between a generic Bethe vector and a Bethe eigenvector) play a key role in the study of correlation functions in such models. To have any chance of studying asymptotics of correlation functions and related quantities, it is therefore essential to have some manageable expression for the scalar products which are their building-blocks.

In two-dimensional models based on the $SU(2)$-invariant $R$-matrix, the theory of scalar products is well developed. There is a sum formula for the generic scalar product due to Izergin and Korepin (for the original result, see \cite{kor,ik}; for a more detailed derivation based on these earlier papers, see \cite{kbi} and Appendix {\bf A} of the present paper), a determinant formula for the on-shell scalar product proposed by Gaudin \cite{gau} and proved by Korepin \cite{kor}, and a determinant formula for the off-shell scalar product obtained by Slavnov \cite{sla}. The latter representation proved to be very helpful in the algebraic Bethe Ansatz approach to correlation functions of the XXX and XXZ models (see \cite{kmt} and the review article \cite{kmst}).

The subject is not so well understood in models based on higher-rank quantum algebras\footnote{Apart from the results discussed in the remainder of this paragraph, we also mention {\bf 1.} In the models based on $\mathcal{U}_q(\widehat{sl_3})$, multiple-integral formulae for generic scalar products were obtained in \cite{bpr}, {\bf 2.} Determinant formulae for the norm-squared in the $\mathcal{Y}(sl_n)$ and $\mathcal{U}_q(\widehat{sl_n})$ models, which generalize the result of Gaudin \cite{gau}, were obtained in \cite{tv2}. An analogous formula was conjectured in \cite{egsv} for spin chains based on Lie (super) algebras of arbitrarily high rank.}. In the case of models with the $SU(3)$-invariant $R$-matrix, there is a sum formula for the generic scalar product and a determinant formula for the on-shell scalar product, both obtained by Reshetikhin \cite{res} (see also Appendix {\bf B} of the present paper for an amplified discussion of the method used in \cite{res}). More recently, generalizing the work of \cite{res}, the scalar product between a Bethe eigenvector and a twisted Bethe eigenvector was expressed in determinant form \cite{bprs2}. Notably, no determinant formula is known for the off-shell scalar product in these models. The present paper aims to deal with precisely this problem, by evaluating the off-shell scalar product in a limiting case of its variables.

The generic $SU(3)$ scalar product is a function of four sets of variables 
$\{\lambda^C\},\{\lambda^B\},\{\mu^C\},\{\mu^B\}$ and the pseudo-vacuum eigenfunctions $a_{1},a_{2},a_{3}$. Let us denote it by 
$\langle \{\mu^C\},\{\lambda^C\} | \{\lambda^B\},\{\mu^B\} \rangle$, and assume that  $\{\lambda^B\},\{\mu^B\}$ satisfy the nested Bethe equations. Our approach is as follows. {\bf 1.} We take the sum expression for $\langle \{\mu^C\},\{\lambda^C\} | \{\lambda^B\},\{\mu^B\} \rangle$, as given in \cite{res}, as our starting point. This formula contains a certain function, $Z(\{\lambda\},\{\mu\} | \{w\},\{v\})$, which is expressed as the partition function of a lattice with particular boundary conditions. {\bf 2.} We calculate $Z(\{\lambda\},\{\mu\} | \{w\},\{v\})$ explicitly, and find that it is itself given by a sum. {\bf 3.} We list formulae for $Z(\{\lambda\},\{\mu\} | \{w\},\{v\})$ as one of its sets of variables tends to infinity. In this case, we find that it behaves as a domain wall partition function of the six-vertex model. {\bf 4.} Using the results of {\bf 1} and {\bf 3}, we obtain a sum formula for $\langle \{\mu^C\},\{\lambda^C\} | \{\infty\},\{\mu^B\} \rangle$ and $\langle \{\mu^C\},\{\lambda^C\} | \{\lambda^B\},\{\infty\} \rangle$, which denote the limiting cases $\{\lambda^B\} \rightarrow \infty$ and $\{\mu^B\} \rightarrow \infty$ of the off-shell scalar product. {\bf 5.} The summation in {\bf 4} factorizes into two parts, both of which can be evaluated as determinants using results from $SU(2)$ theory. Hence we obtain both $\langle \{\mu^C\},\{\lambda^C\} | \{\infty\},\{\mu^B\} \rangle$ and $\langle \{\mu^C\},\{\lambda^C\} | \{\lambda^B\},\{\infty\} \rangle$ as a product of two determinants.

Our result extends, and was partially motivated by a result of J~Caetano \cite{cae} in the context of {\bf 1.} An XXX spin chain with $SU(3)$-symmetry (\textit{i.e.} a spin chain constructed as an $L$-fold tensor product of fundamental representations of $\mathcal{Y}(sl_3)$), which is a special case of the generalized model presented in this paper (\textit{i.e.} we deal exclusively with operators acting on a Hilbert space $\mathcal{H}$ and never require to write it as $\mathcal{H} = \otimes_{i=1}^{L} \mathcal{H}_i$, or to prescribe $\mathcal{H}$ explicitly), and 
{\bf 2.} In the limit where both sets of Bethe roots $\{\lb\},\{\mb\} \rightarrow \infty$ simultaneously.

The paper is organized as follows. Section {\bf 2} and {\bf 3} review the algebraic and nested Bethe Ans\"atze for models with an $SU(2)$ and $SU(3)$-invariant $R$-matrix, respectively. These sections are designed to fix notation related to the models, such as their transfer matrices, Bethe eigenvectors and Bethe equations. The reader with familiarity of these subjects can skip these sections. Section {\bf 4} reviews results related to scalar products of the $SU(2)$ models. We list the sum formula for the generic scalar product \cite{kor,ik}, determinant formulae for the domain wall partition function \cite{ize,kos1,kos2}, and a determinant formula for an object which appeared recently in \cite{fw}, the {\it partial} domain wall partition function. We also give the determinant formula for the off-shell scalar product, discovered in \cite{sla}. Section {\bf 5} contains our new results relating to the scalar products of the $SU(3)$ models, and proceeds along the lines described in the previous paragraph. Section {\bf 6} contains concluding remarks. Appendices {\bf\ref{app:1}} and {\bf\ref{app:2}} contain detailed derivations of the sum formulae for the scalar product in $SU(2)$ and $SU(3)$-invariant models, respectively.

\section{Algebraic Bethe Ansatz for $SU(2)$-invariant models}

In this section we review the algebraic Bethe Ansatz for models with the $R$-matrix (\ref{Rmat-sl2}). For more details, see the seminal paper \cite{fst} on the Quantum Inverse Scattering Method, and the standard references \cite{fad} and \cite{kbi}.

\subsection{$SU(2)$-invariant $R$-matrix}

Let $V_{\alpha}, V_{\beta}$ be two copies of the vector space $\mathbb{C}^2$. The $SU(2)$-invariant $R$-matrix is given by
\begin{align}
R_{\alpha\beta}(\lambda,\mu)
=
\left(
\begin{array}{cc|cc}
f(\lambda,\mu) & 0 & 0 & 0 \\
0 & 1 & g(\lambda,\mu) & 0 \\ \hline
0 & g(\lambda,\mu) & 1 & 0 \\
0 & 0 & 0 & f(\lambda,\mu)
\end{array}
\right)_{\alpha\beta}
\label{Rmat-sl2}
\end{align}
where the subscript indicates that the $R$-matrix is an element of 
${\rm End}(V_{\alpha} \otimes V_{\beta})$. The entries $f(\lambda,\mu)$, $g(\lambda,\mu)$ are the simple rational functions\footnote{The $R$-matrix (\ref{Rmat-sl2}) occurs in the context of spin chains based on representations of $\mathcal{Y}(sl_2)$. An alternative solution of the Yang-Baxter equation, related to $\mathcal{U}_q(\widehat{sl_2})$, has the weights $f(\lambda,\mu)$, $g(\lambda,\mu)$ parametrized multiplicatively \cite{jim2} or in terms of trigonometric functions. For simplicity, in this paper we restrict our attention to rational models, but our results could be extended to the trigonometric case without difficulty.}
\begin{align}
	f(\lambda,\mu) = \frac{\lambda-\mu+1}{\lambda-\mu}, 
	\quad
	g(\lambda,\mu) = \frac{1}{\lambda-\mu}.
	\label{rat-wt} 
\end{align}
For later purposes, it is useful to represent the components of the $R$-matrix as vertices, as shown in Figure {\bf \ref{vert1}}. This is the well known connection with the six-vertex model of statistical mechanics \cite{bax1}.
\begin{figure}[H]

\begin{center}
\begin{minipage}{4.3in}

\setlength{\unitlength}{0.0005cm}
\begin{picture}(5000,6000)(3000,16500)

\put(6000,20000)
{$
\Big[R_{\alpha\beta}(\lambda,\mu)\Big]^{i_{\alpha} j_{\alpha}}_{i_{\beta} j_{\beta}}
=
$}


\path(14000,20000)(18000,20000) \put(12500,20000){$\lambda$}
\put(13500,20300){$i_{\alpha}$}
\put(18000,20300){$j_{\alpha}$} 


\path(16000,18000)(16000,22000) \put(15500,16500){$\mu$}
\put(16000,17500){$i_{\beta}$}
\put(16000,22500){$j_{\beta}$}

\end{picture}

\end{minipage}
\end{center}

\caption{Representing the components of the $R$-matrix as vertices. The indices $i_{\alpha},j_{\alpha} \in \{1,2\}$ denote block $(i_{\alpha},j_{\alpha})$ of (\ref{Rmat-sl2}), while $i_{\beta},j_{\beta} \in \{1,2\}$ denote the $(i_{\beta},j_{\beta})$-th component within that block.}

\label{vert1}

\end{figure}

\subsection{Generalized $SU(2)$ models}
\label{ssec:gen-su2}

We consider a general $SU(2)$-invariant model with the $2\times 2$ monodromy matrix
\begin{align}
T_{\alpha}(\lambda)
=
\left(
\begin{array}{cc}
A(\lambda) & B(\lambda)
\\
C(\lambda) & D(\lambda)
\end{array}
\right)_{\alpha}.
\label{mon-sl2}
\end{align}
The entries of (\ref{mon-sl2}) are operators which satisfy the Yang-Baxter algebra
\begin{align}
R_{\alpha\beta}(\lambda,\mu)
T_{\alpha}(\lambda)
T_{\beta}(\mu)
=
T_{\beta}(\mu)
T_{\alpha}(\lambda)
R_{\alpha\beta}(\lambda,\mu)
\label{int-sl2}
\end{align}
where $R_{\alpha\beta}(\lambda,\mu)$ is the $SU(2)$-invariant $R$-matrix (\ref{Rmat-sl2}). Such a model is solvable by the algebraic Bethe Ansatz provided one can find {\it pseudo-vacuum} states $|0\rangle,\langle 0|$ on which the operator entries of (\ref{mon-sl2}) act according to the rules
\begin{align}
A(\lambda) |0\rangle = a(\lambda) |0\rangle,
\quad
D(\lambda) |0\rangle = d(\lambda) |0\rangle,
\quad
C(\lambda) |0\rangle = 0,
\quad
B(\lambda) |0\rangle \not= 0
\label{ps-vac1}
\\
\langle 0| A(\lambda) = a(\lambda) \langle 0|,
\quad
\langle 0| D(\lambda) = d(\lambda) \langle 0|,
\quad
\langle 0| C(\lambda) \not= 0,
\quad
\langle 0| B(\lambda) = 0
\label{ps-vac2}
\end{align}
where $a(\lambda),d(\lambda)$ are rational functions of $\lambda$. In the sequel we let $\mathcal{H}$ denote the Hilbert space generated by the action of $B(\lambda)$ on $|0\rangle$. Similarly, we let $\mathcal{H}^{*}$ denote the Hilbert space generated by the action of $C(\lambda)$ on $\langle 0|$. In practice, $\mathcal{H}$ is a highest weight vector space of finite or infinite dimension, and $|0\rangle$ is its highest weight vector.

In this work we will not concern ourselves with the various technical questions which arise in the context of these models, such as {\bf 1.} The existence of a suitable $|0\rangle$, given $T_{\alpha}(\lambda)$ satisfying (\ref{int-sl2}), {\bf 2.} The converse problem, of constructing $T_{\alpha}(\lambda)$, given a pair of eigenfunctions $a(\lambda)$ and $d(\lambda)$, {\bf 3.} A complete characterization of the vector space $\mathcal{H}$. For more details on these points, we refer to \cite{kor2,tar1,tar2}. Some physical models which are special cases of the model described above include, for example, the XXX Heisenberg spin chain \cite{ft} and the quantum nonlinear Schr\"odinger equation. 

\subsection{Fundamental commutation relations}
\label{com-rel1}

To perform the algebraic Bethe Ansatz, one needs three commutation relations which can be extracted as individual components of (\ref{int-sl2}). Firstly, the $B$-operators commute,
\begin{align}
B(\lambda) B(\mu)
=
B(\mu) B(\lambda).
\label{2bb}
\end{align}
Secondly, we obtain the following relation between the $A$ and $B$-operators,
\begin{align}
A(\mu) B(\lambda)
=
f(\lambda,\mu)
B(\lambda) A(\mu)
-
g(\lambda,\mu)
B(\mu) A(\lambda).
\label{2ab}
\end{align}
Thirdly, we obtain the following relation between the $D$ and $B$-operators,
\begin{align}
D(\lambda) B(\mu)
=
f(\lambda,\mu)
B(\mu) D(\lambda)
-
g(\lambda,\mu)
B(\lambda) D(\mu).
\label{2db}
\end{align}

\subsection{Bethe Ansatz for eigenvectors}

The transfer matrix is the trace of the monodromy matrix (\ref{mon-sl2}) on $V_{\alpha}$,
\begin{align}
\mathcal{T}(x)
=
A(x)+D(x)
\end{align}
and it is the goal of the Bethe Ansatz to find states $|\Psi\rangle \in \mathcal{H}$ which are eigenvectors of $\mathcal{T}(x)$, satisfying
\begin{align}
\mathcal{T}(x) |\Psi\rangle
=
\Lambda(x) |\Psi\rangle.
\label{trans-sl2}
\end{align}
The Ansatz for the eigenvectors is a string of $B$-operators acting on the pseudo-vacuum,
\begin{align}
|\Psi\rangle
=
B(\lambda_1)
\dots
B(\lambda_{\ell})
|0\rangle
\label{bethe-state}
\end{align}
where $\lambda_i \not= \lambda_j$ for all $i\not= j$. Such vectors span $\mathcal{H}$, by definition. In the case where $\mathcal{H}$ is finite dimensional, $|\Psi\rangle$ vanishes for $\ell$ sufficiently large.
 
This choice for $|\Psi\rangle$ gives a solution of the equation (\ref{trans-sl2}), provided that the variables $\{\lambda_1,\dots,\lambda_{\ell}\}$ satisfy the Bethe equations\footnote{When the Bethe equations apply to the variables $\{\lambda_1,\dots,\lambda_{\ell}\}$ we will call the object (\ref{bethe-state}) a {\it Bethe eigenvector}. When the variables $\{\lambda_1,\dots,\lambda_{\ell}\}$ are free we call it a {\it generic Bethe vector.}}. We give the details in the following subsections.

\subsection{Action of $A(x)$ on $|\Psi\rangle$}
\label{ssec:actionA}

Using the relation (\ref{2ab}) and the commutativity (\ref{2bb}) of the $B$-operators we obtain the formula
\begin{multline}
A(x)
B(\lambda_1)
\dots
B(\lambda_{\ell})
|0\rangle
=
\left[
\prod_{i=1}^{\ell}
f(\lambda_i,x)
B(\lambda_i)
\right]
A(x)
\vac
\\
-\sum_{i=1}^{\ell}
g(\lambda_i,x)
B(x)
\left[
\prod_{\substack{j\not=i \\ j=1}}^{\ell}
f(\lambda_j,\lambda_i)
B(\lambda_j)
\right]
A(\lambda_i)
\vac
\end{multline}
which allows us to move the $A$-operator entirely to the right of all $B$-operators, so that it acts on the pseudo-vacuum. In view of the defining properties (\ref{ps-vac1}) of the pseudo-vacuum, we have 
$A(x) |0\rangle = a(x) |0\rangle$ for all parameters $x$, which gives
\begin{multline}
A(x)
B(\lambda_1)
\dots
B(\lambda_{\ell})
|0\rangle
=
a(x)
\left[
\prod_{i=1}^{\ell}
f(\lambda_i,x)
B(\lambda_i)
\right]
\vac
\label{abbb}
\\
-\sum_{i=1}^{\ell}
a(\lambda_i)
g(\lambda_i,x)
B(x)
\left[
\prod_{\substack{j\not=i \\ j=1}}^{\ell}
f(\lambda_j,\lambda_i)
B(\lambda_j)
\right]
\vac .
\end{multline}
Following \cite{fad}, we will refer to the terms in the summation of (\ref{abbb}) as \textit{unwanted terms,} since one ultimately needs to eliminate these to obtain a genuine eigenvector of the transfer matrix. 

\subsection{Action of $D(x)$ on $|\Psi\rangle$}
\label{ssec:actionD}

Similarly, using the relation (\ref{2db}) and the commutativity (\ref{2bb}) of the $B$-operators we obtain the formula
\begin{multline}
D(x)
B(\lambda_1) \ldots B(\lambda_{\ell})
\vac
=
\left[
\prod_{i=1}^{\ell}
f(x,\lambda_i)
B(\lambda_i)
\right]
D(x)
\vac
\\
-\sum_{i=1}^{\ell}
g(x,\lambda_i)
B(x)
\left[
\prod_{\substack{j\not=i \\ j=1}}^{\ell}
f(\lambda_i,\lambda_j)
B(\lambda_j)
\right]
D(\lambda_i)
\vac
\end{multline}
which allows us to move the $D$-operator entirely to the right of all $B$-operators, so that it acts on the pseudo-vacuum. Recalling from (\ref{ps-vac1}) that $D(x) |0\rangle = d(x) |0\rangle$ for all $x$, we obtain
\begin{multline}
D(x)
B(\lambda_1) \ldots B(\lambda_{\ell})
\vac
=
d(x)
\left[
\prod_{i=1}^{\ell}
f(x,\lambda_i)
B(\lambda_i)
\right]
\vac
\label{dbbb}
\\
-\sum_{i=1}^{\ell}
d(\lambda_i)
g(x,\lambda_i)
B(x)
\left[
\prod_{\substack{j\not=i \\ j=1}}^{\ell}
f(\lambda_i,\lambda_j)
B(\lambda_j)
\right]
\vac .
\end{multline}
Once again, the unwanted terms are all those in the summation of (\ref{dbbb}).

\subsection{Expression for $\Lambda(x)$ and Bethe equations}

The action of the transfer matrix on $|\Psi\rangle$ is given by summing (\ref{abbb}) and (\ref{dbbb}). We see that
\begin{align}
\mathcal{T}(x) |\Psi\rangle
=
\left[
a(x) \prod_{i=1}^{\ell} f(\lambda_i,x)
+
d(x) \prod_{i=1}^{\ell} f(x,\lambda_i)
\right]
|\Psi\rangle
\label{lambda-sl2}
\end{align}
provided that the unwanted terms in (\ref{abbb}) and (\ref{dbbb}) cancel. This is achieved by assuming that the variables $\{\lambda_1,\dots,\lambda_{\ell}\}$ satisfy the Bethe equations
\begin{align}
a(\lambda_i)
\prod_{j\not=i}^{\ell}
f(\lambda_j,\lambda_i)
-
d(\lambda_i)
\prod_{j\not=i}^{\ell}
f(\lambda_i,\lambda_j)
=
0,
\quad\quad
\forall\ 1 \leq i \leq \ell.
\end{align}
We will find it convenient to rearrange these equations, and write them in the form

\begin{align}
r(\lambda_i)
\equiv
\frac{a(\lambda_i)}{d(\lambda_i)}
=
-
\prod_{j=1}^{\ell}
\frac{\lambda_i-\lambda_j+1}{\lambda_i-\lambda_j-1},
\quad\quad
\forall\ 1 \leq i \leq \ell .
\label{bethe-sl2}
\end{align}
This conludes the construction of the eigenvectors of $\mathcal{T}(x)$, and the eigenvalues $\Lambda(x)$ can be read as the coefficient of the right hand side of (\ref{lambda-sl2}).

\subsection{Dual Bethe eigenvectors}

The above procedure can also be applied to finding states $\langle \Psi| \in \mathcal{H}^{*}$ which are eigenvectors of $\mathcal{T}(x)$, satisfying
\begin{align}
\langle \Psi| \mathcal{T}(x)
=
\Lambda(x)
\langle \Psi|.
\label{dual-eig}
\end{align}
Following similar steps to those already outlined one can show that
\begin{align}
\langle \Psi|
=
\langle 0|
C(\lambda_1)
\dots
C(\lambda_{\ell})
\label{dual-bethe}
\end{align}
satisfies (\ref{dual-eig}), with the same eigenvalue as that calculated in the last subsection, provided that the variables $\{\lambda_1,\dots,\lambda_{\ell}\}$ obey the equations (\ref{bethe-sl2}). For brevity, we omit these details.

\section{Nested Bethe Ansatz for $SU(3)$-invariant models}
\label{s:nba}

In this section we present a fairly detailed outline of the nested Bethe Ansatz \cite{kr}, for generic models based on the $SU(3)$-invariant $R$-matrix. Our exposition closely follows \cite{res,br}. An alternative formulation of the method, leading to combinatorial formulae for the Bethe vectors, can be found in \cite{tv1}.

\subsection{$SU(3)$-invariant $R$-matrix and related definitions}

Let $V_{\alpha}, V_{\beta}$ be two copies of the vector space $\mathbb{C}^3$. The $SU(3)$-invariant $R$-matrix is given by
\begin{align}
R_{\alpha\beta}^{(1)}(\lambda,\mu)
=
\left(
\begin{array}{ccc|ccc|ccc}
f(\lambda,\mu) & 0 & 0 & 0 & 0 & 0 & 0 & 0 & 0 \\
0 & 1 & 0 & g(\lambda,\mu) & 0 & 0 & 0 & 0 & 0 \\
0 & 0 & 1 & 0 & 0 & 0 & g(\lambda,\mu) & 0 & 0 \\ \hline
0 & g(\lambda,\mu) & 0 & 1 & 0 & 0 & 0 & 0 & 0 \\
0 & 0 & 0 & 0 & f(\lambda,\mu) & 0 & 0 & 0 & 0 \\
0 & 0 & 0 & 0 & 0 & 1 & 0 & g(\lambda,\mu) & 0 \\ \hline
0 & 0 & g(\lambda,\mu) & 0 & 0 & 0 & 1 & 0 & 0 \\
0 & 0 & 0 & 0 & 0 & g(\lambda,\mu) & 0 & 1 & 0 \\
0 & 0 & 0 & 0 & 0 & 0 & 0 & 0 & f(\lambda,\mu) 
\end{array}
\right)_{\alpha\beta}
\label{Rmat-sl3}
\end{align}
where the subscript indicates that the $R$-matrix is an element of 
${\rm End}(V_{\alpha} \otimes V_{\beta})$. Define\footnote{The reason for introducing the second type of $R$-matrix (\ref{Rmat-sl3t}) is explained in \cite{res}. For now we just remark that it is necessary to give a description of $SU(3)$ scalar products in complete generality, as we will see in Subsection {\bf\ref{pf-resh}}.}
\begin{align}
R_{\alpha\beta}^{*(1)}(\lambda,\mu)
=
\Big(
R_{\alpha\beta}^{(1)}(-\lambda,-\mu)
\Big)^{{\rm t}_{\beta}}
\end{align}
where ${\rm t}_{\beta}$ stands for transposition on the space $V_{\beta}$. More explicitly, we write $R_{\alpha\beta}^{*(1)}(\lambda,\mu)$ in the matrix form
\begin{align}
R_{\alpha\beta}^{*(1)}(\lambda,\mu)
=
\left(
\begin{array}{ccc|ccc|ccc}
f(-\lambda,-\mu) & 0 & 0 & 0 & g(-\lambda,-\mu) & 0 & 0 & 0 & g(-\lambda,-\mu) \\
0 & 1 & 0 & 0 & 0 & 0 & 0 & 0 & 0 \\
0 & 0 & 1 & 0 & 0 & 0 & 0 & 0 & 0 \\ \hline
0 & 0 & 0 & 1 & 0 & 0 & 0 & 0 & 0 \\
g(-\lambda,-\mu) & 0 & 0 & 0 & f(-\lambda,-\mu) & 0 & 0 & 0 & g(-\lambda,-\mu) \\
0 & 0 & 0 & 0 & 0 & 1 & 0 & 0 & 0 \\ \hline
0 & 0 & 0 & 0 & 0 & 0 & 1 & 0 & 0 \\
0 & 0 & 0 & 0 & 0 & 0 & 0 & 1 & 0 \\
g(-\lambda,-\mu) & 0 & 0 & 0 & g(-\lambda,-\mu) & 0 & 0 & 0 & f(-\lambda,-\mu) 
\end{array}
\right)_{\alpha\beta}
\label{Rmat-sl3t}
\end{align}
and represent its components with the vertex shown in Figure {\bf\ref{vert2}}.

\begin{figure}[H]

\begin{center}
\begin{minipage}{4.3in}

\setlength{\unitlength}{0.0005cm}
\begin{picture}(10000,6000)(7000,16500)

\put(5000,20000)
{$
\Big[R_{\alpha\beta}^{*(1)}(\lambda,\mu)\Big]^{i_{\alpha} j_{\alpha}}_{i_{\beta} j_{\beta}}
=
$}



\path(14000,20000)(18000,20000) \put(12500,20000){$\lambda$}
\put(13500,20300){$i_{\alpha}$}
\put(18000,20300){$j_{\alpha}$} 


\path(16000,18000)(16000,22000) \put(16000,16500){$\mu$}
\put(16000,17500){$i_{\beta}$}
\put(16000,22500){$j_{\beta}$}


\put(16000,20000){\circle*{200}}

\put(20000,20000){$=$}




\path(24000,20000)(28000,20000) \put(21500,20000){$-\lambda$}
\put(23500,20300){$i_{\alpha}$}
\put(28000,20300){$j_{\alpha}$} 


\path(26000,18000)(26000,22000) \put(25500,16500){$-\mu$}
\put(26000,17500){$j_{\beta}$}
\put(26000,22500){$i_{\beta}$}

\end{picture}

\end{minipage}
\end{center}

\caption{Definition of the dotted vertex, which is used to denote components of the matrix (\ref{Rmat-sl3t}). The relationship with an ordinary vertex is shown on the right. The indices now take the values $i_{\alpha},j_{\alpha},i_{\beta},j_{\beta} \in \{1,2,3\}$.}  

\label{vert2}

\end{figure}

Because the nested Bethe Ansatz involves a reduction of the $SU(3)$ eigenvector problem to an $SU(2)$ one, we also need the definition
\begin{align}
R^{(2)}_{\alpha\beta}(\lambda,\mu)
=
\left(
\begin{array}{cc|cc}
f(\lambda,\mu) & 0 & 0 & 0 \\
0 & 1 & g(\lambda,\mu) & 0 \\ \hline
0 & g(\lambda,\mu) & 1 & 0 \\
0 & 0 & 0 & f(\lambda,\mu)
\end{array}
\right)_{\alpha\beta}
\label{Rmat-sl32}
\end{align}
which is nothing but the $SU(2)$-invariant $R$-matrix. Finally let us define 
$\mathbb{R}^{(2)}_{\alpha\beta}(\lambda,\mu) = 
R^{(2)}_{\alpha\beta}(\lambda,\mu)/f(\lambda,\mu)$. Then we have the relation
\begin{align}
\mathbb{R}^{(2)}_{\alpha\beta}(\lambda,\lambda)
=
P^{(2)}_{\alpha\beta}
=
\left(
\begin{array}{cc|cc}
1 & 0 & 0 & 0 \\
0 & 0 & 1 & 0 \\ \hline
0 & 1 & 0 & 0 \\
0 & 0 & 0 & 1
\end{array}
\right)_{\alpha\beta}
\end{align}
where $P^{(2)}_{\alpha\beta}$ is the permutation matrix acting on 
$V_{\alpha} \otimes V_{\beta}$.

\subsection{Yang-Baxter equations}

For $i\in \{1,2\}$, let $V_{\alpha}, V_{\beta}, V_{\gamma}$ be three copies of the vector space $\mathbb{C}^{4-i}$ with the associated rapidities $\lambda,\mu,\nu$. The $R$-matrices (\ref{Rmat-sl3}) and (\ref{Rmat-sl32}) satisfy the Yang-Baxter equation
\begin{align}
R_{\alpha\beta}^{(i)}(\lambda,\mu)
R_{\alpha\gamma}^{(i)}(\lambda,\nu)
R_{\beta\gamma}^{(i)}(\mu,\nu)
=
R_{\beta\gamma}^{(i)}(\mu,\nu)
R_{\alpha\gamma}^{(i)}(\lambda,\nu)
R_{\alpha\beta}^{(i)}(\lambda,\mu)
\label{yb1}
\end{align}
for both values of $i \in \{1,2\}$. We also have
\begin{align}
R^{(1)}_{\alpha\beta}(\lambda,\mu)
R^{*(1)}_{\alpha\gamma}(\lambda,\nu)
R^{*(1)}_{\beta\gamma}(\mu,\nu)
=
R^{*(1)}_{\beta\gamma}(\mu,\nu)
R^{*(1)}_{\alpha\gamma}(\lambda,\nu)
R^{(1)}_{\alpha\beta}(\lambda,\mu).
\label{yb2}
\end{align}

\subsection{Generalized $SU(3)$ models}
\label{ssec:gen-su3}

We consider a general $SU(3)$-invariant model with the $3\times 3$ monodromy matrix
\begin{align}
T^{(1)}_{\alpha}(\lambda)
=
\left(
\begin{array}{ccc}
t_{11}(\lambda) & t_{12}(\lambda) & t_{13}(\lambda) \\
t_{21}(\lambda) & t_{22}(\lambda) & t_{23}(\lambda) \\
t_{31}(\lambda) & t_{32}(\lambda) & t_{33}(\lambda)
\end{array}
\right)_{\alpha} .
\label{mon-sl3}
\end{align}
The entries of (\ref{mon-sl3}) are operators which satisfy the Yang-Baxter algebra
\begin{align}
R_{\alpha\beta}^{(1)}(\lambda,\mu)
T_{\alpha}^{(1)}(\lambda)
T_{\beta}^{(1)}(\mu)
=
T_{\beta}^{(1)}(\mu)
T_{\alpha}^{(1)}(\lambda)
R_{\alpha\beta}^{(1)}(\lambda,\mu)
\label{int-sl3}
\end{align}
where $R_{\alpha\beta}^{(1)}(\lambda,\mu)$ is the $SU(3)$-invariant $R$-matrix (\ref{Rmat-sl3}). Again, the model is solvable by the (nested) algebraic Bethe Ansatz if one can find pseudo-vacuum states $|0\rangle,\langle 0|$ which are acted upon by the operators (\ref{mon-sl3}) according to the rules 
\begin{align}
t_{ii}(\lambda) |0\rangle = a_i(\lambda) |0\rangle,
\quad
t_{kj}(\lambda) |0\rangle = 0,
\quad
t_{jk}(\lambda) |0\rangle \not= 0
\label{ps-vac3}
\\
\langle 0| t_{ii}(\lambda) = a_i(\lambda) \langle 0|,
\quad
\langle 0| t_{kj}(\lambda) \not= 0,
\quad
\langle 0| t_{jk}(\lambda) = 0
\label{ps-vac4}
\end{align}
valid for all $i \in \{1,2,3\}$ and $1\leq j < k \leq 3$, and where the $a_i(\lambda)$ are rational functions (which we call the pseudo-vacuum eigenfunctions). Let $\mathcal{H}$ denote the Hilbert space generated by the action of operators $t_{jk}(\lambda)$ on $|0\rangle$, and $\mathcal{H}^{*}$ the Hilbert space generated by $t_{kj}(\lambda)$ on $\langle 0|$, for all $1\leq j < k \leq 3$.

Again, numerous technical questions arise in the context of such models, which are the obvious extension of the points raised in the closing remarks of Subsection {\bf\ref{ssec:gen-su2}}. For our purposes, we treat the relations (\ref{int-sl3})--(\ref{ps-vac4}) as given \textit{a priori,} and proceed from there. Since (\ref{int-sl3}) are the defining equations of $\mathcal{Y}(sl_3)$, the explicit construction of a pseudo-vacuum satisfying (\ref{ps-vac3}) and (\ref{ps-vac4}) is related to the representation theory of the Yangian. On this point, let us remark that the highest weight representation of the Yangian is due to Drinfel'd \cite{dri}. Further details can be found in \cite{mol}, and references therein.   

\subsection{Decomposition of monodromy matrix}

In the following we consider a $2 \times 2$ decomposition of the monodromy matrix, by defining $A^{(1)}(\lambda) = t_{11}(\lambda)$ and
\begin{align}
B_{\beta}^{(1)}(\lambda) 
=
\left(
\begin{array}{cc}
t_{12}(\lambda) & t_{13}(\lambda)
\end{array}
\right)_{\beta},
\quad\quad
C_{\gamma}^{(1)}(\lambda)
=
\left(
\begin{array}{c}
t_{21}(\lambda) \\
t_{31}(\lambda)
\end{array}
\right)_{\gamma},
\quad\quad
D^{(1)}_{\delta}(\lambda)
=
\left(
\begin{array}{cc}
t_{22}(\lambda) & t_{23}(\lambda) \\
t_{32}(\lambda) & t_{33}(\lambda)
\end{array}
\right)_{\delta}.
\label{decomposition} 
\end{align}
Here $V_{\beta}, V_{\gamma}, V_{\delta}$ are copies of $\mathbb{C}^2$ and the subscripts on these matrices are used to denote the fact that
\begin{align} 
B^{(1)}_{\beta}(\lambda) \in V_{\beta}^{*},\quad\quad
C^{(1)}_{\gamma}(\lambda) \in V_{\gamma}, \quad\quad
D^{(1)}_{\delta}(\lambda) \in {\rm End}(V_{\delta}).
\end{align}

\subsection{First set of commutation relations}

In the following we assume $V_{\alpha}, V_{\beta}$ are copies of $\mathbb{C}^2$. Firstly, we list the commutation between the $B$-operators, 
\begin{align}
B^{(1)}_{\alpha}(\lambda) B^{(1)}_{\beta}(\mu)
=
B^{(1)}_{\beta}(\mu) B^{(1)}_{\alpha}(\lambda)
\mathbb{R}^{(2)}_{\alpha\beta}(\lambda,\mu).
\label{reorder}
\end{align}
Secondly, we give the commutation between the $A$ and $B$-operators,
\begin{align}
A^{(1)}(\mu)
B_{\alpha}^{(1)}(\lambda)
=
f(\lambda,\mu)
B_{\alpha}^{(1)}(\lambda)
A^{(1)}(\mu)
-
g(\lambda,\mu)
B_{\alpha}^{(1)}(\mu)
A^{(1)}(\lambda).
\label{1ab}
\end{align}
Thirdly, we give the commutation between the $D$ and $B$-operators,
\begin{align}
D^{(1)}_{\alpha}(\lambda)
B_{\beta}^{(1)}(\mu)
=
f(\lambda,\mu)
B_{\beta}^{(1)}(\mu)
D^{(1)}_{\alpha}(\lambda)
\mathbb{R}^{(2)}_{\alpha\beta}(\lambda,\mu)
-
g(\lambda,\mu)
B_{\beta}^{(1)}(\lambda)
D^{(1)}_{\alpha}(\mu)
P^{(2)}_{\alpha\beta}.
\label{1db}
\end{align}
Finally, the commutation between the $D$-operators reduces to an intertwining equation of $SU(2)$ type, namely
\begin{align}
R^{(2)}_{\alpha\beta}(\lambda,\mu)
D^{(1)}_{\alpha}(\lambda)
D^{(1)}_{\beta}(\mu)
=
D^{(1)}_{\beta}(\mu)
D^{(1)}_{\alpha}(\lambda)
R^{(2)}_{\alpha\beta}(\lambda,\mu).
\label{1dd}
\end{align}

\subsection{First expression for Bethe eigenvectors}

The aim of the nested Bethe Ansatz is to find the eigenvectors and eigenvalues of the transfer matrix 
$\mathcal{T}^{(1)}(x)
= t_{11}(x) + t_{22}(x) + t_{33}(x)$. We will denote an eigenvector by 
$
|\Psi^{(1)} \rangle 
\in 
\mathcal{H}
$
and its corresponding eigenvalue by $\Lambda^{(1)}(x)$, that is, 
\begin{align}
\mathcal{T}^{(1)}(x)
|\Psi^{(1)}\rangle
=
\Big[
A^{(1)}(x) + {\rm tr}_{\beta} D^{(1)}_{\beta}(x)
\Big]
|\Psi^{(1)}\rangle
=
\Lambda^{(1)}(x)
|\Psi^{(1)}\rangle.
\label{1trans}
\end{align}
The first step in the Ansatz for the eigenvectors is to propose that
\begin{align}
|\Psi^{(1)}\rangle
=
B_{\alpha_1}^{(1)}(\lambda_1)
\ldots
B_{\alpha_{\ell}}^{(1)}(\lambda_{\ell})
\psivec
\label{1eigvec}
\end{align}
where the reference state on the right hand side satisfies
\begin{align}
\psivec
\in 
\mathcal{H}
\otimes
V_{\alpha_1}\otimes \cdots \otimes V_{\alpha_{\ell}}
\end{align}
with $V_{\alpha_1},\dots,V_{\alpha_{\ell}}$ all being copies of $\mathbb{C}^2$. In the following subsections we will derive the necessary conditions for (\ref{1eigvec}) to be an eigenstate of the transfer matrix.

\subsection{Action of $A^{(1)}(x)$ on $|\Psi^{(1)}\rangle$}

Making repeated use of the commutation relations (\ref{reorder}) and (\ref{1ab}), one can derive the equation
\begin{multline}
A^{(1)}(x)
B_{\alpha_1}^{(1)}(\lambda_1)
\ldots
B_{\alpha_{\ell}}^{(1)}(\lambda_{\ell})
\psivec
=
\left[
\prod_{i=1}^{\ell}
f(\lambda_i,x)
B_{\alpha_i}^{(1)}(\lambda_i)
\right] 
A^{(1)}(x)
\psivec
\label{1abbb}
\\
-\sum_{i=1}^{\ell}
g(\lambda_i,x)
B_{\alpha_i}^{(1)}(x)
\left[
\prod_{\substack{j\not= i \\ j=1}}^{\ell}
f(\lambda_j,\lambda_i)
B_{\alpha_j}^{(1)}(\lambda_j)
\right]
\left[
\prod_{j < i}
\mathbb{R}^{(2)}_{\alpha_j \alpha_i}(\lambda_j,\lambda_i)
\right]
A^{(1)}(\lambda_i)
\psivec
\end{multline}
where we have defined $\mathbb{R}^{(2)}_{\alpha\beta}(\lambda,\mu) = R^{(2)}_{\alpha\beta}(\lambda,\mu)/f(\lambda,\mu)$. Notice that in every term on the right hand side, the $A$-operator has been threaded past all $B$-operators.

\subsection{Action of $D^{(1)}_{\beta}(x)$ on $|\Psi^{(1)}\rangle$}

Making repeated use of the commutation relations (\ref{reorder}) and (\ref{1db}), it is similarly possible to derive the equation
\begin{multline}
D^{(1)}_{\beta}(x)
B_{\alpha_1}^{(1)}(\lambda_1)
\ldots
B_{\alpha_{\ell}}^{(1)}(\lambda_{\ell})
\psivec
=
\left[
\prod_{i=1}^{\ell}
f(x,\lambda_i)
B_{\alpha_i}^{(1)}(\lambda_i)
\right]
T^{(2)}_{\beta}(x)
\psivec
\label{1dbbb}
\\
-\sum_{i=1}^{\ell}
g(x,\lambda_i)
B_{\alpha_i}^{(1)}(x)
\left[
\prod_{\substack{j \not= i \\ j=1}}^{\ell}
f(\lambda_i,\lambda_j)
B_{\alpha_j}^{(1)}(\lambda_j)
\right]
\left[
\prod_{j < i}
\mathbb{R}^{(2)}_{\alpha_j \alpha_i}(\lambda_j,\lambda_i)
\right]
T^{(2)}_{\beta}(\lambda_i)
\psivec
\end{multline}
where we have defined the monodromy matrix of $SU(2)$ type
\begin{align}
T^{(2)}_{\beta}(x)
\equiv
T^{(2)}_{\beta}(x|\lambda_{\ell},\dots,\lambda_1)
=
D^{(1)}_{\beta}(x)
\mathbb{R}^{(2)}_{\beta \alpha_{\ell}}(x,\lambda_{\ell})
\ldots
\mathbb{R}^{(2)}_{\beta \alpha_1}(x,\lambda_1).
\label{2mon}
\end{align}
Since we wish to compute the action of ${\rm tr}_{\beta} D_{\beta}^{(1)}(x)$ on $|\Psi^{(1)}\rangle$, what actually interests us about (\ref{1dbbb}) is its trace on $V_{\beta}$, which can be taken trivially,
\begin{multline}
{\rm tr}_{\beta} D^{(1)}_{\beta}(x)
B_{\alpha_1}^{(1)}(\lambda_1)
\ldots
B_{\alpha_{\ell}}^{(1)}(\lambda_{\ell})
\psivec
=
\left[
\prod_{i=1}^{\ell}
f(x,\lambda_i)
B_{\alpha_i}^{(1)}(\lambda_i)
\right]
\mathcal{T}^{(2)}(x)
\psivec
\label{1dbbb'}
\\
-\sum_{i=1}^{\ell}
g(x,\lambda_i)
B_{\alpha_i}^{(1)}(x)
\left[
\prod_{\substack{j \not= i \\ j=1}}^{\ell}
f(\lambda_i,\lambda_j)
B_{\alpha_j}^{(1)}(\lambda_j)
\right]
\left[
\prod_{j < i}
\mathbb{R}^{(2)}_{\alpha_j \alpha_i}(\lambda_j,\lambda_i)
\right]
\mathcal{T}^{(2)}(\lambda_i)
\psivec
\end{multline}
where we have defined the secondary transfer matrix $\mathcal{T}^{(2)}(x)= {\rm tr}_{\beta} T^{(2)}_{\beta}(x)$.

\subsection{First eigenvalue $\Lambda^{(1)}(x)$}

Assume that the reference state $\psivec$ is an eigenvector of both $A^{(1)}(x)$ and $\mathcal{T}^{(2)}(x)$, satisfying the equations
\begin{align}
A^{(1)}(x)
\psivec
&=
a_1(x)
\psivec,
\label{ass}
\\
\mathcal{T}^{(2)}(x)
\psivec
&=
\Lambda^{(2)}(x)
\psivec.
\label{2trans}
\end{align}
The validation of these two equations will be achieved by our subsequent choice of $\psivec$. In particular, it will be our aim to construct solutions of (\ref{2trans}), as the second and final step of the nested Bethe Ansatz. Having done so, it will be possible to verify that the resulting expression for $\psivec$ satisfies (\ref{ass}).

By adding (\ref{1abbb}) to (\ref{1dbbb'}) and using the assumptions (\ref{ass}) and (\ref{2trans}), we obtain
\begin{align}
\mathcal{T}^{(1)}(x)
|\Psi^{(1)}\rangle
=
\Lambda^{(1)}(x)
|\Psi^{(1)}\rangle
=
\left(
a_1(x)
\prod_{i=1}^{\ell} f(\lambda_i,x)
+
\prod_{i=1}^{\ell} f(x,\lambda_i)
\Lambda^{(2)}(x)
\right)
|\Psi^{(1)}\rangle
\label{1trans-2}
\end{align}
if we assume that all terms not proportional to $|\Psi^{(1)}\rangle$ cancel. This assumption is equivalent to enforcing the first set of Bethe equations, which we discuss in the next subsection. Reading the coefficient in equation (\ref{1trans-2}), we have our first expression for the eigenvalue $\Lambda^{(1)}(x)$, 
\begin{align}
\Lambda^{(1)}(x)
=
a_1(x)
\prod_{i=1}^{\ell} f(\lambda_i,x)
+
\prod_{i=1}^{\ell} f(x,\lambda_i)
\Lambda^{(2)}(x) .
\label{1eigval}
\end{align}

\subsection{First set of Bethe equations}

As we just mentioned, $|\Psi^{(1)}\rangle$ is an eigenstate of 
$\mathcal{T}^{(1)}(x)$ if and only if the unwanted terms in (\ref{1abbb}) and (\ref{1dbbb'}) sum to zero. From the form of the weights (\ref{rat-wt}) it is clear that $g(\lambda_i,x)=-g(x,\lambda_i)$, hence the unwanted terms in (\ref{1abbb}) and (\ref{1dbbb'}) cancel if the variables $\{\lambda_1,\dots,\lambda_{\ell}\}$ satisfy the Bethe equations
\begin{align}
a_1(\lambda_i)
\prod_{j\not=i}^{\ell}
f(\lambda_j,\lambda_i)
-
\prod_{j\not=i}^{\ell}
f(\lambda_i,\lambda_j)
\Lambda^{(2)}(\lambda_i)
=
0.
\end{align}
After simple manipulation, the Bethe equations can be expressed in the form
\begin{align}
\frac{a_1(\lambda_i)}{\Lambda^{(2)}(\lambda_i)}
=
-
\prod_{j=1}^{\ell}
\frac{\lambda_i-\lambda_j+1}{\lambda_i-\lambda_j-1}.
\end{align}

\subsection{Second set of commutation relations}

By virtue of the commutation relation (\ref{1dd}) and the Yang-Baxter equation (\ref{yb1}), the $SU(2)$ monodromy matrix (\ref{2mon}) obeys its own intertwining equation,
\begin{align}
R^{(2)}_{\alpha\beta}(x,y)
T^{(2)}_{\alpha}(x)
T^{(2)}_{\beta}(y)
=
T^{(2)}_{\beta}(y)
T^{(2)}_{\alpha}(x)
R^{(2)}_{\alpha\beta}(x,y).
\end{align}
If we explicitly exhibit the $V_{\alpha}$ dependence of $T^{(2)}_{\alpha}(x)$, by writing
\begin{align}
T^{(2)}_{\alpha}(x|\lambda_{\ell},\dots,\lambda_1)
=
\left(
\begin{array}{cc}
A^{(2)}(x|\lambda_{\ell},\dots,\lambda_1) 
& 
B^{(2)}(x|\lambda_{\ell},\dots,\lambda_1)
\\
C^{(2)}(x|\lambda_{\ell},\dots,\lambda_1) 
& 
D^{(2)}(x|\lambda_{\ell},\dots,\lambda_1)
\end{array}
\right)_{\alpha}
\end{align}
then the commutation relations between its operator entries are the same as those in Subsection {\bf\ref{com-rel1}}. The only point of difference is that each operator now comes with the superscript ${}^{(2)}$, that is, one should replace $A(x) \rightarrow A^{(2)}(x)$ and so on. 

\subsection{Second set of Bethe eigenvectors}

It is now our goal to obtain solutions to the equation (\ref{2trans}), which can be done using the ordinary algebraic Bethe Ansatz for $SU(2)$ models. Namely, we let
\begin{align}
\psivec
=
B^{(2)}(\mu_1)
\ldots
B^{(2)}(\mu_m)
\vac \otimes \avac,
\quad
\avac
=
\bigotimes_{i=1}^{\ell}
\left(
\begin{array}{c}
1 \\ 0
\end{array}
\right)_{\alpha_i}
\label{2eigvec}
\end{align}
where the $B$-operators act on 
$\mathcal{H} \otimes V_{\alpha_1} \otimes \cdots \otimes V_{\alpha_{\ell}}$. The following two subsections are a direct repetition of \ref{ssec:actionA} and \ref{ssec:actionD}, but transcribed into the notation used throughout this section. We include them for the sake of clarity.

\subsection{Action of $A^{(2)}(x)$ on $\psivec$}

Using the relation (\ref{2ab}) and the commutativity (\ref{2bb}) of the $B$-operators, we obtain the formula
\begin{multline}
A^{(2)}(x)
B^{(2)}(\mu_1)
\dots
B^{(2)}(\mu_m)
|0\rangle \otimes \avac
=
\left[
\prod_{i=1}^{m}
f(\mu_i,x)
B^{(2)}(\mu_i)
\right]
A^{(2)}(x)
\vac \otimes \avac
\label{2abbb}
\\
-\sum_{i=1}^{m}
g(\mu_i,x)
B^{(2)}(x)
\left[
\prod_{\substack{j\not=i \\ j=1}}^{m}
f(\mu_j,\mu_i)
B^{(2)}(\mu_j)
\right]
A^{(2)}(\mu_i)
\vac \otimes \avac .
\end{multline}

\subsection{Action of $D^{(2)}(x)$ on $\psivec$}

Using the relation (\ref{2db}) and the commutativity (\ref{2bb}) of the $B$-operators, we obtain the formula
\begin{multline}
D^{(2)}(x)
B^{(2)}(\mu_1)
\dots
B^{(2)}(\mu_m)
|0\rangle \otimes \avac
=
\left[
\prod_{i=1}^{m}
f(x,\mu_i)
B^{(2)}(\mu_i)
\right]
D^{(2)}(x)
\vac \otimes \avac
\label{2dbbb}
\\
-\sum_{i=1}^{m}
g(x,\mu_i)
B^{(2)}(x)
\left[
\prod_{\substack{j\not=i \\ j=1}}^{m}
f(\mu_i,\mu_j)
B^{(2)}(\mu_j)
\right]
D^{(2)}(\mu_i)
\vac \otimes \avac .
\end{multline}

\subsection{Second eigenvalue $\Lambda^{(2)}(x)$}

Using the definitions from above, it is straightforward to calculate
\begin{align}
A^{(2)}(x) |0\rangle \otimes \avac
=
a_2(x) 
|0\rangle \otimes \avac,
\quad
D^{(2)}(x) |0\rangle \otimes \avac
=
a_3(x)
\prod_{k=1}^{\ell}
(1/f(x,\lambda_k)) 
|0\rangle \otimes \avac .
\end{align}
Using these equations and adding (\ref{2abbb}) to (\ref{2dbbb}), we find that
\begin{align}
\mathcal{T}^{(2)}(x)
\psivec
=
\Lambda^{(2)}(x)
\psivec
=
\left(
a_2(x)
\prod_{i=1}^{m}
f(\mu_i,x)
+
a_3(x)
\prod_{k=1}^{\ell}
(1/f(x,\lambda_k))
\prod_{i=1}^{m}
f(x,\mu_i)
\right)
\psivec
\end{align}
if we assume that all terms not proportional to $\psivec$ cancel. This assumption, in turn, gives rise to the second set of Bethe equations, as we discuss in the next subsection. Hence we obtain the explicit form of the eigenvalue $\Lambda^{(2)}(x)$,
\begin{align}
\Lambda^{(2)}(x)
=
a_2(x)
\prod_{i=1}^{m}
f(\mu_i,x)
+
a_3(x)
\prod_{k=1}^{\ell}
(1/f(x,\lambda_k))
\prod_{i=1}^{m}
f(x,\mu_i).
\label{2eigval}
\end{align}

\subsection{Second set of Bethe equations}

As we just mentioned, $\psivec$ is an eigenstate of 
$\mathcal{T}^{(2)}(x)$ if and only if the unwanted terms in (\ref{2abbb}) and (\ref{2dbbb}) sum to zero. We find that this is the case when the variables $\{\mu_1,\dots,\mu_m\}$ satisfy the Bethe equations
\begin{align}
a_2(\mu_i)
\prod_{j\not= i}^{m}
f(\mu_j,\mu_i)
-
a_3(\mu_i)
\prod_{k=1}^{\ell}
(1/f(\mu_i,\lambda_k))
\prod_{j\not= i}^{m}
f(\mu_i,\mu_j)
=
0.
\end{align}
Rearranging slightly, we put the second set of Bethe equations in the form
\begin{align}
r_2(\mu_i)
\equiv
\frac{a_2(\mu_i)}{a_3(\mu_i)}
=
-\prod_{j=1}^{m}
\frac{\mu_i-\mu_j+1}{\mu_i-\mu_j-1}
\prod_{k=1}^{\ell}
\frac{1}{f(\mu_i,\lambda_k)}.
\label{2bethe-sl3}
\end{align}

\subsection{Summary}

We are now able to provide explicit formulae for the eigenvectors and eigenvalues of equation (\ref{1trans}), which was our original goal. Combining the expression (\ref{1eigvec}) with (\ref{2eigvec}) we recover the final form of the Bethe eigenvectors,
\begin{align}
|\Psi^{(1)}\rangle
=
B^{(1)}_{\alpha_1}(\lambda_1)
\ldots
B^{(1)}_{\alpha_{\ell}}(\lambda_{\ell})
B^{(2)}(\mu_1)
\ldots
B^{(2)}(\mu_m)
|0\rangle \otimes \avac.
\label{bethe-state-sl3}
\end{align}
Similarly, combining equation (\ref{1eigval}) with (\ref{2eigval}) we find that the eigenvalues are given by
\begin{align}
\Lambda^{(1)}(x)
=
a_1(x)
\prod_{i=1}^{\ell} f(\lambda_i,x)
+
a_2(x)
\prod_{i=1}^{m} f(\mu_i,x)
\prod_{j=1}^{\ell} f(x,\lambda_j)
+
a_3(x)
\prod_{i=1}^{m}
f(x,\mu_i).
\end{align}
Since $\Lambda^{(2)}(\lambda_i) = a_2(\lambda_i) \prod_{j=1}^{m} f(\mu_j,\lambda_i)$, the first set of Bethe equations becomes
\begin{align}
r_1(\lambda_i)
\equiv
\frac{a_1(\lambda_i)}{a_2(\lambda_i)}
=
-\prod_{j=1}^{\ell}
\frac{\lambda_i-\lambda_j+1}{\lambda_i-\lambda_j-1}
\prod_{k=1}^{m}
f(\mu_k,\lambda_i).
\label{1bethe-sl3}
\end{align}

\subsection{Dual Bethe eigenvectors}

Naturally, it is also possible to construct states $\langle \Psi^{(1)}| \in \mathcal{H}^{*}$ which are eigenvectors of the transfer matrix $\mathcal{T}^{(1)}(x)$. Repeating the steps from above with slight modification, one obtains
\begin{align}
\langle \Psi^{(1)}|
=
\langle \Uparrow_{\alpha}|
\otimes
\langle 0|
C^{(2)}(\mu_1)
\dots
C^{(2)}(\mu_m)
C^{(1)}_{\alpha_1}(\lambda_1)
\dots
C^{(1)}_{\alpha_{\ell}}(\lambda_{\ell})
\label{dual-bethe-sl3}
\end{align}
where each $C^{(1)}_{\alpha_i}(\lambda_i)$ is a column vector given by (\ref{decomposition}), and $C^{(2)}(x)$ denotes component $(2,1)$ of the monodromy matrix
\begin{align}
T^{(2)}_{\beta}(\lambda_{\ell},\dots,\lambda_1|x)
=
\mathbb{R}^{(2)}_{\beta\alpha_{\ell}}(x,\lambda_{\ell})
\dots
\mathbb{R}^{(2)}_{\beta\alpha_1}(x,\lambda_1)
D^{(1)}_{\beta}(x)
=
\left(
\begin{array}{cc}
A^{(2)}(\lambda_{\ell},\dots,\lambda_1|x)
&
B^{(2)}(\lambda_{\ell},\dots,\lambda_1|x)
\\
C^{(2)}(\lambda_{\ell},\dots,\lambda_1|x)
&
D^{(2)}(\lambda_{\ell},\dots,\lambda_1|x)
\end{array}
\right)_{\beta} .
\end{align}
The vacuum state in (\ref{dual-bethe-sl3}) is the tensor product of 
\begin{align}
\langle \Uparrow_{\alpha}| 
= 
\bigotimes_{i=1}^{\ell}
\left(
\begin{array}{cc} 
1 & 0
\end{array}
\right)_{\alpha_i}
\end{align} 
and the dual pseudo-vacuum $\langle 0|$. In this situation, the Bethe equations and the expression for the eigenvalue $\Lambda^{(1)}(x)$ are the same as those given above.

\section{Generic $SU(2)$ scalar products}

\subsection{Notation}

In the case of sets $\{\lambda\} = \{\lambda_1,\dots,\lambda_{\ell}\}$ and $\{\mu\} = \{\mu_1,\dots,\mu_m\}$, we define
\begin{align}
f(\{\lambda\},\{\mu\})
=
\prod_{i=1}^{\ell}
\prod_{j=1}^{m}
f(\lambda_i,\mu_j)
\end{align}
to make our subsequent equations more compact.

\subsection{Definition of $SU(2)$ scalar product}

In the case of $SU(2)$ models, the scalar product is defined as
\begin{align}
\langle \{\lambda^C\} | \{\lambda^B\} \rangle
=
\langle 0|
\prod_{i=1}^{\ell}
C(\lambda^C_i)
\prod_{j=1}^{\ell}
B(\lambda^B_j)
|0\rangle,
\quad\quad
\langle\langle \{\lambda^C\} | \{\lambda^B\} \rangle\rangle
=
\frac{
\langle \{\lambda^C\} | \{\lambda^B\} \rangle
}
{
\prod_{i=1}^{\ell}
d(\lambda^C_i)
\prod_{j=1}^{\ell}
d(\lambda^B_j)
} .
\label{sp-def}
\end{align}
The scalar product is equal to the action of a generic dual Bethe vector (\ref{dual-bethe}) on another generic Bethe vector (\ref{bethe-state}). At this stage, no restriction is imposed on the variables $\{\lc\},\{\lb\}$. The quantity on the right of (\ref{sp-def}), obtained by dividing by a product of vacuum eigenfunctions, is simply a convenient renormalization.  

\subsection{Sum formula for generic $SU(2)$ scalar product}

Using the commutation relations between the monodromy matrix operators, as well as the  action of these operators on the vacuum states, it is possible to derive a sum formula for the scalar product \cite{kor,ik}. We give the full details of this calculation in Appendix {\bf\ref{app:1}}. Here we only quote the formula, which reads
\begin{align}
\langle \{\lambda^C\} | \{\lambda^B\} \rangle
=
\sum
\prod_{\lbi}
a(\lbi)
\prod_{\lcii}
a(\lcii)
\prod_{\lbii}
d(\lbii)
\prod_{\lci}
d(\lci)
f(\lci,\lcii) f(\lbii,\lbi)
Z(\lbii | \lcii) Z(\lci | \lbi)
\label{sum-sl2}
\end{align}
where the sum is taken over all partitions of the sets $\{\lc\},\{\lb\}$ into two disjoint subsets
\begin{align}
\{\lambda^C\} = \{\lci\} \cup \{\lcii\}, \quad
\{\lambda^B\} = \{\lbi\} \cup \{\lbii\}, \quad
\text{such that}\ 
|\lbi| = |\lci|, \quad 
|\lbii| = |\lcii|
\label{part-rule}
\end{align}
and $Z(\lbii|\lcii),Z(\lci|\lbi)$ denote domain wall partition functions of the six-vertex model. It is common to normalize the scalar product by dividing by $\prod_{i=1}^{\ell} d(\lc_i) d(\lb_i)$, which gives
\begin{align}
\langle\langle \{\lambda^C\} | \{\lambda^B\} \rangle\rangle
=
\sum
\prod_{\lbi}
r(\lbi)
\prod_{\lcii} 
r(\lcii)
f(\lci,\lcii) f(\lbii,\lbi)
Z(\lbii | \lcii) Z(\lci | \lbi) .
\label{sl2-g}
\end{align}

\subsection{Domain wall partition function $Z(\{\lambda\}|\{w\})$}
\label{ssec:dom}

The domain wall partition function depends on two sets of variables $\{\lambda\}_{\ell} = \{\lambda_1,\dots,\lambda_{\ell}\}$ and $\{w\}_{\ell} = \{w_1,\dots,w_{\ell}\}$, and we denote it by $Z(\{\lambda\}|\{w\})$. We define it as the partition function of the lattices shown in Figure {\bf \ref{dwpf}}.

\begin{figure}[H]

\begin{center}
\begin{minipage}{4.3in}

\setlength{\unitlength}{0.00035cm}
\begin{picture}(40000,13500)(13000,8500)



\path(14000,20000)(26000,20000) \put(12000,20000){\tiny $\lambda_1$} 
\put(14000,20500){\tiny 1}  
\put(25500,20500){\tiny 2}

\path(14000,18000)(26000,18000)
\put(14000,18500){\tiny 1}  
\put(25500,18500){\tiny 2}

\path(14000,16000)(26000,16000)
\put(14000,16500){\tiny 1}  
\put(25500,16500){\tiny 2}

\path(14000,14000)(26000,14000)
\put(14000,14500){\tiny 1}  
\put(25500,14500){\tiny 2}

\path(14000,12000)(26000,12000) \put(12000,12000){\tiny $\lambda_{\ell}$}
\put(14000,12500){\tiny 1} 
\put(25500,12500){\tiny 2}


\path(16000,10000)(16000,22000) \put(15500,8000){\tiny $w_1$}
\put(16000,9000){\tiny 2} \put(15500,22500){\tiny 1}

\path(18000,10000)(18000,22000)
\put(18000,9000){\tiny 2} \put(17500,22500){\tiny 1}

\path(20000,10000)(20000,22000)
\put(20000,9000){\tiny 2} \put(19500,22500){\tiny 1}

\path(22000,10000)(22000,22000)
\put(22000,9000){\tiny 2} \put(21500,22500){\tiny 1}

\path(24000,10000)(24000,22000) \put(23500,8000){\tiny $w_{\ell}$}
\put(24000,9000){\tiny 2} \put(23500,22500){\tiny 1}



\path(34000,20000)(46000,20000) \put(32000,20000){\tiny $w_1$} 
\put(34000,20500){\tiny 1}  
\put(45500,20500){\tiny 2}

\path(34000,18000)(46000,18000)
\put(34000,18500){\tiny 1}  
\put(45500,18500){\tiny 2}

\path(34000,16000)(46000,16000)
\put(34000,16500){\tiny 1}  
\put(45500,16500){\tiny 2}

\path(34000,14000)(46000,14000)
\put(34000,14500){\tiny 1}  
\put(45500,14500){\tiny 2}

\path(34000,12000)(46000,12000) \put(32000,12000){\tiny $w_{\ell}$}
\put(34000,12500){\tiny 1} 
\put(45500,12500){\tiny 2}


\path(36000,10000)(36000,22000) \put(35500,8000){\tiny $\lambda_1$}
\put(36000,9000){\tiny 1} \put(35500,22500){\tiny 2}

\path(38000,10000)(38000,22000)
\put(38000,9000){\tiny 1} \put(37500,22500){\tiny 2}

\path(40000,10000)(40000,22000)
\put(40000,9000){\tiny 1} \put(39500,22500){\tiny 2}

\path(42000,10000)(42000,22000)
\put(42000,9000){\tiny 1} \put(41500,22500){\tiny 2}

\path(44000,10000)(44000,22000) \put(43500,8000){\tiny $\lambda_{\ell}$}
\put(44000,9000){\tiny 1} \put(43500,22500){\tiny 2}


\put(36000,20000){\circle*{200}}
\put(38000,20000){\circle*{200}}
\put(40000,20000){\circle*{200}}
\put(42000,20000){\circle*{200}}
\put(44000,20000){\circle*{200}}
\put(36000,18000){\circle*{200}}
\put(38000,18000){\circle*{200}}
\put(40000,18000){\circle*{200}}
\put(42000,18000){\circle*{200}}
\put(44000,18000){\circle*{200}}
\put(36000,16000){\circle*{200}}
\put(38000,16000){\circle*{200}}
\put(40000,16000){\circle*{200}}
\put(42000,16000){\circle*{200}}
\put(44000,16000){\circle*{200}}
\put(36000,14000){\circle*{200}}
\put(38000,14000){\circle*{200}}
\put(40000,14000){\circle*{200}}
\put(42000,14000){\circle*{200}}
\put(44000,14000){\circle*{200}}
\put(36000,12000){\circle*{200}}
\put(38000,12000){\circle*{200}}
\put(40000,12000){\circle*{200}}
\put(42000,12000){\circle*{200}}
\put(44000,12000){\circle*{200}}

\end{picture}

\end{minipage}
\end{center}

\caption{Two equivalent lattice representations of $Z(\{\lambda\} | \{w\})$. Each intersection of a horizontal and vertical line is a vertex, as defined in Figure {\bf\ref{vert1}}, or Figure {\bf\ref{vert2}}. The state variables on all external segments are fixed to the values shown, while all internal segments are summed over. For an explanation of how these boundary conditions arise naturally in (\ref{sum-sl2}), refer to Figure {\bf\ref{fig-sl2-sp}} in Appendix {\bf\ref{app:1}}.}
\label{dwpf}

\end{figure}
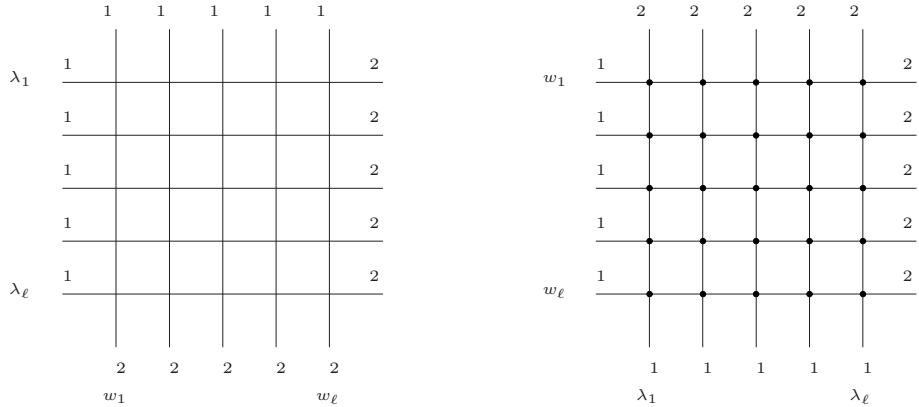
Generally, we shall use the definition of $Z(\{\lambda\}|\{w\})$ on the left of Figure {\bf\ref{dwpf}}, but it will be necessary later (in Subsection {\bf\ref{pf-resh}}) to refer to the definition on the right. The equivalence of these two partition functions is an easy consequence of the definition of the dotted vertex (Figure {\bf\ref{vert2}}). Below we summarize several powerful properties of $Z(\{\lambda\}|\{w\})$.

\begin{myTheorem}
The domain wall partition function satisfies the following conditions: 
\begin{enumerate}
\item[{\bf 1.}] In the case of a single variable, $Z(\lambda_1 | w_1) = g(\lambda_1, w_1)$.

\item[{\bf 2.}] $Z(\{\lambda\} | \{w\})$ is symmetric in the sets of variables $\{\lambda\}$ and $\{w\}$, separately.

\item[{\bf 3.}] $Z(\{\lambda\} | \{w\}) \rightarrow 0$ as any 
$\lambda_i \rightarrow \infty$ or $w_j \rightarrow \infty$.

\item[{\bf 4.}] $Z(\{\lambda\} | \{w\})$ is holomorphic in its variables, apart from simple poles at the points $\lambda_i = w_j$. The residues of these poles are given by
\begin{align}
\lim_{\lambda_i \rightarrow w_j}
\Big(
(\lambda_i - w_j) Z(\{\lambda\} | \{w\})
\Big)
=
f(\lambda_i,\{\widehat{w}_j\})
f(\{\widehat{\lambda}_i\},w_j)
Z(\{\widehat{\lambda}_i\} | \{\widehat{w}_j\})
\label{res-prop}
\end{align}
where $\{\widehat{\lambda}_i\}$ and $\{\widehat{w}_j\}$ denote the sets $\{\lambda\}$ and $\{w\}$, with the omission of $\lambda_i$ and $w_j$, respectively.

%
%
%
\end{enumerate}
\end{myTheorem}

\begin{proof}
These properties are due to Korepin \cite{kor}. We only sketch the proof, since they are well known in the literature. 
{\bf 1.} The domain wall partition function in a single variable is just a single vertex, and the boundary conditions fix this vertex to have the weight $g(\lambda_1,w_1)$.
{\bf 2.} Using the Yang-Baxter equation in graphical form, it is possible to exchange the order of any pair of horizontal or vertical lines in the lattice. This establishes the required symmetry.
{\bf 3.} Consider the bottom row of the lattice on the left of Figure {\bf\ref{dwpf}}. Every configuration in this partition function contains one vertex with weight $g(\lambda_{\ell},w_j)$, and all remaining vertices in the bottom row being of weight 1 or $f(\lambda_{\ell},w_k)$, $k\not=j$. From the form of the Boltzmann weights (\ref{rat-wt}), it is therefore clear that $Z(\{\lambda\} | \{w\}) \rightarrow 0$ as $\lambda_{\ell} \rightarrow \infty$. A similar argument applies to the right-most column, and the variable $w_{\ell}$. Symmetry in $\{\lambda\}$ and $\{w\}$ establishes the result.

{\bf 4.} From the Boltzmann weights (\ref{rat-wt}), it is clear that $Z(\{\lambda\}|\{w\})$ is holomorphic apart from poles at $\lambda_i = w_j$. Consider the top-right vertex of the lattice on the left of Figure {\bf\ref{dwpf}}. There exist configurations where it has weight 1, or weight $g(\lambda_1,w_{\ell})$. In the latter case, the rest of the top row and right-most column are frozen to be vertices of the type $f(\lambda_1,w_j)$ and $f(\lambda_i,w_{\ell})$, respectively. Since only these configurations contribute to the residue of the pole at $\lambda_1 = w_{\ell}$, we find that
\begin{align}
\lim_{\lambda_1 \rightarrow w_{\ell}}
\Big(
(\lambda_1 - w_{\ell}) Z(\{\lambda\} | \{w\})
\Big)
=
\prod_{j=1}^{\ell-1}
f(\lambda_1,w_j)
\prod_{i=2}^{\ell}
f(\lambda_i,w_{\ell})
Z(\{\widehat{\lambda}_1\} | \{\widehat{w}_{\ell}\}) .
\end{align}
Symmetry in $\{\lambda\}$ and $\{w\}$ gives the result (\ref{res-prop}). It is a standard argument to prove that properties {\bf 1}--{\bf 4} determine $Z(\{\lambda\}|\{w\})$ uniquely. Supposing that another function $Z'(\{\lambda\}|\{w\})$ obeys the same set of properties, $Z-Z'$ is holomorphic and bounded everywhere in the complex plane, and therefore constant. This constant must be zero, from property {\bf 3}, proving that $Z'=Z$.  


\end{proof}

\subsection{Determinant formulae for domain wall partition function}

The domain wall partition function is given by the Izergin determinant formula \cite{ize}
\begin{align}
Z(\{\lambda\} | \{w\})
=
\frac{
\prod_{i,j=1}^{\ell}
(\lambda_i - w_j + 1)
}
{
\prod_{1 \leq i < j \leq \ell}
(\lambda_j - \lambda_i)
(w_i - w_j)
}
\det\left(
\frac{1}{(\lambda_i - w_j + 1)(\lambda_i - w_j)}
\right)_{1 \leq i,j \leq \ell} .
\label{ize-det}
\end{align}
Recently, another determinant representation appeared for the domain wall partition function, due to Kostov \cite{kos1,kos2,fw}. This determinant formula is given by 
\begin{align}
Z(\{\lambda\} | \{w\})
=
\frac{1}
{
\prod_{1 \leq i < j \leq \ell} (\lambda_j - \lambda_i)
}
\det\left(
\lambda_i^{j-1}
\prod_{k=1}^{\ell}
\frac{(\lambda_i - w_k +1)}{(\lambda_i - w_k)}
-
(\lambda_i+1)^{j-1}
\right)_{1\leq i ,j \leq \ell} .
\label{kost1}
\end{align}
To prove either determinant representation, (\ref{ize-det}) or (\ref{kost1}), it is simply a matter of showing that they satisfy the properties {\bf 1}--{\bf 4} of Subsection {\bf\ref{ssec:dom}}.

\subsection{Partial domain wall partition function}
\label{ssec:partial}

Let $n$ be an integer satisfying $1 \leq n \leq \ell $. Consider the partition function generated by deleting the bottom $(\ell-n)$ rows from the lattice in Figure {\bf \ref{dwpf}}, and whose lower boundary is summed over all state variables. We denote this object $Z(\{\lambda\}_n | \{w\}_{\ell})$ and represent it by the lattice in Figure {\bf \ref{dwpf-res}}.

\begin{figure}[H]

\begin{center}
\begin{minipage}{4.3in}

\setlength{\unitlength}{0.00035cm}
\begin{picture}(40000,7500)(4500,14000)


\path(14000,20000)(26000,20000) \put(12000,20000){\tiny $\lambda_1$} 
\put(14000,20500){\tiny 1}  
\put(25500,20500){\tiny 2}

\path(14000,18000)(26000,18000)
\put(14000,18500){\tiny 1}  
\put(25500,18500){\tiny 2}

\path(14000,16000)(26000,16000) \put(12000,16000){\tiny $\lambda_{n}$}
\put(14000,16500){\tiny 1} 
\put(25500,16500){\tiny 2}


\path(16000,14000)(16000,22000) \put(15500,13000){\tiny $w_1$}
\put(15500,22500){\tiny 1}

\path(18000,14000)(18000,22000)
\put(17500,22500){\tiny 1}

\path(20000,14000)(20000,22000)
\put(19500,22500){\tiny 1}

\path(22000,14000)(22000,22000)
\put(21500,22500){\tiny 1}

\path(24000,14000)(24000,22000) \put(23500,13000){\tiny $w_{\ell}$}
\put(23500,22500){\tiny 1}

\end{picture}

\end{minipage}
\end{center}

\caption{Lattice representation of $Z(\{\lambda\}_n | \{w\}_{\ell})$. The number of horizontal rapidities $\lambda_i$ is less than the number of vertical rapidities $w_j$. The lower boundary segments are devoid of state variables to indicate summation at these points.}

\label{dwpf-res}
\end{figure}
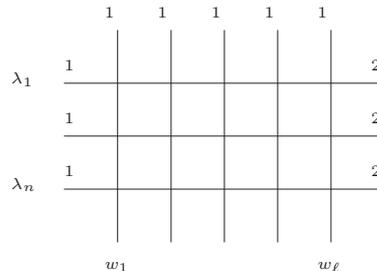
We emphasize that, unlike in the domain wall partition function, the lower boundary segments in Figure {\bf \ref{dwpf-res}} are {\it not}\/ fixed to definite state variable values but summed over them all. 

In \cite{fw} it was explained that (up to a combinatoric factor) $Z(\{\lambda\}_n | \{w\}_{\ell})$ is the leading term in the domain wall partition function $Z(\{\lambda\}_{\ell} | \{w\}_{\ell})$ as $\lambda_{\ell},\dots,\lambda_{n+1} \rightarrow \infty$. In this limit the contribution from the bottom $(\ell-n)$ rows of Figure {\bf\ref{dwpf}} is the same for all state variable configurations, and up to the factor $(\ell-n)!$ we are left with the lattice shown in Figure {\bf\ref{dwpf-res}}. More formally, 
\begin{align}
Z(\{\lambda\}_n | \{w\}_\ell)
=
\frac{1}{(\ell-n)!}
\lim_{\lambda_{\ell},\dots,\lambda_{n+1} \rightarrow \infty}
\Big(
\lambda_{\ell} \cdots \lambda_{n+1}
Z(\{\lambda\}_{\ell} | \{w\}_{\ell})
\Big)
\label{many-limits}
\end{align}
where the limits should be taken sequentially, in any order. Performing the limits (\ref{many-limits}) on the determinant (\ref{ize-det}), one obtains
\begin{align}
\label{ize-det2}
Z( \{\lambda\}_n | \{w\}_{\ell} )
=
\frac{\displaystyle{
\prod_{i=1}^{n} 
\prod_{j=1}^{\ell}
(\lambda_i - w_j+1)
}}
{\displaystyle{
\prod_{1 \leq i<j \leq n}
(\lambda_j - \lambda_i)
\prod_{1 \leq i<j \leq \ell}
(w_i - w_j)
}}
\left|
\begin{array}{ccc}
\frac{1}{(\lambda_1 - w_1)(\lambda_1 - w_1 + 1)} 
& 
\cdots 
& 
\frac{1}{(\lambda_1 - w_{\ell})(\lambda_1 - w_{\ell} + 1)}
\\
\vdots 
& & 
\vdots
\\
\frac{1}{(\lambda_{n} - w_1)(\lambda_{n} - w_1 + 1)}
&
\cdots
&
\frac{1}{(\lambda_{n} - w_{\ell})(\lambda_{n} - w_{\ell} + 1)}
\\
\\
w_1^{\ell-n-1} & \cdots & w_{\ell}^{\ell-n-1}
\\
\vdots & & \vdots
\\
w_1^{0} & \cdots & w_{\ell}^{0}
\end{array}
\right| .
\end{align}
Alternatively, starting from (\ref{kost1}), one finds that
\begin{align}
Z(\{\lambda\}_n | \{w\}_{\ell})
=
\frac{1}{\prod_{1 \leq i < j \leq n} (\lambda_j - \lambda_i)}
\det\left(
\lambda_i^{j-1}
\prod_{k=1}^{\ell}
\frac{(\lambda_i - w_k +1)}{(\lambda_i - w_k)}
-
(\lambda_i+1)^{j-1}
\right)_{1\leq i, j \leq n} 
\label{kost2}
\end{align}
where the determinant is now $n\times n$. For more details on the derivation of (\ref{ize-det2}) and (\ref{kost2}), we refer the reader to \cite{fw}. In the case where all variables are sent to infinity, we obtain
\begin{align}
\lim_{\lambda_{\ell},\dots,\lambda_1 \rightarrow \infty}
\Big(
\lambda_{\ell} \dots \lambda_1
Z(\{\lambda\}_{\ell} | \{w\}_{\ell})
\Big)
=
\ell!,
\quad\quad
\lim_{w_{\ell},\dots,w_1 \rightarrow \infty}
\Big(
w_{\ell} \dots w_1
Z(\{\lambda\}_{\ell} | \{w\}_{\ell})
\Big)
=
(-)^{\ell} \ell!
\label{all-limits}
\end{align}
We will use these results frequently throughout the rest of the paper. 

\subsection{Imposing the Bethe equations on $\{\lb\}$}

From now on we consider the case when one set of variables in the scalar product, $\{\lb\}$, satisfies the Bethe equations (\ref{bethe-sl2}). That is, we will assume that 
\begin{align}
r(\lb_i)
=
\frac{a(\lb_i)}{d(\lb_i)}
=
-
\prod_{j=1}^{\ell}
\frac{\lb_i-\lb_j + 1}{\lb_i - \lb_j - 1},
\quad\quad
\forall\ 1 \leq i \leq \ell.
\label{bethe}
\end{align}

\subsection{Slavnov determinant formula}
\label{slav-det}

An important case of the scalar product was considered by Slavnov in \cite{sla}. Assuming that the Bethe equations (\ref{bethe}) apply, one is able to replace all instances of $r(\lambda^B_i)$ with a function purely in the variables $\{\lb\}$, leading to a determinant expression for the scalar product. The most direct way to prove this is to start from (\ref{sl2-g}) and enforce the equations (\ref{bethe}), which gives
\begin{align}
\langle\langle \{\lambda^C\} | \{\lambda^B\} \rangle\rangle
=
\sum
(-)^{|\lbi|}
\prod_{\lbi} 
\prod_{j=1}^{\ell} \left(\frac{\lbi-\lb_j+1}{\lbi-\lb_j-1}\right)
\prod_{\lcii} r(\lcii)
f(\lci,\lcii) 
f(\lbii,\lbi)
Z(\lbii | \lcii) Z(\lci | \lbi) .
\label{slav-sum}
\end{align}
The sum (\ref{slav-sum}) can then be evaluated in determinant form, using the determinant expression (\ref{ize-det}) for each domain wall partition function and the Laplace formula for the determinant of a sum of matrices. For details of this calculation we refer the reader to \cite{kkmst}, but here we only state the result 
\begin{align}
\langle\langle \{\lambda^C\} | \{\lambda^B\} \rangle\rangle
=
\label{betscal}
\frac{\displaystyle{
\det\left(
\frac{1}{\lambda^B_j-\lambda^C_i}
\left(
\prod_{k\not= j}^{\ell}(\lb_k-\lc_i+1)
r(\lc_i)
-
\prod_{k\not=j}^{\ell} (\lb_k-\lc_i-1)
\right)
\right)_{1\leq i,j \leq \ell}
}}
{\displaystyle{
\prod_{1\leq i < j \leq \ell}
(\lc_j-\lc_i)(\lb_i-\lb_j)
}} .
\end{align}
In the rest of the paper, we treat the equality of (\ref{slav-sum}) and (\ref{betscal}) as an identity between meromorphic functions in the variables $\{\lc\},\{\lb\}$, with each $r(\lc_i)$ playing the role of a constant. 

\subsection{Behaviour in the $\{\lambda^B\} \rightarrow \infty$ limit}
\label{ssec:beh}

We conclude the section by studying the behaviour of the function (\ref{slav-sum}) in the limit $\{\lambda^B\}\rightarrow \infty$. In doing so, we treat (\ref{slav-sum}) as a free function in the variables $\{\lb\}$, despite the fact that it arises by imposing the constraints (\ref{bethe}) on $\{\lb\}$. We begin by fixing the notation
\begin{align}
\langle\langle \{\lc\} | \{\infty\} \rangle\rangle 
\equiv
\frac{1}{\ell!}
\lim_{\lb_{\ell}, \dots, \lb_1 \rightarrow \infty}
\Big(
\lb_{\ell} \cdots \lb_1 
\langle\langle \{\lambda^C\} | \{\lambda^B\} \rangle\rangle
\Big) .
\end{align}
Applying (\ref{all-limits}) to each domain wall partition function in (\ref{slav-sum}) we easily take the limit, and find that 
\begin{align}
\langle\langle \{\lc\} | \{\infty\} \rangle\rangle 
&=
\frac{1}{\ell!}
\sum_{k=0}^{\ell}
\sum_{
|\lci|=\ell-k,
|\lcii|=k
}
\binom{\ell}{k}
\prod_{\lcii}
r(\lcii)
f(\lci,\lcii)
k! (-)^{\ell-k} (\ell - k)!
\label{sl2-b2}
\\
&=
\sum_{
\{\lc\}=\{\lci\} \cup \{\lcii\}
}
(-)^{|\lci|}
\prod_{\lcii} r(\lcii)
\prod_{\lci,\lcii}
\frac{(\lci-\lcii+1)}{(\lci-\lcii)} .
\nonumber
\end{align}
In fact the sum (\ref{sl2-b2}) can be evaluated as the determinant
\begin{align}
\langle\langle \{\lambda^C\} | \{\infty\} \rangle\rangle
=
\frac{\det\Big(
(\lambda^C_i)^{j-1}
r(\lambda^C_i)
-
(\lambda^C_i+1)^{j-1}
\Big)_{1 \leq i, j \leq \ell}
}
{\displaystyle{
\prod_{1 \leq i < j \leq \ell} (\lambda^C_j - \lambda^C_i)
}} .
\label{slav-lim}
\end{align}
To see this, one uses the Laplace formula for the determinant of a sum of two matrices to expand (\ref{slav-lim}), as well as the classic evaluation of the Vandermonde determinant. Alternatively, one can derive the expression (\ref{slav-lim}) starting directly from the Slavnov determinant (\ref{betscal}) and taking the required limits\footnote{This observation is due to I~Kostov, and is explained in greater detail in \cite{fw}.}. The determinant (\ref{slav-lim}) is closely related to the partial domain wall partition function of Subsection {\bf\ref{ssec:partial}}; for a length $L$ inhomogeneous XXX Heisenberg spin-1/2 chain one has $r(\lambda^C_i) = \prod_{j=1}^{L} (\lambda^C_i-w_j+1)/(\lambda^C_i-w_j)$, and the agreement between (\ref{kost2}) and (\ref{slav-lim}) is exact.

\section{Generic $SU(3)$ scalar products}
\label{sec:su3-sp}

\subsection{Definition of $SU(3)$ scalar product}

In the case of $SU(3)$ models, the scalar product is defined as 
\begin{align}
\label{sl3-sp-def}
&
\langle \{\mc\}, \{\lc\} | \{\lb\}, \{\mb\} \rangle
=
\prod_{i=1}^{m} \prod_{j=1}^{\ell}
f(\mu^C_i,\lambda^C_j)
f(\mu^B_i,\lambda^B_j)
\times
\\
&
\langle \Uparrow_{\alpha}|
\otimes
\langle 0|
\prod_{i=1}^{m}
C^{(2)}(\mc_i)
\prod_{j=1}^{\ell}
C^{(1)}_{\alpha_j}(\lc_j)
\prod_{i=1}^{\ell}
B^{(1)}_{\beta_i}(\lb_i)
\prod_{j=1}^{m}
B^{(2)}(\mb_j)
|0\rangle
\otimes
|\Uparrow_{\beta} \rangle,
\nonumber
\\
&
\langle\langle \{\mc\}, \{\lc\} | \{\lb\}, \{\mb\} \rangle\rangle
=
\frac{
\langle \{\mc\}, \{\lc\} | \{\lb\}, \{\mb\} \rangle
}
{
\prod_{i=1}^{m}
a_3(\mc_i)
\prod_{j=1}^{\ell}
a_2(\lc_j)
\prod_{i=1}^{\ell}
a_2(\lb_i)
\prod_{j=1}^{m}
a_3(\mb_j)
} .
\nonumber
\end{align}
Once again, the scalar product is the action of a generic dual Bethe vector (\ref{dual-bethe-sl3}) on another generic Bethe vector (\ref{bethe-state-sl3}), up to the normalization $\prod_{i=1}^{m} \prod_{j=1}^{\ell}
f(\mu^C_i,\lambda^C_j) f(\mu^B_i,\lambda^B_j)$ which we include for consistency with \cite{res}. The auxiliary spaces $V_{\alpha_i},V_{\beta_j}$ participating in the scalar product are taken to be different in each half. No assumptions have yet been made in regard to the variables $\{\mc\},\{\lc\},\{\lb\},\{\mb\}$. 

\subsection{Sum formula for generic $SU(3)$ scalar product}

Following the work of Reshetikhin \cite{res}, the generic $SU(3)$ scalar product is given by the sum formula
\begin{multline}
\label{sum-sl3}
\langle \{\mc\}, \{\lc\} | \{\lb\}, \{\mb\} \rangle
=
\sum
\prod_{\lbi} a_1(\lbi) \prod_{\lcii} a_1(\lcii)
\prod_{\lbii} a_2(\lbii) \prod_{\lci} a_2(\lci)
\\
\times
\prod_{\mbii} a_2(\mbii) \prod_{\mci} a_2(\mci)
\prod_{\mbi} a_3(\mbi) \prod_{\mcii} a_3(\mcii)
f(\lci,\lcii) f(\lbii,\lbi) 
f(\mcii,\mci) f(\mbi,\mbii)
\\
\times
f(\mbii,\lbii) f(\mci,\lci)
Z(\{\lbii\},\{\mci\} | \{\lcii\},\{\mbi\})
Z(\{\lci\},\{\mbii\} | \{\lbi\},\{\mcii\}) .
\end{multline}
The sum in (\ref{sum-sl3}) is taken over all partitions of the sets $\{\lambda^C\},\{\lambda^B\},\{\mu^C\},\{\mu^B\}$ into disjoint subsets
\begin{align}
&
\{\lambda^C\} = \{\lci\} \cup \{\lcii\}, \quad
\{\lambda^B\} = \{\lbi\} \cup \{\lbii\}, \quad
\text{such that}\ 
|\lbi| = |\lci|, \quad 
|\lbii| = |\lcii|
\label{disj1}
\\
&
\{\mu^C\} = \{\mci\} \cup \{\mcii\}, \quad
\{\mu^B\} = \{\mbi\} \cup \{\mbii\}, \quad
\text{such that}\ 
|\mbi| = |\mci|, \quad 
|\mbii| = |\mcii|
\label{disj2}
\end{align}
and the quantities $Z(\{\lbii\},\{\mci\} | \{\lcii\},\{\mbi\}),
Z(\{\lci\},\{\mbii\} | \{\lbi\},\{\mcii\})$ are partition functions defined in Subsection {\bf\ref{pf-resh}}. More details regarding the proof of (\ref{sum-sl3}) are given in Appendix {\bf\ref{app:2}}. 

Normalizing by dividing by $\prod_{i=1}^{\ell} a_2(\lc_i)a_2(\lb_i) \prod_{j=1}^{m} a_3(\mc_j)a_3(\mb_j)$, we have
\begin{multline}
\label{sl3-g}
\langle\langle \{\mc\}, \{\lc\} | \{\lb\}, \{\mb\} \rangle\rangle
=
\sum
\prod_{\lbi} r_1(\lbi) \prod_{\lcii} r_1(\lcii)
\prod_{\mbii} r_2(\mbii) \prod_{\mci} r_2(\mci)
\\
\times
f(\lci,\lcii) f(\lbii,\lbi) 
f(\mcii,\mci) f(\mbi,\mbii)
f(\mbii,\lbii) f(\mci,\lci)
\\
\times
Z(\{\lbii\},\{\mci\} | \{\lcii\},\{\mbi\})
Z(\{\lci\},\{\mbii\} | \{\lbi\},\{\mcii\}) .
\end{multline}

\subsection{Partition function 
$Z(\{\lambda\},\{\mu\}|\{w\},\{v\})$}
\label{pf-resh}

This quantity, defined graphically in \cite{res}, is equal to the lattice sum
\begin{figure}[H]

\begin{center}
\begin{minipage}{4.3in}

\setlength{\unitlength}{0.0003cm}
\begin{picture}(40000,20000)(8000,-6000)


\put(11000,12000){\tiny $\lambda_1$}
\put(13000,12000){\tiny 1}
\path(14000,12000)(36000,12000)
\put(36500,12000){\tiny 2} 

\put(13000,10000){\tiny 1}
\path(14000,10000)(36000,10000) 
\put(36500,10000){\tiny 2}

\put(13000,8000){\tiny 1}
\path(14000,8000)(36000,8000)
\put(36500,8000){\tiny 2} 

\put(13000,6000){\tiny 1}
\path(14000,6000)(36000,6000)
\put(36500,6000){\tiny 2} 

\put(11000,4000){\tiny $\lambda_{\ell}$}
\put(13000,4000){\tiny 1}
\path(14000,4000)(36000,4000) 
\put(36500,4000){\tiny 2}


\put(11000,2000){\tiny $\mu_1$}
\put(13000,2000){\tiny 3}
\path(14000,2000)(36000,2000)
\put(36500,2000){\tiny 2} 

\put(13000,000){\tiny 3}
\path(14000,000)(36000,000)
\put(36500,000){\tiny 2} 

\put(13000,-2000){\tiny 3}
\path(14000,-2000)(36000,-2000)
\put(36500,-2000){\tiny 2} 

\put(11000,-4000){\tiny $\mu_m$}
\put(13000,-4000){\tiny 3}
\path(14000,-4000)(36000,-4000)
\put(36500,-4000){\tiny 2}


\path(16000,-6000)(16000,14000)
\put(15500,-7000){\tiny $w_1$}
\put(16200,-6000){\tiny 2} 
\put(16000,14500){\tiny 1} 

\path(18000,-6000)(18000,14000)
\put(18200,-6000){\tiny 2} 
\put(18000,14500){\tiny 1} 

\path(20000,-6000)(20000,14000)
\put(20200,-6000){\tiny 2} 
\put(20000,14500){\tiny 1} 

\path(22000,-6000)(22000,14000)
\put(22200,-6000){\tiny 2} 
\put(22000,14500){\tiny 1} 

\path(24000,-6000)(24000,14000) 
\put(23500,-7000){\tiny $w_{\ell}$}
\put(24200,-6000){\tiny 2} 
\put(24000,14500){\tiny 1}


\path(28000,-6000)(28000,14000)
\put(27500,-7000){\tiny $v_1$}
\put(28200,-6000){\tiny 3} 
\put(28000,14500){\tiny 2} 

\path(30000,-6000)(30000,14000)
\put(30200,-6000){\tiny 3} 
\put(30000,14500){\tiny 2} 

\path(32000,-6000)(32000,14000)
\put(32200,-6000){\tiny 3} 
\put(32000,14500){\tiny 2}

\path(34000,-6000)(34000,14000)
\put(33500,-7000){\tiny $v_m$}
\put(34200,-6000){\tiny 3} 
\put(34000,14500){\tiny 2}


\put(28000,-4000){\circle*{200}}
\put(28000,-2000){\circle*{200}}
\put(28000,000){\circle*{200}}
\put(28000,2000){\circle*{200}}
\put(28000,4000){\circle*{200}}
\put(28000,6000){\circle*{200}}
\put(28000,8000){\circle*{200}}
\put(28000,10000){\circle*{200}}
\put(28000,12000){\circle*{200}}

\put(30000,-4000){\circle*{200}}
\put(30000,-2000){\circle*{200}}
\put(30000,000){\circle*{200}}
\put(30000,2000){\circle*{200}}
\put(30000,4000){\circle*{200}}
\put(30000,6000){\circle*{200}}
\put(30000,8000){\circle*{200}}
\put(30000,10000){\circle*{200}}
\put(30000,12000){\circle*{200}}

\put(32000,-4000){\circle*{200}}
\put(32000,-2000){\circle*{200}}
\put(32000,000){\circle*{200}}
\put(32000,2000){\circle*{200}}
\put(32000,4000){\circle*{200}}
\put(32000,6000){\circle*{200}}
\put(32000,8000){\circle*{200}}
\put(32000,10000){\circle*{200}}
\put(32000,12000){\circle*{200}}

\put(34000,-4000){\circle*{200}}
\put(34000,-2000){\circle*{200}}
\put(34000,000){\circle*{200}}
\put(34000,2000){\circle*{200}}
\put(34000,4000){\circle*{200}}
\put(34000,6000){\circle*{200}}
\put(34000,8000){\circle*{200}}
\put(34000,10000){\circle*{200}}
\put(34000,12000){\circle*{200}}

\end{picture}

\end{minipage}
\end{center}

\caption{Lattice representation of $Z(\{\lambda\},\{\mu\}|\{w\},\{v\})$. Undotted vertices denote the entries of the $R$-matrix (\ref{Rmat-sl3}), dotted vertices denote the entries of (\ref{Rmat-sl3t}). Each lattice line denotes an entry of a certain $SU(3)$ monodromy matrix. The top $\ell$ rows denote $t_{12}(\lambda_i)$ operators, and the bottom $m$ rows denote the $t_{32}(\mu_j)$ operators.}
\label{resh-pf}

\end{figure}
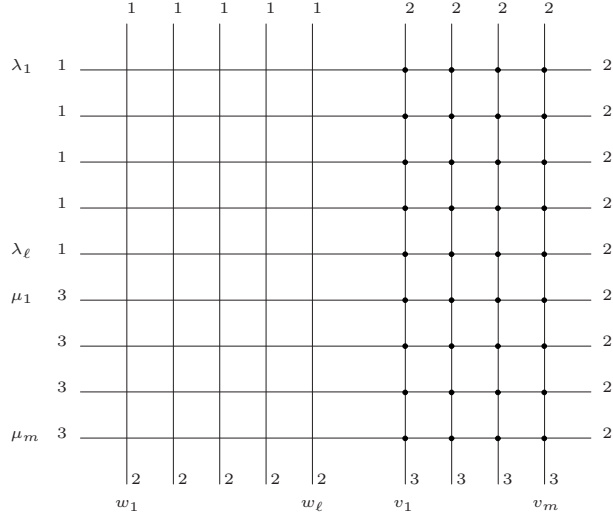
Notice that both types of $SU(3)$ $R$-matrix, namely (\ref{Rmat-sl3}) and (\ref{Rmat-sl3t}), are present in this function. The construction of $Z(\{\lambda\},\{\mu\}|\{w\},\{v\})$ and the reason for its appearance in (\ref{sl3-g}) is a rather complicated story, which we explain in more detail in Appendix {\bf\ref{app:2}}. For our purposes it plays the role of the domain wall partition function at $SU(3)$ level.
\begin{myTheorem}
The partition function in Figure \ref{resh-pf} is given by\footnote{A number of alternative expressions for $Z(\{\lambda\},\{\mu\}|\{w\},\{v\})$, of an analogous form to (\ref{resh-pf2}), appeared subsequently to the first version of this paper in \cite{bprs1}.} 
\begin{multline}
\label{resh-pf2}
Z(\{\lambda\},\{\mu\}|\{w\},\{v\})
=
\sum_{\substack{
\{\lambda\}=\{\li\}\cup\{\lii\}
\\
\{\mu\}=\{\mi\}\cup\{\mii\}
}}
\prod_{\mi,\mii}
f(\mi,\mii)
\prod_{\li,\lii}
f(\lii,\li)
\prod_{\mi,\li}
f(\mi,\li)
Z(\{\lii\}|\{\mii\})
\\
\times
Z(\{\li\}\cup\{\mii\}|\{w\})
Z(\{v\}|\{\mi\}\cup\{\lii\})
\end{multline}
where the sum in (\ref{resh-pf2}) is over all partitions of $\{\lambda\}$ and $\{\mu\}$ into disjoint subsets, such that $|\lii|=|\mii|$.
\end{myTheorem}

\begin{proof}
The proof is of a similar nature to the proof of equation (\ref{sum-sl2}), which is thoroughly outlined in Appendix {\bf\ref{app:1}}. Consider a monodromy matrix formed by taking a product of the $R$-matrices (\ref{Rmat-sl3}) and (\ref{Rmat-sl3t}),
\begin{align}
T_{\alpha}(x)
\equiv
T_{\alpha}(x|w_1,\dots,w_{\ell},v_1,\dots,v_m)
&=
R^{(1)}_{\alpha 1}(x,w_1) \dots 
R^{(1)}_{\alpha \ell}(x,w_{\ell})
R^{*(1)}_{\alpha 1'}(x,v_1) \dots 
R^{*(1)}_{\alpha m'}(x,v_m)
\label{resh-mon}
\\
&=
\left(
\begin{array}{ccc}
t_{11}(x) & t_{12}(x) & t_{13}(x)
\\
t_{21}(x) & t_{22}(x) & t_{23}(x)
\\
t_{31}(x) & t_{32}(x) & t_{33}(x)
\end{array}
\right)_{\alpha} .
\nonumber
\end{align}
Because of the Yang-Baxter equations (\ref{yb1}) and (\ref{yb2}), the monodromy matrix (\ref{resh-mon}) obeys the intertwining equation $R^{(1)}_{\alpha\beta}(x,y) T_{\alpha}(x) T_{\beta}(y)
= T_{\beta}(y) T_{\alpha}(x) R^{(1)}_{\alpha\beta}(x,y)$. From this equation we can extract one particular identity between the operator entries of (\ref{resh-mon}), namely
\begin{align}
t_{32}(x) t_{12}(y)
=
f(x,y) t_{12}(y) t_{32}(x)
-
g(x,y) t_{12}(x) t_{32}(y).
\label{com-rel}
\end{align}
Noticing that each horizontal line in Figure {\bf\ref{resh-pf}} is the graphical representation of an operator $t_{12}(\lambda_i)$ or $t_{32}(\mu_j)$, we can use the commutation relation (\ref{com-rel}) repeatedly to exchange the lattice lines. The aim is to transfer the $t_{32}$ lines to the top, and the $t_{12}$ lines to the bottom, as shown in Figure {\bf \ref{resh-pf3}}.

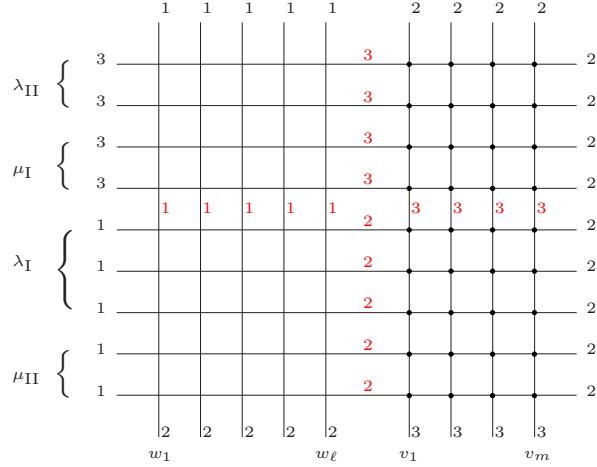
\begin{figure}[H]

\begin{center}
\begin{minipage}{4.3in}

\setlength{\unitlength}{0.00027cm}
\begin{picture}(40000,20000)(4000,-6000)


\put(9000,10750){\tiny $\lii$}
\put(11000,10750){$\left\{\phantom{\frac{\prod}{\prod}}\right.$}

\put(9000,6750){\tiny $\mi$}
\put(11000,6750){$\left\{\phantom{\frac{\prod}{\prod}}\right.$}

\put(9000,2250){\tiny $\li$}
\put(11000,1750){$\left\{\phantom{\displaystyle{\prod_{i}^{j}}}\right.$}

\put(9000,-3200){\tiny $\mii$}
\put(11000,-3250){$\left\{\phantom{\frac{\prod}{\prod}}\right.$}


\put(13000,12000){\tiny 3}
\path(14000,12000)(36000,12000)
\put(25800,12200){\tiny\color{red} 3}
\put(36500,12000){\tiny 2} 

\put(13000,10000){\tiny 3}
\path(14000,10000)(36000,10000)
\put(25800,10200){\tiny\color{red} 3} 
\put(36500,10000){\tiny 2}

\put(13000,8000){\tiny 3}
\path(14000,8000)(36000,8000)
\put(25800,8200){\tiny\color{red} 3}
\put(36500,8000){\tiny 2} 

\put(13000,6000){\tiny 3}
\path(14000,6000)(36000,6000)
\put(25800,6200){\tiny\color{red} 3}
\put(36500,6000){\tiny 2} 


\put(13000,4000){\tiny 1}
\path(14000,4000)(36000,4000)
\put(25800,4200){\tiny\color{red} 2} 
\put(36500,4000){\tiny 2}

\put(13000,2000){\tiny 1}
\path(14000,2000)(36000,2000)
\put(25800,2200){\tiny\color{red} 2}
\put(36500,2000){\tiny 2} 

\put(13000,000){\tiny 1}
\path(14000,000)(36000,000)
\put(25800,200){\tiny\color{red} 2}
\put(36500,000){\tiny 2} 

\put(13000,-2000){\tiny 1}
\path(14000,-2000)(36000,-2000)
\put(25800,-1800){\tiny\color{red} 2}
\put(36500,-2000){\tiny 2} 

\put(13000,-4000){\tiny 1}
\path(14000,-4000)(36000,-4000)
\put(25800,-3800){\tiny\color{red} 2}
\put(36500,-4000){\tiny 2}


\path(16000,-6000)(16000,14000)
\put(15500,-7000){\tiny $w_1$}
\put(16100,-6000){\tiny 2}
\put(16100,4800){\tiny\color{red} 1} 
\put(16100,14500){\tiny 1} 

\path(18000,-6000)(18000,14000)
\put(18100,-6000){\tiny 2}
\put(18100,4800){\tiny\color{red} 1}  
\put(18100,14500){\tiny 1} 

\path(20000,-6000)(20000,14000)
\put(20100,-6000){\tiny 2} 
\put(20100,4800){\tiny\color{red} 1} 
\put(20100,14500){\tiny 1} 

\path(22000,-6000)(22000,14000)
\put(22100,-6000){\tiny 2} 
\put(22100,4800){\tiny\color{red} 1} 
\put(22100,14500){\tiny 1} 

\path(24000,-6000)(24000,14000) 
\put(23500,-7000){\tiny $w_{\ell}$}
\put(24100,-6000){\tiny 2} 
\put(24100,4800){\tiny\color{red} 1} 
\put(24100,14500){\tiny 1}


\path(28000,-6000)(28000,14000)
\put(27500,-7000){\tiny $v_1$}
\put(28100,-6000){\tiny 3}
\put(28100,4800){\tiny\color{red} 3}  
\put(28100,14500){\tiny 2} 

\path(30000,-6000)(30000,14000)
\put(30100,-6000){\tiny 3} 
\put(30100,4800){\tiny\color{red} 3} 
\put(30100,14500){\tiny 2} 

\path(32000,-6000)(32000,14000)
\put(32100,-6000){\tiny 3} 
\put(32100,4800){\tiny\color{red} 3} 
\put(32100,14500){\tiny 2}

\path(34000,-6000)(34000,14000)
\put(33500,-7000){\tiny $v_m$}
\put(34100,-6000){\tiny 3}
\put(34100,4800){\tiny\color{red} 3}  
\put(34100,14500){\tiny 2}


\put(28000,-4000){\circle*{200}}
\put(28000,-2000){\circle*{200}}
\put(28000,000){\circle*{200}}
\put(28000,2000){\circle*{200}}
\put(28000,4000){\circle*{200}}
\put(28000,6000){\circle*{200}}
\put(28000,8000){\circle*{200}}
\put(28000,10000){\circle*{200}}
\put(28000,12000){\circle*{200}}

\put(30000,-4000){\circle*{200}}
\put(30000,-2000){\circle*{200}}
\put(30000,000){\circle*{200}}
\put(30000,2000){\circle*{200}}
\put(30000,4000){\circle*{200}}
\put(30000,6000){\circle*{200}}
\put(30000,8000){\circle*{200}}
\put(30000,10000){\circle*{200}}
\put(30000,12000){\circle*{200}}

\put(32000,-4000){\circle*{200}}
\put(32000,-2000){\circle*{200}}
\put(32000,000){\circle*{200}}
\put(32000,2000){\circle*{200}}
\put(32000,4000){\circle*{200}}
\put(32000,6000){\circle*{200}}
\put(32000,8000){\circle*{200}}
\put(32000,10000){\circle*{200}}
\put(32000,12000){\circle*{200}}

\put(34000,-4000){\circle*{200}}
\put(34000,-2000){\circle*{200}}
\put(34000,000){\circle*{200}}
\put(34000,2000){\circle*{200}}
\put(34000,4000){\circle*{200}}
\put(34000,6000){\circle*{200}}
\put(34000,8000){\circle*{200}}
\put(34000,10000){\circle*{200}}
\put(34000,12000){\circle*{200}}

\end{picture}

\end{minipage}
\end{center}

\caption{The result of using the commutation relation (\ref{com-rel}) repeatedly. One obtains a sum over all ways of partitioning $\{\lambda\},\{\mu\}$ into disjoint subsets, and up to a constant that depends on the partition, each term in this sum is of the form shown above. This lattice factorizes into a product of domain wall partition functions $Z(\{\li\}\cup\{\mii\}|\{w\})$ and $Z(\{v\}|\{\mi\}\cup\{\lii\})$.}
\label{resh-pf3}

\end{figure}
At the end of this procedure (which is remniscent of, but less complicated than the algorithm described in Appendix {\bf\ref{app:1}}) we obtain a sum over all terms which have $m$ $t_{32}$ lines at the top of the lattice (with arguments in the set $\{\mi\}\cup\{\lii\}$) and $\ell$ $t_{12}$ lines at the bottom (with arguments in the set $\{\li\}\cup\{\mii\}$). Clearly, the fact that $|\mi|+|\lii| = m$ and $|\li|+|\mii|=\ell$ implies that the sum is constrained by $|\lii| = |\mii|$. Furthermore, from Figure {\bf \ref{resh-pf3}}, we see that all of these terms factorize into a product of two domain wall partition functions (see both lattices in Figure {\bf\ref{dwpf}}). Hence we can deduce the equation
\begin{align}
Z(\{\lambda\},\{\mu\}|\{w\},\{v\})
=
\sum_{\substack{
\{\lambda\}=\{\li\}\cup\{\lii\}
\\
\{\mu\}=\{\mi\}\cup\{\mii\}
}}
K(\{\li\},\{\lii\}|\{\mi\},\{\mii\})
Z(\{\li\}\cup\{\mii\}|\{w\})
Z(\{v\}|\{\mi\}\cup\{\lii\})
\label{97}
\end{align}
where the coefficient $K(\{\li\},\{\lii\}|\{\mi\},\{\mii\})$ depends on the partitioning of sets. Since $K$ does not depend on $\{w\}$ or $\{v\}$, we can exploit the freedom of these variables to isolate a single term in this sum. Let us firstly change the normalization of $Z(\{\lambda\},\{\mu\} | \{w\}, \{v\})$, by defining
\begin{align}
\widetilde{Z}(\{\lambda\},\{\mu\} | \{w\}, \{v\})
=
\frac{Z(\{\lambda\},\{\mu\} | \{w\}, \{v\})}
{f(\{\lambda\},\{w\}) f(\{\mu\},\{w\}) f(\{v\},\{\lambda\}) f(\{v\},\{\mu\})} .
\end{align}
Now let $\{\lambda\} = \{\li\} \cup \{\lii\}$ and $\{\mu\} = \{\mi\} \cup \{\mii\}$ be a given, fixed partitioning of the sets $\{\lambda\}$ and $\{\mu\}$, and consider the limits 
\begin{align*}
\{w\} \rightarrow \{\li\} \cup \{\mii\},
\qquad
\{v\} \rightarrow \{\mi\} \cup \{\lii\}.
\end{align*}
Due to the poles of the domain wall partition function (\ref{res-prop}), a single term in the sum on the right hand side of (\ref{97}) will survive in this limit, corresponding to the selfsame partitioning of the variables. This isolates $K(\{\li\},\{\lii\} | \{\mi\},\{\mii\})$, up to a multiplicative term:  
\begin{align}
\widetilde{Z}(\{\lambda\},\{\mu\} | \{w\}, \{v\})
\Big|_{\substack{\{w\} \rightarrow \{\li\} \cup \{\mii\} \\ \{v\} \rightarrow \{\mi\} \cup \{\lii\}}}
=
\frac{K(\{\li\},\{\lii\} | \{\mi\},\{\mii\})}
{(f(\lii,\li) f(\lii,\mii) f(\mi,\li) f(\mi,\mii))^2} .
\label{99}
\end{align}
To calculate the left hand side of (\ref{99}) explicitly we use the graphical form of 
$\widetilde{Z}(\{\lambda\},\{\mu\} | \{w\},\{v\})$, and note that in this limit, it degenerates into the frozen regions shown in Figure {\bf\ref{resh-frozen}}.
\begin{figure}[H]

\begin{center}
\begin{minipage}{4.3in}

\setlength{\unitlength}{0.0003cm}
\begin{picture}(40000,22000)(8000,-8000)


\put(16800,10700){\scriptsize\color{blue} 1}
\put(20600,10700){\scriptsize\color{blue} $Z$}
\put(28800,10700){\scriptsize\color{blue} 1}
\put(32800,10700){\scriptsize\color{blue} 1}

\put(16600,4700){\scriptsize\color{blue} $Z$}
\put(20300,4700){\scriptsize\color{blue} $f^{-1}$}
\put(28600,4700){\scriptsize\color{blue} $Z$}

\put(16800,0700){\scriptsize\color{blue} 1}
\put(20300,0700){\scriptsize\color{blue} $f^{-1}$}
\put(28800,0700){\scriptsize\color{blue} 1}
\put(32600,0700){\scriptsize\color{blue} $Z$}

\put(16600,-3300){\scriptsize\color{blue} $Z$}
\put(20800,-3300){\scriptsize\color{blue} 1}
\put(28300,-3300){\scriptsize\color{blue} $f^{-1}$}


\put(13000,12000){\tiny 1}
\put(18800,12200){\tiny\color{red} 1}
\put(25800,12200){\tiny\color{red} 2}
\put(30800,12200){\tiny\color{red} 2}
\path(14000,12000)(36000,12000)
\put(36500,12000){\tiny 2} 

\put(13000,10000){\tiny 1}
\put(18800,10200){\tiny\color{red} 1}
\put(25800,10200){\tiny\color{red} 2}
\put(30800,10200){\tiny\color{red} 2}
\path(14000,10000)(36000,10000) 
\put(36500,10000){\tiny 2}

\put(13000,8000){\tiny 1}
\put(18800,8200){\tiny\color{red} 1}
\put(25800,8200){\tiny\color{red} 2}
\put(30800,8200){\tiny\color{red} 2}
\path(14000,8000)(36000,8000)
\put(36500,8000){\tiny 2} 

\put(13000,6000){\tiny 1}
\put(18800,6200){\tiny\color{red} 3}
\put(25800,6200){\tiny\color{red} 3}
\put(30800,6200){\tiny\color{red} 2}
\path(14000,6000)(36000,6000)
\put(36500,6000){\tiny 2} 

\put(13000,4000){\tiny 1}
\put(18800,4200){\tiny\color{red} 3}
\put(25800,4200){\tiny\color{red} 3}
\put(30800,4200){\tiny\color{red} 2}
\path(14000,4000)(36000,4000) 
\put(36500,4000){\tiny 2}


\put(13000,2000){\tiny 3}
\put(18800,2200){\tiny\color{red} 3}
\put(25800,2200){\tiny\color{red} 3}
\put(30800,2200){\tiny\color{red} 3}
\path(14000,2000)(36000,2000)
\put(36500,2000){\tiny 2} 

\put(13000,000){\tiny 3}
\put(18800,200){\tiny\color{red} 3}
\put(25800,200){\tiny\color{red} 3}
\put(30800,200){\tiny\color{red} 3}
\path(14000,000)(36000,000)
\put(36500,000){\tiny 2} 

\put(13000,-2000){\tiny 3}
\put(18800,-1800){\tiny\color{red} 2}
\put(25800,-1800){\tiny\color{red} 2}
\path(14000,-2000)(36000,-2000)
\put(36500,-2000){\tiny 2} 

\put(13000,-4000){\tiny 3}
\put(18800,-3800){\tiny\color{red} 2}
\put(25800,-3800){\tiny\color{red} 2}
\path(14000,-4000)(36000,-4000)
\put(36500,-4000){\tiny 2}


\put(10000,9750){\tiny $\li$}
\put(12000,9750){$\left\{\phantom{\displaystyle{\prod_{\int}^{j}}}\right.$}

\put(10000,4750){\tiny $\lii$}
\put(12000,4750){$\left\{\phantom{\frac{\prod}{\prod}}\right.$}

\put(10000,750){\tiny $\mi$}
\put(12000,750){$\left\{\phantom{\frac{\prod}{\prod}}\right.$}

\put(10000,-3250){\tiny $\mii$}
\put(12000,-3250){$\left\{\phantom{\frac{\prod}{\prod}}\right.$}


\path(16000,-6000)(16000,14000)
\put(15500,6800){\tiny\color{red} 1}
\put(15500,2800){\tiny\color{red} 3}
\put(15500,-1200){\tiny\color{red} 3}
\put(15500,-6000){\tiny 2} 
\put(15500,14200){\tiny 1} 

\path(18000,-6000)(18000,14000)
\put(17500,6800){\tiny\color{red} 1}
\put(17500,2800){\tiny\color{red} 3}
\put(17500,-1200){\tiny\color{red} 3}
\put(17500,-6000){\tiny 2} 
\put(17500,14200){\tiny 1} 

\path(20000,-6000)(20000,14000)
\put(19500,6800){\tiny\color{red} 2}
\put(19500,2800){\tiny\color{red} 2}
\put(19500,-1200){\tiny\color{red} 2}
\put(19500,-6000){\tiny 2} 
\put(19500,14200){\tiny 1} 

\path(22000,-6000)(22000,14000)
\put(21500,6800){\tiny\color{red} 2}
\put(21500,2800){\tiny\color{red} 2}
\put(21500,-1200){\tiny\color{red} 2}
\put(21500,-6000){\tiny 2} 
\put(21500,14200){\tiny 1} 

\path(24000,-6000)(24000,14000)
\put(23500,6800){\tiny\color{red} 2}
\put(23500,2800){\tiny\color{red} 2}
\put(23500,-1200){\tiny\color{red} 2} 
\put(23500,-6000){\tiny 2} 
\put(23500,14200){\tiny 1}


\path(28000,-6000)(28000,14000)
\put(27500,6800){\tiny\color{red} 2}
\put(27500,2800){\tiny\color{red} 3}
\put(27500,-1200){\tiny\color{red} 3}
\put(27500,-6000){\tiny 3} 
\put(27500,14200){\tiny 2} 

\path(30000,-6000)(30000,14000)
\put(29500,6800){\tiny\color{red} 2}
\put(29500,2800){\tiny\color{red} 3}
\put(29500,-1200){\tiny\color{red} 3}
\put(29500,-6000){\tiny 3} 
\put(29500,14200){\tiny 2} 

\path(32000,-6000)(32000,14000)
\put(31500,2800){\tiny\color{red} 2}
\put(31500,-1200){\tiny\color{red} 3}
\put(31500,-6000){\tiny 3} 
\put(31500,14200){\tiny 2}

\path(34000,-6000)(34000,14000)
\put(33500,2800){\tiny\color{red} 2}
\put(33500,-1200){\tiny\color{red} 3}
\put(33500,-6000){\tiny 3} 
\put(33500,14200){\tiny 2}


\put(16300,-8000){\tiny $\mii$}
\put(15950,-6500){$\underbrace{\phantom{...}}$}

\put(21800,-8000){\tiny $\li$}
\put(20100,-6500){$\underbrace{\phantom{............}}$}

\put(28300,-8000){\tiny $\lii$}
\put(27950,-6500){$\underbrace{\phantom{...}}$}

\put(32300,-8000){\tiny $\mi$}
\put(31950,-6500){$\underbrace{\phantom{...}}$}


\put(28000,-4000){\circle*{200}}
\put(28000,-2000){\circle*{200}}
\put(28000,000){\circle*{200}}
\put(28000,2000){\circle*{200}}
\put(28000,4000){\circle*{200}}
\put(28000,6000){\circle*{200}}
\put(28000,8000){\circle*{200}}
\put(28000,10000){\circle*{200}}
\put(28000,12000){\circle*{200}}

\put(30000,-4000){\circle*{200}}
\put(30000,-2000){\circle*{200}}
\put(30000,000){\circle*{200}}
\put(30000,2000){\circle*{200}}
\put(30000,4000){\circle*{200}}
\put(30000,6000){\circle*{200}}
\put(30000,8000){\circle*{200}}
\put(30000,10000){\circle*{200}}
\put(30000,12000){\circle*{200}}

\put(32000,-4000){\circle*{200}}
\put(32000,-2000){\circle*{200}}
\put(32000,000){\circle*{200}}
\put(32000,2000){\circle*{200}}
\put(32000,4000){\circle*{200}}
\put(32000,6000){\circle*{200}}
\put(32000,8000){\circle*{200}}
\put(32000,10000){\circle*{200}}
\put(32000,12000){\circle*{200}}

\put(34000,-4000){\circle*{200}}
\put(34000,-2000){\circle*{200}}
\put(34000,000){\circle*{200}}
\put(34000,2000){\circle*{200}}
\put(34000,4000){\circle*{200}}
\put(34000,6000){\circle*{200}}
\put(34000,8000){\circle*{200}}
\put(34000,10000){\circle*{200}}
\put(34000,12000){\circle*{200}}

\end{picture}

\end{minipage}
\end{center}

\caption{In the limit 
$\{w\} \rightarrow \{\li\}\cup\{\mii\},\ \{v\} \rightarrow \{\mi\}\cup\{\lii\}$, 
$\widetilde{Z}(\{\lambda\},\{\mu\} | \{w\},\{v\})$ factorizes into a product of weight $f^{-1}$ vertices, weight 1 vertices and domain wall partition functions. Furthermore, in this normalization, a domain wall partition function with matching horizontal and vertical variables has weight 1. Regions of vertices with a common weight are indicated on the diagram.}
\label{resh-frozen}

\end{figure}
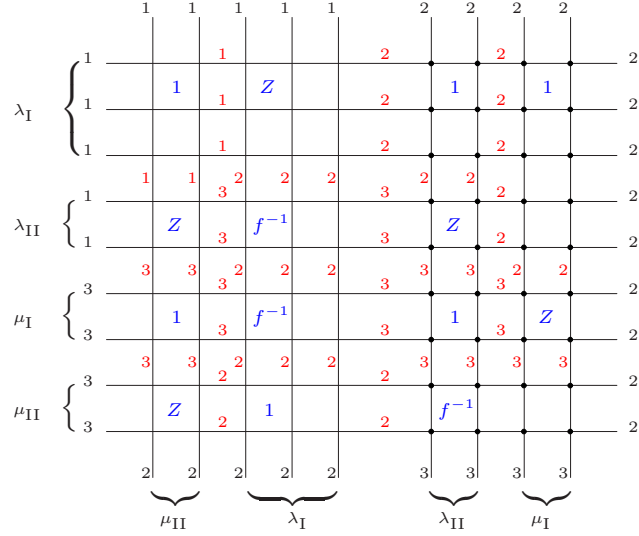
From Figure {\bf\ref{resh-frozen}} we conclude that
\begin{align}
\widetilde{Z}(\{\lambda\},\{\mu\} | \{w\}, \{v\})
\Big|_{\substack{\{w\} \rightarrow \{\li\} \cup \{\mii\} \\ \{v\} \rightarrow \{\mi\} \cup \{\lii\}}}
=
\frac{Z(\lii|\mii)/f(\lii,\mii)}{f(\lii,\li) f(\mi,\li) f(\lii,\mii) f(\mi,\mii)} .
\label{100}
\end{align}
Equating the right hand sides of (\ref{99}) and (\ref{100}), and cancelling common factors, the final result is
\begin{align}
K(\{\li\},\{\lii\}|\{\mi\},\{\mii\})
=
\prod_{\mi,\mii}
f(\mi,\mii)
\prod_{\li,\lii}
f(\lii,\li)
\prod_{\mi,\li}
f(\mi,\li)
Z(\{\lii\}|\{\mii\}) .
\label{K}
\end{align}
Substituting this value into (\ref{97}), we obtain (\ref{resh-pf2}). 

\end{proof}

\subsection{Limiting cases of $Z(\{\lambda\}_{\ell},\{\mu\}_m |\{w\}_{\ell},\{v\}_m)$}
\label{pf-limits}

In this subsection we list results about the function $Z(\{\lambda\}_{\ell},\{\mu\}_m |\{w\}_{\ell},\{v\}_m)$ when one of its sets of variables becomes infinite. These results are needed when we study the scalar product (\ref{sl3-g}) in the same limit.
\begin{myTheorem}
We claim the following limits,
\begin{align}
\label{zlim1}
Z(\{\lambda\}_{\ell},\{\infty\}_m|\{w\}_{\ell},\{v\}_m)
&\equiv
\frac{1}{m!}
\lim_{\mu_m,\dots,\mu_1 \rightarrow \infty}
\Big(
\mu_m \dots \mu_1
Z(\{\lambda\}_{\ell},\{\mu\}_m|\{w\}_{\ell},\{v\}_m)
\Big)
\\
&=
(-)^m
Z(\{\lambda\} | \{w\})
\nonumber
\\
\label{zlim3}
Z(\{\infty\}_{\ell},\{\mu\}_m|\{w\}_{\ell},\{v\}_m)
&\equiv
\frac{1}{\ell!}
\lim_{\lambda_{\ell},\dots,\lambda_1 \rightarrow \infty}
\Big(
\lambda_{\ell} \dots \lambda_1
Z(\{\lambda\}_{\ell},\{\mu\}_m|\{w\}_{\ell},\{v\}_m)
\Big)
\\
&=
Z(\{v\}|\{\mu\})
\nonumber
\\
Z(\{\lambda\}_{\ell},\{\mu\}_m|\{w\}_{\ell},\{\infty\}_m)
&\equiv
\frac{1}{m!}
\lim_{v_m,\dots,v_1 \rightarrow \infty}
\Big(
v_m \dots v_1
Z(\{\lambda\}_{\ell},\{\mu\}_m|\{w\}_{\ell},\{v\}_m)
\Big)
\label{zlim2}
\\
&= 
f(\{\mu\},\{w\})
Z(\{\lambda\} | \{w\})
\nonumber
\\
\label{zlim4}
Z(\{\lambda\}_{\ell},\{\mu\}_m|\{\infty\}_{\ell},\{v\}_m)
&\equiv
\frac{1}{\ell!}
\lim_{w_{\ell},\dots,w_1 \rightarrow \infty}
\Big(
w_{\ell} \dots w_1
Z(\{\lambda\}_{\ell},\{\mu\}_m|\{w\}_{\ell},\{v\}_m)
\Big)
\\
&=
(-)^{\ell}
f(\{v\},\{\lambda\})
Z(\{v\}|\{\mu\}) .
\nonumber
\end{align}
\end{myTheorem}

\begin{proof}
Starting from the exact expression (\ref{resh-pf2}), the limit (\ref{zlim1}) can be taken without difficulty. In this case the sum over partitions of $\{\mu\}$ trivializes, because when $\{\mii\}$ is non-empty each term in the sum contains the product $Z(\{\lii\}|\{\mii\}) Z(\{\li\}\cup\{\mii\}|\{w\})$, and this vanishes in the proposed limit. Hence only one term in (\ref{resh-pf2}) will survive, corresponding to $\{\li\} = \{\lambda\},\{\mi\} = \{\mu\}$, and we see that
\begin{align}
Z(\{\lambda\}_{\ell},\{\infty\}_m|\{w\}_{\ell},\{v\}_m)
&=
\frac{1}{m!}
\lim_{\mu_m,\dots,\mu_1 \rightarrow \infty}
\Big(
\mu_m \dots \mu_1
f(\{\mu\},\{\lambda\})
Z(\{\lambda\}|\{w\})
Z(\{v\}|\{\mu\})
\Big)
\\
&=
(-)^m
Z(\{\lambda\}|\{w\}) .
\nonumber
\end{align}
A similar argument applies to proving (\ref{zlim3}).

The limits (\ref{zlim2}) and (\ref{zlim4}) are slightly more complicated. We consider only (\ref{zlim2}), as this indicates the way to prove (\ref{zlim4}). Starting from (\ref{resh-pf2}) and using (\ref{all-limits}), we straight away find that   
\begin{multline}
Z(\{\lambda\},\{\mu\}|\{w\},\{\infty\})
=
\sum_{\substack{
\{\lambda\}=\{\li\}\cup\{\lii\}
\\
\{\mu\}=\{\mi\}\cup\{\mii\}
}}
\prod_{\mi,\mii}
f(\mi,\mii)
\prod_{\li,\lii}
f(\lii,\li)
\prod_{\mi,\li}
f(\mi,\li)
\\
\times
Z(\{\lii\}|\{\mii\})
Z(\{\li\}\cup\{\mii\}|\{w\}) .
\label{zlim2'}
\end{multline}
Now it becomes a matter of showing that the right hand sides of (\ref{zlim2}) and (\ref{zlim2'}) are equivalent. We do that in the following lemma.
\begin{myLemma}
The domain wall partition function satisfies the identity
\begin{multline}
f(\{\mu\},\{w\})
Z(\{\lambda\} | \{w\})
=
\sum_{\substack{
\{\lambda\}=\{\li\}\cup\{\lii\}
\\
\{\mu\}=\{\mi\}\cup\{\mii\}
}}
\prod_{\mi,\mii}
f(\mi,\mii)
\prod_{\li,\lii}
f(\lii,\li)
\prod_{\mi,\li}
f(\mi,\li)
\\
\times
Z(\{\lii\}|\{\mii\})
Z(\{\li\}\cup\{\mii\}|\{w\}) .
\end{multline}
We remark that some similar formulae (but with summation over different partitionings) appear in \cite{bprs2}.
\end{myLemma}

\begin{proof} 
The arguments required are analogous to those in the proof of Theorem {\bf 1.} We represent $f(\{\mu\},\{w\}) Z(\{\lambda\}|\{w\})$ as the partition function of the lattice on the left in Figure {\bf\ref{fZ}}.
\begin{figure}[H]

\begin{center}
\begin{minipage}{4.3in}

\setlength{\unitlength}{0.00025cm}
\begin{picture}(40000,19500)(10000,1500)


\path(14000,20000)(26000,20000) \put(12000,20000){\tiny $\lambda_1$} 
\put(14000,20500){\tiny 1}  
\put(25500,20500){\tiny 2}

\path(14000,18000)(26000,18000)
\put(14000,18500){\tiny 1}  
\put(25500,18500){\tiny 2}

\path(14000,16000)(26000,16000)
\put(14000,16500){\tiny 1}  
\put(25500,16500){\tiny 2}

\path(14000,14000)(26000,14000)
\put(14000,14500){\tiny 1}  
\put(25500,14500){\tiny 2}

\path(14000,12000)(26000,12000) \put(12000,12000){\tiny $\lambda_{\ell}$}
\put(14000,12500){\tiny 1} 
\put(25500,12500){\tiny 2}


\path(14000,10000)(26000,10000) \put(12000,10000){\tiny $\mu_1$} 
\put(14000,10500){\tiny 2}  
\put(25500,10500){\tiny 2}

\path(14000,8000)(26000,8000)
\put(14000,8500){\tiny 2}  
\put(25500,8500){\tiny 2}

\path(14000,6000)(26000,6000)
\put(14000,6500){\tiny 2}  
\put(25500,6500){\tiny 2}

\path(14000,4000)(26000,4000) \put(12000,4000){\tiny $\mu_m$}
\put(14000,4500){\tiny 2}  
\put(25500,4500){\tiny 2}


\path(16000,2000)(16000,22000) \put(15500,000){\tiny $w_1$}
\put(16100,1000){\tiny 2}
\put(16100,10800){\tiny\color{red} 2} 
\put(16100,22500){\tiny 1}

\path(18000,2000)(18000,22000)
\put(18100,1000){\tiny 2} 
\put(18100,10800){\tiny\color{red} 2} 
\put(18100,22500){\tiny 1}

\path(20000,2000)(20000,22000)
\put(20100,1000){\tiny 2} 
\put(20100,10800){\tiny\color{red} 2} 
\put(20100,22500){\tiny 1}

\path(22000,2000)(22000,22000)
\put(22100,1000){\tiny 2} 
\put(22100,10800){\tiny\color{red} 2} 
\put(22100,22500){\tiny 1}

\path(24000,2000)(24000,22000) \put(23500,000){\tiny $w_{\ell}$}
\put(24100,1000){\tiny 2} 
\put(24100,10800){\tiny\color{red} 2} 
\put(24100,22500){\tiny 1}



\path(39000,20000)(51000,20000)  
\put(39000,20500){\tiny 2}  
\put(50500,20500){\tiny 2}

\path(39000,18000)(51000,18000)
\put(39000,18500){\tiny 2}  
\put(50500,18500){\tiny 2}

\path(39000,16000)(51000,16000)
\put(39000,16500){\tiny 2}  
\put(50500,16500){\tiny 2}

\path(39000,14000)(51000,14000)
\put(39000,14500){\tiny 2}  
\put(50500,14500){\tiny 2}

\path(39000,12000)(51000,12000) 
\put(39000,12500){\tiny 1} 
\put(50500,12500){\tiny 2}


\path(39000,10000)(51000,10000)
\put(39000,10500){\tiny 1}  
\put(50500,10500){\tiny 2}

\path(39000,8000)(51000,8000)
\put(39000,8500){\tiny 1}  
\put(50500,8500){\tiny 2}

\path(39000,6000)(51000,6000)
\put(39000,6500){\tiny 1}  
\put(50500,6500){\tiny 2}

\path(39000,4000)(51000,4000) 
\put(39000,4500){\tiny 1}  
\put(50500,4500){\tiny 2}


\path(41000,2000)(41000,22000) \put(40500,000){\tiny $w_1$}
\put(41100,1000){\tiny 2} 
\put(41100,12800){\tiny\color{red} 1} 
\put(41100,22500){\tiny 1}

\path(43000,2000)(43000,22000)
\put(43100,1000){\tiny 2} 
\put(43100,12800){\tiny\color{red} 1} 
\put(43100,22500){\tiny 1}

\path(45000,2000)(45000,22000)
\put(45100,1000){\tiny 2} 
\put(45100,12800){\tiny\color{red} 1} 
\put(45100,22500){\tiny 1}

\path(47000,2000)(47000,22000)
\put(47100,1000){\tiny 2} 
\put(47100,12800){\tiny\color{red} 1} 
\put(47100,22500){\tiny 1}

\path(49000,2000)(49000,22000) \put(48500,000){\tiny $w_{\ell}$}
\put(49100,1000){\tiny 2}
\put(49100,12800){\tiny\color{red} 1} 
\put(49100,22500){\tiny 1}


\put(35000,18750){\tiny $\lii$}
\put(37000,18750){$\left\{\phantom{\frac{\prod}{\prod}}\right.$}

\put(35000,14750){\tiny $\mi$}
\put(37000,14750){$\left\{\phantom{\frac{\prod}{\prod}}\right.$}

\put(35000,10000){\tiny $\li$}
\put(37000,9500){$\left\{\phantom{\displaystyle{\prod_{i}^{j}}}\right.$}

\put(35000,4750){\tiny $\mii$}
\put(37000,4750){$\left\{\phantom{\frac{\prod}{\prod}}\right.$}
\end{picture}

\end{minipage}
\end{center}

\caption{On the left, lattice representation of $f (\{\mu\},\{w\}) Z(\{\lambda\} | \{w\})$. The bottom $m$ horizontal lines factorize trivially into the product of weights $\prod_{i=1}^{m} \prod_{j=1}^{\ell} f(\mu_i,w_j)$ and the remaining part of the lattice constitutes the domain wall partition function, as shown in Figure {\bf\ref{dwpf}}. On the right, lattice representation of $Z(\{\li\}\cup\{\mii\}|\{w\})$. The top $m$ lines contribute trivially, since all vertices in this part of the lattice have weight 1.}

\label{fZ}

\end{figure}
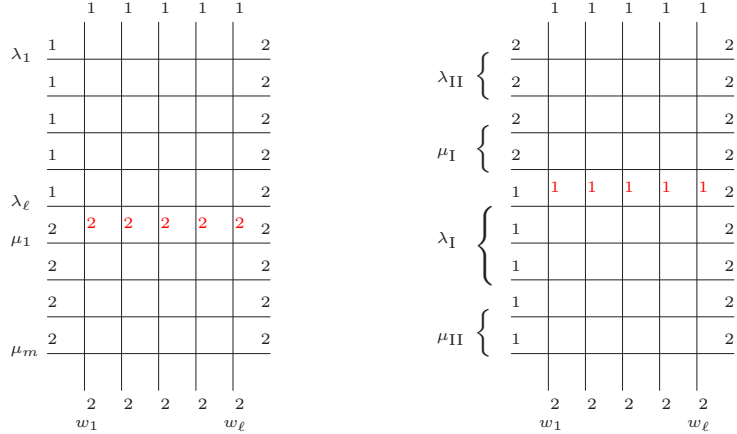
Considering the diagram on the left of Figure {\bf\ref{fZ}}, the $m$ lowest lattice lines can be repositioned to the top using the commutation relation (\ref{2db}) repeatedly. This produces a sum over partitions of $\{\lambda\},\{\mu\}$ into disjoint subsets, with coefficients $K(\{\li\},\{\lii\}|\{\mi\},\{\mii\})$, multiplying the partition function shown on the right of Figure {\bf\ref{fZ}}.
In other words, we conclude that
\begin{align}
f(\{\mu\},\{w\}) Z(\{\lambda\}|\{w\})
=
\sum_{\substack{
\{\lambda\}=\{\li\}\cup\{\lii\}
\\
\{\mu\}=\{\mi\}\cup\{\mii\}
}}
K(\{\li\},\{\lii\}|\{\mi\},\{\mii\})
Z(\{\li\}\cup\{\mii\}|\{w\}) .
\label{zlim2''}
\end{align}
Using the freedom of choice of the $\{w\}$ variables to isolate a single term in the sum (\ref{zlim2''}), it is possible to show that $K(\{\li\},\{\lii\}|\{\mi\},\{\mii\})$ is as given by (\ref{K}). This is done using a completely analogous method to that explained at the end of Subsection {\bf\ref{pf-resh}}.
\end{proof}

With Lemma {\bf 1} we have shown that the right hand sides of (\ref{zlim2}) and (\ref{zlim2'}) are equal, proving (\ref{zlim2}). 

\end{proof}

\subsection{Imposing Bethe equations on $\{\lb\}$ and $\{\mb\}$}

Similarly to the previous section we now restrict our attention to the case when two sets of variables in the scalar product, $\{\lb\}$ and $\{\mb\}$, satisfy the Bethe equations (\ref{1bethe-sl3}) and (\ref{2bethe-sl3}). That is, we assume that
\begin{align}
r_1(\lb_i)
&=
\frac{a_1(\lb_i)}{a_2(\lb_i)}
=
-
\prod_{j=1}^{\ell}
\left(
\frac{\lb_i-\lb_j + 1}{\lb_i - \lb_j - 1}
\right)
\prod_{k=1}^{m}
f(\mb_k,\lb_i),
\quad\quad
\forall\ 1 \leq i \leq \ell.
\label{n-bethe1}
\\
r_2(\mb_i)
&=
\frac{a_2(\mb_i)}{a_3(\mb_i)}
=
-
\prod_{j=1}^{m}
\left(
\frac{\mb_i-\mb_j + 1}{\mb_i - \mb_j - 1}
\right)
\prod_{k=1}^{\ell}
\frac{1}{f(\mb_i,\lb_k)},
\quad\quad
\forall\ 1 \leq i \leq m.
\label{n-bethe2}
\end{align}

\subsection{Summation formula with allowance for Bethe equations}

Substituting (\ref{n-bethe1}) and (\ref{n-bethe2}) into (\ref{sl3-g}), we obtain the expression
\begin{multline}
\label{resh-sum}
\langle\langle \{\mc\}, \{\lc\} | \{\lb\}, \{\mb\} \rangle\rangle
=
\sum
(-)^{|\lbi|+|\mbii|}
\prod_{\lcii} r_1(\lcii)
\prod_{\mci} r_2(\mci)
\\
\times
\prod_{\lbi}\left( 
\prod_{j=1}^{\ell} \left( \frac{\lbi-\lb_j+1}{\lbi-\lb_j-1} \right)
\prod_{k=1}^{m} f(\mb_k,\lbi) 
\right)
\prod_{\mbii}\left( 
\prod_{j=1}^{m} \left( \frac{\mbii-\mb_j+1}{\mbii-\mb_j-1} \right)
\prod_{k=1}^{\ell} \frac{1}{f(\mbii,\lb_k)}
\right)
\\
\times
f(\lci,\lcii) f(\lbii,\lbi) 
f(\mcii,\mci) f(\mbi,\mbii)
f(\mbii,\lbii) f(\mci,\lci)
\\
\times
Z(\{\lbii\},\{\mci\} | \{\lcii\},\{\mbi\})
Z(\{\lci\},\{\mbii\} | \{\lbi\},\{\mcii\}) .
\end{multline}
Drawing upon what we learn from the $SU(2)$ scalar product, it is natural to expect that having allowed for the Bethe equations, (\ref{resh-sum}) can be summed to a more compact expression. For the moment, we do not know how to do that. Part of the difficulty arises with the terms $Z(\{\lbii\},\{\mci\} | \{\lcii\},\{\mbi\})$ and $Z(\{\lci\},\{\mbii\} | \{\lbi\},\{\mcii\})$ which have great combinatorial complexity, whereas their $SU(2)$ analogues (domain wall partition functions) are determinants. For this reason we move on to consider limiting cases of (\ref{resh-sum}), in the hope it will illuminate its structure without taking any limit. 

\subsection{$\{\mb\} \rightarrow \infty$ limit of $SU(3)$ scalar product}
\label{mb-infty}

Starting from (\ref{resh-sum}), we consider the limit
\begin{align}
\langle\langle \{\mc\},\{\lc\}|\{\lb\},\{\infty\} \rangle\rangle
\equiv
\frac{1}{m!}
\lim_{\mb_{m}, \dots, \mb_1 \rightarrow \infty}
\Big(
\mb_{m} \dots \mb_1 
\langle\langle 
\{\mc\} , \{\lambda^C\} | \{\lambda^B\}, \{\mb\} 
\rangle\rangle
\Big) .
\end{align}
Using (\ref{zlim1}) and (\ref{zlim2}) to take limits of $Z(\{\lbii\},\{\mci\} | \{\lcii\},\{\mbi\})$ and $Z(\{\lci\},\{\mbii\} | \{\lbi\},\{\mcii\})$, we find that
\begin{multline}
\label{sl3-b2}
\langle\langle \{\mc\},\{\lc\}|\{\lb\},\{\infty\} \rangle\rangle
=
\frac{1}{m!}
\sum_{\substack{
\{\lc\}=\{\lci\}\cup\{\lcii\}
\\
\{\lb\}=\{\lbi\}\cup\{\lbii\}
}}
\left(
\sum_{k=0}^{m}
\sum_{
|\mci|=m-k,|\mcii|=k 
}
(-)^{|\lbi|}
\binom{m}{k}
\right.
\\
\times
\prod_{\lcii}
r_1(\lcii)
\prod_{\mci}
r_2(\mci)
\prod_{\lbi} \prod_{j=1}^{\ell} \left(\frac{\lbi-\lb_j+1}{\lbi-\lb_j-1}\right)
f(\lci,\lcii) f(\lbii,\lbi) 
f(\mcii,\mci)
f(\mci,\lci)
\\
\left.
\phantom{\prod_{\lbi} \prod_{j=1}^{\ell}}
\times
(m-k)!
f(\mci,\lcii)
Z(\{\lbii\}|\{\lcii\})
(-)^k
k!
Z(\{\lci\}|\{\lbi\})
\right).
\end{multline}
Now we make the trivial observation $f(\mci,\lci) f(\mci,\lcii) = f(\mci,\lc)$, for any partitioning of $\{\lc\}$ into disjoint subsets. Using this in (\ref{sl3-b2}) and cancelling combinatoric factors, we obtain the factorization 
\begin{multline}
\label{sl3-b3}
\langle\langle \{\mc\} , \{\lambda^C\} | \{\lambda^B\}, \{\infty\} \rangle\rangle
=
\left(
\sum
(-)^{|\mcii|}
\prod_{\mci}
\left(
r_2(\mci)
\prod_{k=1}^{\ell}
f(\mci,\lc_k)
\right)
f(\mcii,\mci)
\right)
\\
\times
\left(
\sum
(-)^{|\lbi|}
\prod_{\lbi} \prod_{j=1}^{\ell} \left(\frac{\lbi-\lb_j+1}{\lbi-\lb_j-1}\right)
\prod_{\lcii}
r_1(\lcii)
f(\lci,\lcii) f(\lbii,\lbi)
Z(\{\lbii\}|\{\lcii\})
Z(\{\lci\}|\{\lbi\}) 
\right)
\end{multline}
where the first sum ranges over partitions of $\{\mc\}$ into $\{\mci\}\cup\{\mcii\}$, and the second sum ranges over partitions of $\{\lc\},\{\lb\}$ which obey (\ref{disj1}). But using the equality of (\ref{slav-sum}) and (\ref{betscal}), as well as (\ref{sl2-b2}) and (\ref{slav-lim}), we know how to compute both sums in (\ref{sl3-b3}). Hence we obtain the product of determinants
\begin{multline}
\langle\langle \{\mc\} , \{\lambda^C\} | \{\lambda^B\}, \{\infty\} \rangle\rangle
=
\frac{
\det\left(
\displaystyle{
(\mu^C_i)^{j-1}
r_2(\mu^C_i)
\prod_{k=1}^{\ell}
\left(
\frac{\displaystyle{\mu^C_i-\lambda^C_k+1}}{\displaystyle{\mu^C_i-\lambda^C_k}}
\right)
-
(\mu^C_i+1)^{j-1}
}
\right)_{1 \leq i, j \leq m}
}
{\displaystyle{
\prod_{1\leq i<j \leq m}
(\mc_j-\mc_i)
}}
\label{factorized1}
\\
\times
\frac{
\det \left(
\frac{\displaystyle{1}}{\displaystyle{\lambda^B_j-\lambda^C_i}}
\left(\displaystyle{
\prod_{k\not= j}^{\ell}(\lb_k-\lc_i+1)
r_1(\lc_i)
-
\prod_{k\not=j}^{\ell} (\lb_k-\lc_i-1)
}
\right)
\right)_{1\leq i,j \leq \ell}
}
{\displaystyle{
\prod_{1\leq i < j \leq \ell}
(\lc_j-\lc_i)(\lb_i-\lb_j)
}}
\end{multline}

\subsection{$\{\lb\} \rightarrow \infty$ limit of $SU(3)$ scalar product}
\label{lb-infty}

We basically repeat the process of the last subsection, and take the limit
\begin{align}
\langle\langle
\{\mc\},\{\lc\}|
\{\infty\},\{\mb\}
\rangle\rangle
\equiv
\frac{1}{\ell!}
\lim_{\lb_{\ell},\dots,\lb_1 \rightarrow \infty}
\Big(
\lb_{\ell}\dots \lb_1
\langle\langle
\{\mc\},\{\lc\}|
\{\lb\},\{\mb\}
\rangle\rangle
\Big) .
\end{align}
Again we start from (\ref{resh-sum}) and now use (\ref{zlim3}) and (\ref{zlim4}) to take limits of $Z(\{\lbii\},\{\mci\} | \{\lcii\},\{\mbi\})$ and $Z(\{\lci\},\{\mbii\} | \{\lbi\},\{\mcii\})$, which gives
\begin{multline}
\label{sl3-b4}
\langle\langle \{\mc\},\{\lc\}|\{\infty\},\{\mb\} \rangle\rangle
=
\frac{1}{\ell!}
\sum_{\substack{
\{\mc\}=\{\mci\}\cup\{\mcii\}
\\
\{\mb\}=\{\mbi\}\cup\{\mbii\}
}}
\left(
\sum_{k=0}^{\ell}
\sum_{
|\lci|=\ell-k,|\lcii|=k 
}
(-)^{|\mbii|}
\binom{\ell}{k}
\right.
\\
\times
\prod_{\lcii}
r_1(\lcii)
\prod_{\mci}
r_2(\mci)
\prod_{\mbii} \prod_{j=1}^{m} \left(\frac{\mbii-\mb_j+1}{\mbii-\mb_j-1}\right)
f(\lci,\lcii)  
f(\mcii,\mci)
f(\mbi,\mbii)
f(\mci,\lci)
\\
\left.
\phantom{\prod_{\mbii} \prod_{j=1}^{m}}
\times
k!
Z(\{\mbi\}|\{\mci\})
(-)^{\ell-k}
(\ell-k)!
f(\mcii,\lci)
Z(\{\mcii\}|\{\mbii\})
\right) .
\end{multline}
Observe that $f(\mci,\lci) f(\mcii,\lci) = f(\mc,\lci)$ for any partitioning of $\{\mc\}$ into disjoint subsets. Using this fact and cancelling combinatoric factors in (\ref{sl3-b4}), we obtain the factorization 
\begin{multline}
\label{sl3-b5}
\langle\langle \{\mc\} , \{\lambda^C\} | \{\infty\}, \{\mb\} \rangle\rangle
=
f(\mc,\lc)
\left(
\sum
(-)^{|\lci|}
\prod_{\lcii}
\left(
r_1(\lcii)
\prod_{k=1}^{m}
\frac{1}{f(\mc_k,\lcii)}
\right)
f(\lci,\lcii)
\right)
\\
\times
\left(
\sum
(-)^{|\mbii|}
\prod_{\mbii} \prod_{j=1}^{m} \left(\frac{\mbii-\mb_j+1}{\mbii-\mb_j-1}\right)
\prod_{\mci}
r_2(\mci)
f(\mcii,\mci) f(\mbi,\mbii)
Z(\{\mbi\}|\{\mci\})
Z(\{\mcii\}|\{\mbii\}) 
\right)
\end{multline}
where the first sum ranges over partitions of $\{\lc\}$ into $\{\lci\}\cup\{\lcii\}$, and the second sum ranges over partitions of $\{\mc\},\{\mb\}$ which obey (\ref{disj2}). As before, we know how to compute both the sums in (\ref{sl3-b5}). The result is
\begin{multline}
\langle\langle \{\mc\} , \{\lambda^C\} | \{\infty\}, \{\mb\} \rangle\rangle
=
\frac{
\det\left(
\displaystyle{
(\lambda^C_i)^{j-1}
r_1(\lambda^C_i)
-
(\lambda^C_i+1)^{j-1}
\prod_{k=1}^{m}
\left(
\frac{\displaystyle{\mu^C_k-\lambda^C_i+1}}{\displaystyle{\mu^C_k-\lambda^C_i}}
\right)
}
\right)_{1 \leq i, j \leq \ell}
}
{\displaystyle{
\prod_{1\leq i<j \leq \ell}
(\lc_j-\lc_i)
}}
\label{factorized2}
\\
\times
\frac{
\det \left(
\frac{\displaystyle{1}}{\displaystyle{\mu^B_j-\mu^C_i}}
\left(\displaystyle{
\prod_{k\not= j}^{m}(\mb_k-\mc_i+1)
r_2(\mc_i)
-
\prod_{k\not=j}^{m} (\mb_k-\mc_i-1)
}
\right)
\right)_{1\leq i,j \leq m}
}
{\displaystyle{
\prod_{1\leq i < j \leq m}
(\mc_j-\mc_i)(\mb_i-\mb_j)
}}
\end{multline}

\subsection{Comments about consistency}

We end by remarking that equations (\ref{factorized1}) and (\ref{factorized2}) are valid in the regimes $\{\mb\} \rightarrow \infty$ and $\{\lb\} \rightarrow \infty$, respectively, and are each independent of the order that their variables tend to infinity.

However this is definitely {\it not} the case if we now send the surviving set of Bethe variables in (\ref{factorized1}) and (\ref{factorized2}) to infinity. In fact the quantity $\langle\langle \{\mc\},\{\lc\}|\{\infty\},\{\infty\}\rangle\rangle$ is ambiguous, since it is sensitive to the order in which its limits are taken. For example, sending $\{\lb\} \rightarrow \infty$ in (\ref{factorized1}) gives 
\begin{multline}
\frac{1}{\ell!}
\lim_{\lb_{\ell},\dots,\lb_1 \rightarrow \infty}
\Big(
\lb_{\ell} \dots \lb_1 
\langle\langle \{\mc\}, \{\lc\} | \{\lb\}, \{\infty\} \rangle\rangle
\Big)
=
\frac{\det\Big(
(\lambda^C_i)^{j-1}
r_1(\lambda^C_i)
-
(\lambda^C_i+1)^{j-1}
\Big)_{1 \leq i, j \leq \ell}
}
{\displaystyle{
\prod_{1 \leq i < j \leq \ell} (\lambda^C_j - \lambda^C_i)
}}
\\
\times
\frac{
\det\left(
\displaystyle{
(\mu^C_i)^{j-1}
r_2(\mu^C_i)
\prod_{k=1}^{\ell}
\left(
\frac{\displaystyle{\mu^C_i-\lambda^C_k+1}}{\displaystyle{\mu^C_i-\lambda^C_k}}
\right)
-
(\mu^C_i+1)^{j-1}
}
\right)_{1 \leq i, j \leq m}
}
{\displaystyle{
\prod_{1\leq i<j \leq m}
(\mc_j-\mc_i)
}}
\label{factorized3}
\end{multline}
while sending $\{\mb\} \rightarrow \infty$ in (\ref{factorized2}) gives
\begin{multline}
\frac{1}{m!}
\lim_{\mb_{m},\dots,\mb_1 \rightarrow \infty}
\Big(
\mb_{m} \dots \mb_1 
\langle\langle \{\mc\}, \{\lc\} | \{\infty\}, \{\mb\} \rangle\rangle
\Big)
=
\frac{\det\Big(
(\mu^C_i)^{j-1}
r_2(\mu^C_i)
-
(\mu^C_i+1)^{j-1}
\Big)_{1 \leq i, j \leq m}
}
{\displaystyle{
\prod_{1 \leq i < j \leq m} (\mu^C_j - \mu^C_i)
}}
\\
\times
\frac{
\det\left(
\displaystyle{
(\lambda^C_i)^{j-1}
r_1(\lambda^C_i)
-
(\lambda^C_i+1)^{j-1}
\prod_{k=1}^{m}
\left(
\frac{\displaystyle{\mu^C_k-\lambda^C_i+1}}{\displaystyle{\mu^C_k-\lambda^C_i}}
\right)
}
\right)_{1 \leq i, j \leq \ell}
}
{\displaystyle{
\prod_{1\leq i<j \leq \ell}
(\lc_j-\lc_i)
}}
\label{factorized4}
\end{multline}
The answers obtained are different, contradicting the na\"{\i}ve expectation that one should obtain consistent answers in this way. In reality there is no contradiction, because although 
$\langle \langle \{\mc\}, \{\lc\} | \{\lb\}, \{\mb\} \rangle \rangle$ is symmetric in $\{\lb\}$ and $\{\mb\}$ separately, it is not symmetric in both sets simultaneously (indeed, from (\ref{sl3-g}), we can see that the scalar product has simple poles at $\mb_i = \lb_j$). Hence there is no reason for the two limits $\{\lb\} \rightarrow \infty$ and $\{\mb\} \rightarrow \infty$ to commute. From another point of view,  
\begin{align}
\lim_{x\rightarrow \infty}
\frac{\lambda^B_{\ell} \dots \lambda^B_1 \mu^B_m \dots \mu^B_1}{\ell!\ m!}
\langle \langle \{\mc\}, \{\lc\} | \{\lb\}, \{\mb\} \rangle \rangle
=
\left\{
\begin{array}{lll}
(\ref{factorized3}), & \text{if\ } \lambda^B_i = x^i, & \mu^B_j = x^{\ell+j}
\\
\\
(\ref{factorized4}), & \text{if\ } \lambda^B_i = x^{m+i}, & \mu^B_j = x^{j}
\end{array}
\right.
\end{align}
which makes it explicit that (\ref{factorized3}) and (\ref{factorized4}) come from different limits of the scalar product (\ref{sl3-g}).    

\section{Discussion}
\label{s:discuss}

This work contains two new results. The first is equation (\ref{resh-pf2}), which evaluates the partition function $Z(\{\lambda\},\{\mu\}|\{w\},\{v\})$. Using this result, the sum formula (\ref{sl3-g}) becomes a completely explicit (but rather complicated) expression for the generic $SU(3)$ scalar product. As we commented, $Z(\{\lambda\},\{\mu\}|\{w\},\{v\})$ is the natural analogue of the domain wall partition function at the $SU(3)$ level. For this reason, it would be nice to obtain a more compact expression for this quantity. Unfortunately, tests of small examples of this object reveal that it does not factorize and cannot easily be expressed as a determinant. This casts serious doubt on the claim that (\ref{resh-sum}) can be summed as a single determinant. Indeed, following a remark in the conclusion of \cite{bprs1}, $Z(\{\lambda\},\{\mu\}|\{w\},\{v\})$ is actually a particular case of the scalar product (\ref{resh-sum}) (obtained by an appropriate specialization of the parameters in (\ref{resh-sum}), such that only one term in the sum survives). Hence if $Z(\{\lambda\},\{\mu\}|\{w\},\{v\})$ cannot be written as a determinant, neither can the sum (\ref{resh-sum}).

The second is equations (\ref{factorized1}) and (\ref{factorized2}), which evaluate the scalar product between a generic Bethe vector and a Bethe eigenvector in the limit where one set of Bethe variables becomes infinite. The expressions obtained are quite simple, since they are products of two objects which are familiar from $SU(2)$ theory. Furthermore, the appearance of a Slavnov determinant in (\ref{factorized1}) and (\ref{factorized2}) can be justified by the fact that these limits effectively send an $SU(3)$ Bethe eigenvector to an $SU(2)$ one, \textit{cf.} the Bethe equations (\ref{n-bethe1}) and (\ref{n-bethe2}), which collapse to a single set of $SU(2)$ Bethe equations when $\{\lambda^B\} \rightarrow \infty$ or $\{\mu^B\} \rightarrow \infty$. The appearance of the second determinant in (\ref{factorized1}) and (\ref{factorized2}) is less obvious, but since these are themselves just limits of Slavnov-type determinants (\textit{cf.} Subsection {\bf\ref{ssec:beh}}), we can argue that in these limits the $SU(3)$ scalar product disentangles into a product of two $SU(2)$ ones. 

From a mathematical point of view, the factorization formulae are interesting for the following reason: given that the scalar product is known in factorized form in the cases (\ref{factorized1}) and (\ref{factorized2}), and has a determinant form when both Bethe vectors are on-shell \cite{res,bprs2}, it is still plausible that the sum (\ref{resh-sum}) admits a more compact expression without taking limits of the variables. In particular, any new proposal for the form of this object (determinant or otherwise) must behave appropriately under the limits or specializations described above. Hence the results in this paper provide a new way to check the validity of an Ansatz for the scalar product, or rather, to dismiss those guesses which may not otherwise be obviously false.

From a physical point of view, it is natural to hope for a simplification of (\ref{resh-sum}), which would make the calculation of correlation functions in $SU(3)$-invariant models computationally tractable. While this has not yet been achieved, let us remark that the limiting cases (\ref{factorized1}) and (\ref{factorized2}) have already found application in the calculation of certain tree-level 3-point functions in the $SU(3)$ sector of $\mathcal{N}=4$ supersymmetric Yang--Mills theory \cite{fjks}.

\section*{Acknowledgments}

I would like to thank O~Foda for collaboration on \cite{fw} which greatly influenced this work, I~Kostov for sharing the formula (\ref{kost2}), and E~Ragoucy and P~Zinn-Justin for discussions on related topics. This work was supported by the Australian Research Council.

After the work in Subsections {\bf\ref{pf-resh}}--{\bf\ref{pf-limits}} was completed, O~Foda communicated to me that J~Caetano has obtained a factorization formula \cite{cae} for the $SU(3)$ scalar product in the context mentioned previously: {\bf 1.} An XXX spin chain based on fundamental representations of $\mathcal{Y}(sl_3)$, which is a special case of the generalized model presented in this paper, and 
{\bf 2.} In the limit where both sets of Bethe roots $\{\lb\},\{\mb\} \rightarrow \infty$ simultaneously. This communication, together with the results obtained in \cite{fw} and those in Subsection {\bf\ref{pf-limits}}, led to the study of the individual limits $\{\lb\} \rightarrow \infty$ and $\{\mb\} \rightarrow \infty$ and to equations (\ref{factorized1}) and (\ref{factorized2}), in Subsections {\bf\ref{mb-infty}} and {\bf\ref{lb-infty}}.

\appendix
\section{Derivation of the sum form for the $SU(2)$-invariant scalar product}
\label{app:1}

We give a detailed derivation of equation (\ref{sum-sl2}), paraphrasing chapter IX of \cite{kbi}. The first step is to extract from (\ref{int-sl2}) one more commutation relation, this time between the $C$ and $B$-operators; namely,
\begin{align}
C(\lambda) B(\mu) 
= 
B(\mu) C(\lambda) 
+ 
g(\lambda,\mu) A(\mu) D(\lambda) 
- 
g(\lambda,\mu) A(\lambda) D(\mu) .
\label{2cb}
\end{align}
By virtue of equations (\ref{2ab}), (\ref{2db}) and (\ref{2cb}), one can develop an algorithm for calculating the scalar product (\ref{sp-def}), which we describe below. 

{\bf 1.} Using (\ref{2cb}), we exchange the pair of operators at the center of the scalar product, say $C(\lambda^C_{\ell}) B(\lambda^B_{\ell})$. Three terms are obtained in the process, one of which is simply $B(\lambda^B_{\ell}) C(\lambda^C_{\ell})$. {\bf 2.} Repeat step {\bf 1} for the term containing $C(\lambda^C_{\ell})$, a further $\ell-1$ times. After the final iteration, $C(\lambda^C_{\ell})$ acts on $|0\rangle$, and annihilates it. {\bf 3.} After this process, the surviving terms are all of the form 
\begin{align*}
\langle 0| C(\lambda^C_1) \dots C(\lambda^C_{\ell-1}) O_{\ell} \dots O_{0} |0\rangle  
\end{align*}
where a single pair of operators $O_i O_{i-1}$ is equal to $A(\lambda^B_i) D(\lambda^C_{\ell})$ or 
$A(\lambda^C_{\ell}) D(\lambda^B_i)$ and all the rest are $B$-operators. Using (\ref{2ab}) and (\ref{2db}) repeatedly, all $A$ and $D$-operators can ultimately be transferred to the right so that they act on $|0\rangle$. Since $|0\rangle$ is an eigenvector of $A$ and $D$, these operators are replaced by their respective eigenfunctions $a$ and $d$. {\bf 4.} The result of this process is a sum over terms of the form
\begin{align*}
\langle 0| C(\lambda^C_1) \dots C(\lambda^C_{\ell-1}) B_{\ell-1} \dots B_{1} |0\rangle
\end{align*}
where each $B_i$ is a $B$-operator but whose argument is not necessarily $\lambda^B_i$. Hence we have reduced the original object to a sum over scalar products of one dimension smaller. Now return to step {\bf 1} and repeat the entire process.

At the end of the algorithm, we find that
\begin{align}
\langle 0|
\prod_{i=1}^{\ell}
C(\lambda^C_i)
\prod_{j=1}^{\ell}
B(\lambda^B_j)
|0\rangle
=
\sum
\prod_{\lbi}
a(\lbi)
\prod_{\lcii}
a(\lcii)
\prod_{\lbii}
d(\lbii)
\prod_{\lci}
d(\lci)
K\left( \begin{array}{c|c} \{\lci\} & \{\lcii\} \\ \{\lbi\} & \{\lbii\} \end{array} \right)
\label{sp-sum1}
\end{align}
where the sum is over all ways of partitioning $\{\lambda^C\}$ and $\{\lambda^B\}$ into disjoint subsets, 
{\it but} such that each term contains a product of exactly $\ell$ $a$ eigenfunctions and $\ell$ $d$ eigenfunctions. Clearly, this is equivalent to the constraint that $|\lci| = |\lbi|$ and $|\lcii| = |\lbii|$. The coefficients $K$ are {\it a priori} unknown, and potentially very complicated combinations of the functions $f$ and $g$ which appear as coefficients during the course of the algorithm.

To determine the coefficients, one makes use of the fact that $K$ {\it does not} depend on the eigenfunctions $a$ and $d$. Hence these coefficients are the same for any $SU(2)$-invariant model obeying the set of axioms (\ref{mon-sl2})--(\ref{ps-vac2}), and we are free to specialize to one such model in what follows. To that end, we choose the monodromy matrix of a length $\ell$ XXX Heisenberg spin-1/2 chain. The accompanying pseudo-vacuum is the state with all spins up. Explicitly, we set
\begin{align}
T_{\alpha}(\lambda)
=
R_{\alpha 1}(\lambda,w_1) \dots R_{\alpha \ell}(\lambda,w_{\ell}),
\quad\quad
|0\rangle
=
\bigotimes_{i=1}^{\ell} \left( \begin{array}{c} 1 \\ 0 \end{array} \right)_i\,,
\quad\quad
\langle 0|
=
\bigotimes_{i=1}^{\ell} \left( \begin{array}{cc} 1 & 0 \end{array} \right)_i\,
\label{sl2-mon-psvac}
\end{align}
where each $R_{\alpha i}(\lambda,w_i)$ is an $SU(2)$-invariant $R$-matrix (\ref{Rmat-sl2}) acting in $V_{\alpha} \otimes V_i$, and $V_i$ are the vector spaces associated to each site in the spin chain, $1 \leq i \leq \ell$. It is a simple calculation to show that, for this model, the eigenfunctions $a$ and $d$ are given by
\begin{align}
a(\lambda)
=
\prod_{i=1}^{\ell}
\frac{(\lambda-w_i+1)}{(\lambda-w_i)}
=
\prod_{i=1}^{\ell}
f(\lambda,w_i),
\quad\quad
d(\lambda)
=
1.
\end{align}
The key point regarding the specialization to this model is that the parameters $w_i$ (usually called the inhomogeneities, or quantum variables of the spin chain) do not appear in the coefficients $K$. Hence we are free to choose them as we see fit. We firstly change the normalization of the scalar product by defining
\begin{align}
S(\{\lc\} | \{\lb\})
=
\frac{\langle \{\lc\} | \{\lb\} \rangle}{f(\{\lc\},\{w\}) f(\{\lb\},\{w\})} 
\end{align}
and consider the limit
\begin{align}
\{w_1,\dots,w_{\ell}\}
\rightarrow
\{\lcii\} \cup \{\lbi\}
\label{spec}
\end{align}
where $\{\lcii\} \subseteq \{\lc\}$ and $\{\lbi\} \subseteq \{\lb\}$ are fixed subsets of the original variables, such that $|\lcii|+|\lbi| = \ell$. Making this choice, the only term which survives in the sum (\ref{sp-sum1}) is the one for which $\{\lcii\}$ and $\{\lbi\}$ are the arguments of the $a$ eigenfunctions. It follows that
\begin{align}
S(\{\lc\} | \{\lb\}) \Big|_{\{w\} \rightarrow \{\lcii\} \cup \{\lbi\}}
=
\Big( f(\lci,\lbi) f(\lbii,\lbi) f(\lci,\lcii) f(\lbii,\lcii) \Big)^{-1}
K\left( \begin{array}{c|c} \{\lci\} & \{\lcii\} \\ \{\lbi\} & \{\lbii\} \end{array} \right) .
\label{coeff-isol}
\end{align}
To finish the calculation, it remains to explicitly evaluate the scalar product of this spin chain under the specialization (\ref{spec}) of the inhomogeneities. For that, we use the graphical representation of the scalar product (see, for example, \cite{whe}) as shown in Figure {\bf\ref{fig-sl2-sp}}.
\begin{figure}

\begin{center}
\begin{minipage}{4.3in}

\setlength{\unitlength}{0.00028cm}
\begin{picture}(40000,18500)(0000,7000)

\put(1000,16400){$
\langle \{\lambda^C\} | \{\lambda^B\} \rangle 
=
$}


\put(11500,23000){\tiny $\lbii$}
\put(13000,22700){$\left\{\phantom{\frac{\prod}{\prod}}\right.$} 

\path(14000,24000)(24000,24000) 
\put(14000,24500){\tiny 1}  
\put(23500,24500){\tiny 2}

\path(14000,22000)(24000,22000) 
\put(14000,22500){\tiny 1}  
\put(23500,22500){\tiny 2}

\put(11500,19000){\tiny $\lbi$}
\put(13000,18700){$\left\{\phantom{\frac{\prod}{\prod}}\right.$} 

\path(14000,20000)(24000,20000)
\put(14000,20500){\tiny 1}  
\put(23500,20500){\tiny 2}

\path(14000,18000)(24000,18000) 
\put(14000,18500){\tiny 1}  
\put(23500,18500){\tiny 2}


\put(11500,15000){\tiny $\lcii$}
\put(13000,14700){$\left\{\phantom{\frac{\prod}{\prod}}\right.$} 

\path(14000,16000)(24000,16000)
\put(14000,16500){\tiny 2}  
\put(23500,16500){\tiny 1}

\path(14000,14000)(24000,14000)
\put(14000,14500){\tiny 2}  
\put(23500,14500){\tiny 1}

\put(11500,11000){\tiny $\lci$}
\put(13000,10700){$\left\{\phantom{\frac{\prod}{\prod}}\right.$} 

\path(14000,12000)(24000,12000)
\put(14000,12500){\tiny 2} 
\put(23500,12500){\tiny 1}

\path(14000,10000)(24000,10000)
\put(14000,10500){\tiny 2} 
\put(23500,10500){\tiny 1}


\path(16000,8000)(16000,26000)
\put(16000,6000){\tiny $w_1$}
\put(16000,7000){\tiny 1}
\put(16200,16800){\tiny\color{red} 2}
\put(16000,26500){\tiny 1} 

\path(18000,8000)(18000,26000)
\put(18000,7000){\tiny 1}
\put(18200,16800){\tiny\color{red} 2}
\put(18000,26500){\tiny 1} 
 
\path(20000,8000)(20000,26000)
\put(20000,7000){\tiny 1}
\put(20200,16800){\tiny\color{red} 2}
\put(20000,26500){\tiny 1} 
 
\path(22000,8000)(22000,26000)
\put(22000,6000){\tiny $w_{\ell}$}
\put(22000,7000){\tiny 1}
\put(22200,16800){\tiny\color{red} 2}
\put(22000,26500){\tiny 1}


\put(27000,16400){$=$}

\put(29000,19400){$Z(\{\lbi\}\cup\{\lbii\}|\{w\})$}

\put(32000,16400){$\times$}

\put(29000,13400){$Z(\{\lci\}\cup\{\lcii\}|\{w\})$}

\end{picture}

\end{minipage}
\end{center}

\caption{The graphical representation of the scalar product (\ref{sp-def}), in the case of the monodromy matrix (\ref{sl2-mon-psvac}). Each of the top $\ell$ horizontal lines represent $B(\lambda^B_i)$ operators, while the bottom $\ell$ horizontal lines represent $C(\lambda^C_i)$ operators. The pseudo-vacua in (\ref{sl2-mon-psvac}) are indicated by the state variables at the ends of each vertical line. Since the number of vertical sites is also $\ell$, this lattice factorizes into a product of two domain wall partition functions, as shown.}

\label{fig-sl2-sp}

\end{figure}
From Figure {\bf\ref{fig-sl2-sp}}, we see that the scalar product factorizes as
\begin{align}
\langle \{\lambda^C\} | \{\lambda^B\} \rangle
=
Z(\{\lbi\}\cup\{\lbii\} | \{w\})
Z(\{\lci\}\cup\{\lcii\} | \{w\})
\end{align}
where $Z$ denotes a domain wall partition function of Subsection {\bf\ref{ssec:dom}}. In view of this factorization, using (\ref{res-prop}) we find that
\begin{align}
S(\{\lc\} | \{\lb\}) \Big|_{\{w\} \rightarrow \{\lcii\} \cup \{\lbi\}}
=
\frac{Z(\lbii | \lcii) Z(\lci | \lbi)}{f(\lbii,\lcii) f(\lci,\lbi)} .
\label{xxx-fact}
\end{align}
Comparing (\ref{coeff-isol}) and (\ref{xxx-fact}), we have shown that
\begin{align}
K\left( \begin{array}{c|c} \{\lci\} & \{\lcii\} \\ \{\lbi\} & \{\lbii\} \end{array} \right)
=
f(\lbii,\lbi) f(\lci,\lcii) Z(\lbii | \lcii) Z(\lci | \lbi) .
\end{align}
Combining this with (\ref{sp-sum1}) completes the proof of (\ref{sum-sl2}).

\section{Derivation of the sum form for the $SU(3)$-invariant scalar product} 
\label{app:2}

In close analogy with the previous appendix, we describe here the derivation of the sum formula (\ref{sum-sl3}). The approach used is the same as in the original article \cite{res}, but we go into slightly more detail. As our starting point, we will assume that the scalar product can be written in the form 
\begin{multline}
\label{sl3-sp-sum1}
\langle \{\mc\}, \{\lc\} | \{\lb\}, \{\mb\} \rangle
=
\sum
\prod_{\lbi} a_1(\lbi) \prod_{\lcii} a_1(\lcii)
\prod_{\lbii} a_2(\lbii) \prod_{\lci} a_2(\lci)
\\
\times
\prod_{\mbii} a_2(\mbii) \prod_{\mci} a_2(\mci)
\prod_{\mbi} a_3(\mbi) \prod_{\mcii} a_3(\mcii)
K\left( 
\begin{array}{c|c||c|c} 
\{\lci\} & \{\lcii\} & \{\mci\} & \{\mcii\} \\ \{\lbi\} & \{\lbii\} & \{\mbi\} & \{\mbii\} 
\end{array} 
\right)
\end{multline}
where the sum is over all partitions of the variables, such that $|\lci| = |\lbi|$, $|\lcii| = |\lbii|$, 
$|\mci| = |\mbi|$, $|\mcii| = |\mbii|$, and $K$ are coefficients which depend on the partitioning. This fact is proved in an algorithmic fashion similar to that in Appendix {\bf\ref{app:1}}. However the situation is substantially more complicated now, since $|\{\lambda^B\},\{\mu^B\} \rangle$ is a sum over monomials of the operators $\{t_{12},t_{13},t_{22},t_{23}\}$ acting on $|0\rangle$, while $\langle \{\mu^C\},\{\lambda^C\}|$ is a sum over monomials of $\{t_{21},t_{31},t_{22},t_{32}\}$ acting on $\langle 0|$. Hence it is necessary to extract from (\ref{int-sl3}) a host of commutation relations for the algebraic manipulation of these operators, which can all be written in the form 
\begin{align}
\label{sl3-gen-com}
t_{il}(\lambda) t_{jk}(\mu) + g(\lambda,\mu) t_{jl}(\lambda) t_{ik}(\mu)
=
t_{jk}(\mu) t_{il}(\lambda) + g(\lambda,\mu) t_{jl}(\mu) t_{ik}(\lambda). 
\end{align}
Making successive use of (\ref{sl3-gen-com}) with differing values of $i,j,k,l$, and using the rules 
(\ref{ps-vac3}) for the action on the pseudo-vacuum, one can conclude that (\ref{sl3-sp-sum1}) holds.
 
To calculate the coefficients $K$, we proceed in close analogy with Appendix {\bf\ref{app:1}}. As before, the key point is that $K$ is independent of the eigenfunctions $a_1,a_2,a_3$, meaning that we are free to specialize to any $SU(3)$-invariant model satisfying the list of axioms (\ref{mon-sl3})--(\ref{ps-vac4}). We choose the monodromy matrix of a length $(\ell+m)$ XXX Heisenberg chain, whose first $\ell$ sites are in the fundamental representation and whose last $m$ sites are in the anti-fundamental representation:
\begin{align}
\label{sl3-mon}
T^{(1)}_{\alpha}(\lambda)
=
R^{(1)}_{\alpha 1}(\lambda,w_1) 
\dots 
R^{(1)}_{\alpha \ell}(\lambda,w_{\ell})
R^{*(1)}_{\alpha 1'}(\lambda,v_1) 
\dots 
R^{*(1)}_{\alpha m'}(\lambda,v_m) 
\end{align}
where each $R^{(1)}_{\alpha i}(\lambda,w_i)$ and $R^{*(1)}_{\alpha j'}(\lambda,v_j)$ is an $SU(3)$-invariant $R$-matrix, given by (\ref{Rmat-sl3}) and (\ref{Rmat-sl3t}), acting in $V_{\alpha} \otimes V_i$ and 
$V_{\alpha} \otimes V_{j'}$ respectively. For the pseudo-vacuum states, one chooses
\begin{align}
\label{sl3-psvac}
|0\rangle
=
\bigotimes_{i=1}^{\ell} 
\left( \begin{array}{c} 1 \\ 0 \\ 0 \end{array} \right)_i\ 
\bigotimes_{j=1}^{m}
\left( \begin{array}{c} 0 \\ 0 \\ 1 \end{array} \right)_{j'}\,,
\quad\quad
\langle 0|
=
\bigotimes_{i=1}^{\ell} 
\left( \begin{array}{ccc} 1 & 0 & 0 \end{array} \right)_i\ 
\bigotimes_{j=1}^{m} 
\left( \begin{array}{ccc} 0 & 0 & 1 \end{array} \right)_{j'}\ .
\end{align}
For this model, we find that the eigenfunctions $a_1,a_2,a_3$ are given by
\begin{align}
a_1(\lambda)
=
\prod_{i=1}^{\ell} \frac{(\lambda - w_i + 1)}{(\lambda - w_i)},
\qquad
a_2(\lambda)
=
1,
\qquad
a_3(\lambda)
=
\prod_{j=1}^{m} \frac{(v_j-\lambda+1)}{(v_j-\lambda)} .
\end{align}
We have introduced two sets of parameters $w_i$ and $v_j$ which do not appear in the coefficients $K$, and which we may specialize in any way. We firstly change the normalization of the scalar product by defining
\begin{align}
S(\{\mu^C\},\{\lambda^C\} | \{\lambda^B\},\{\mu^B\})
=
\frac{\langle \{\mu^C\},\{\lambda^C\} | \{\lambda^B\},\{\mu^B\} \rangle}
{f(\{\lambda^C\},\{w\}) f(\{\lambda^B\},\{w\}) f(\{v\},\{\mu^C\}) f(\{v\},\{\mu^B\})}
\end{align}
and consider the limit
\begin{align}
\{w_1,\dots,w_{\ell}\} \rightarrow \{\lcii\}\cup\{\lbi\},
\qquad
\{v_1,\dots,v_m\} \rightarrow \{\mcii\}\cup\{\mbi\}
\label{spec-sl3}
\end{align}
where $\{\lcii\} \subseteq \{\lc\}$, $\{\lbi\} \subseteq \{\lb\}$, $\{\mcii\} \subseteq \{\mc\}$, 
$\{\mbi\} \subseteq \{\mb\}$ are fixed subsets of the original variables, such that 
$|\lcii|+|\lbi| = \ell$ and $|\mcii|+|\mbi| = m$. Making this choice, the only term which survives in the sum (\ref{sl3-sp-sum1}) is the one for which $\{\lcii\},\{\lbi\}$ and $\{\mcii\},\{\mbi\}$ are arguments of the $a_1$ and $a_3$ eigenfunctions, respectively. Hence we find that
\begin{multline}
S( \{\mc\}, \{\lc\} | \{\lb\}, \{\mb\} )
\Big|_{\substack{\{w\} \rightarrow \{\lcii\}\cup\{\lbi\} \\ \{v\} \rightarrow \{\mcii\}\cup\{\mbi\}}}
=
\Big( f(\lci,\lbi) f(\lbii,\lbi) f(\lci,\lcii) f(\lbii,\lcii) \Big)^{-1}
\\
\times
\Big( f(\mbi,\mci) f(\mbi,\mbii) f(\mcii,\mci) f(\mcii,\mbii) \Big)^{-1}
K\left( 
\begin{array}{c|c||c|c} 
\{\lci\} & \{\lcii\} & \{\mci\} & \{\mcii\} \\ \{\lbi\} & \{\lbii\} & \{\mbi\} & \{\mbii\} 
\end{array} 
\right) .
\label{138}
\end{multline}
To obtain a more explicit expression for $K$, it is again helpful to use the graphical representation of the scalar product of this spin chain (see, for example, \cite{fw2}). This is given in Figure {\bf\ref{fig-sl3-sp}}.
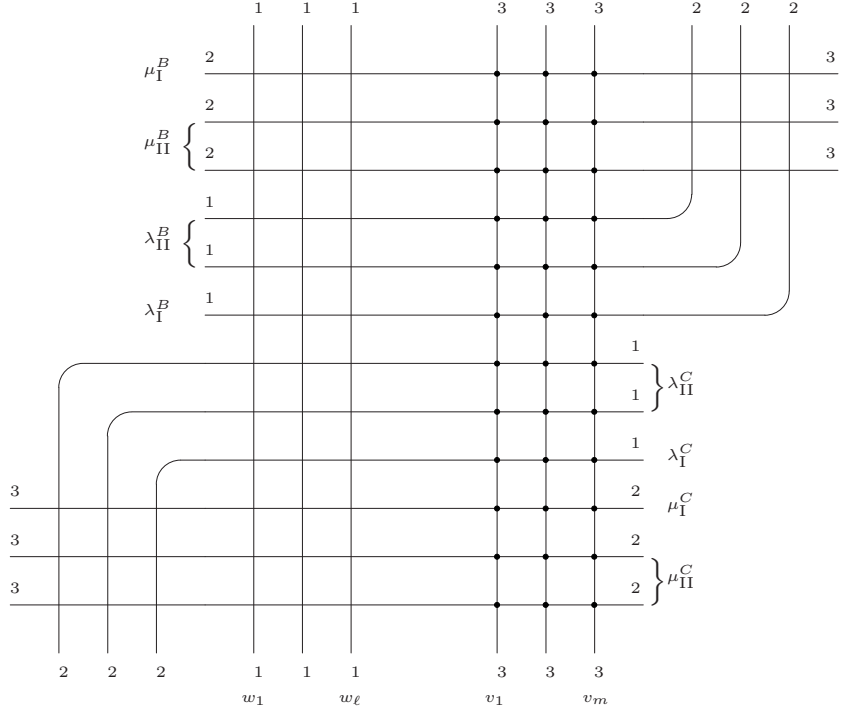
\begin{figure}

\begin{center}
\begin{minipage}{4.3in}

\setlength{\unitlength}{0.00032cm}
\begin{picture}(40000,28000)(6000,-3000)


\put(11500,22000){\tiny $\mbi$}

\put(11500,19000){\tiny $\mbii$}
\put(13000,18700){$\left\{\phantom{\frac{\prod}{\prod}}\right.$} 

\put(11500,15000){\tiny $\lbii$}
\put(13000,14700){$\left\{\phantom{\frac{\prod}{\prod}}\right.$}

\put(11500,12000){\tiny $\lbi$}

\path(14000,22000)(32000,22000)
 \put(14000,22500){\tiny 2}
 \put(39500,22500){\tiny 3}

\path(14000,20000)(32000,20000)
 \put(14000,20500){\tiny 2}
 \put(39500,20500){\tiny 3}

\path(14000,18000)(32000,18000)
 \put(14000,18500){\tiny 2}
 \put(39500,18500){\tiny 3}

\path(14000,16000)(32000,16000)
 \put(14000,16500){\tiny 1}

\path(14000,14000)(32000,14000)
\put(14000,14500){\tiny 1}

\path(14000,12000)(32000,12000)
 \put(14000,12500){\tiny 1}


\path(32000,22000)(40000,22000)
\path(32000,20000)(40000,20000)
\path(32000,18000)(40000,18000)

\path(32000,16000)(33000,16000)
\put(33000,17000){\arc{2000}{0}{1.5708}}
\path(34000,17000)(34000,24000)
\put(34000,24500){\tiny 2}

\path(32000,14000)(35000,14000)
\put(35000,15000){\arc{2000}{0}{1.5708}}
\path(36000,15000)(36000,24000)
\put(36000,24500){\tiny 2}

\path(32000,12000)(37000,12000)
\put(37000,13000){\arc{2000}{0}{1.5708}}
\path(38000,13000)(38000,24000)
\put(38000,24500){\tiny 2}


\put(33000,9000){\tiny $\lcii$}
\put(31000,8700){$\left.\phantom{\frac{\prod}{\prod}}\right\}$}

\put(33000,6000){\tiny $\lci$}

\put(33000,4000){\tiny $\mci$}

\put(33000,1000){\tiny $\mcii$}
\put(31000,0700){$\left.\phantom{\frac{\prod}{\prod}}\right\}$}

\path(14000,10000)(32000,10000)
 \put(31500,10500){\tiny 1}

\path(14000,8000)(32000,8000)
 \put(31500,8500){\tiny 1}

\path(14000,6000)(32000,6000)
 \put(31500,6500){\tiny 1}

\path(14000,4000)(32000,4000)
 \put(6000,4500){\tiny 3}
 \put(31500,4500){\tiny 2}

\path(14000,2000)(32000,2000)
 \put(6000,2500){\tiny 3}
 \put(31500,2500){\tiny 2}

\path(14000,0000)(32000,0000)
 \put(6000,500){\tiny 3}
 \put(31500,500){\tiny 2}


\path(9000,10000)(14000,10000)
\put(9000,9000){\arc{2000}{3.142}{4.712}}
\path(8000,9000)(8000,-2000)
\put(8000,-3000){\tiny 2}

\path(11000,8000)(14000,8000)
\put(11000,7000){\arc{2000}{3.142}{4.712}}
\path(10000,7000)(10000,-2000)
\put(10000,-3000){\tiny 2}

\path(13000,6000)(14000,6000)
\put(13000,5000){\arc{2000}{3.142}{4.712}}
\path(12000,5000)(12000,-2000)
\put(12000,-3000){\tiny 2}

\path(6000,4000)(14000,4000)
\path(6000,2000)(14000,2000)
\path(6000,0000)(14000,0000)


\path(16000,-2000)(16000,24000)
\put(15500,-4000){\tiny $w_1$}
\put(16000,-3000){\tiny 1}
\put(16000,24500){\tiny 1}

\path(18000,-2000)(18000,24000)
\put(18000,-3000){\tiny 1}
\put(18000,24500){\tiny 1}

\path(20000,-2000)(20000,24000)
\put(19500,-4000){\tiny $w_{\ell}$}
\put(20000,-3000){\tiny 1}
\put(20000,24500){\tiny 1}


\path(26000,-2000)(26000,24000)
\put(25500,-4000){\tiny $v_1$}
\put(26000,-3000){\tiny 3}
\put(26000,24500){\tiny 3}

\path(28000,-2000)(28000,24000)
\put(28000,-3000){\tiny 3}
\put(28000,24500){\tiny 3}

\path(30000,-2000)(30000,24000)
\put(29500,-4000){\tiny $v_m$}
\put(30000,-3000){\tiny 3}
\put(30000,24500){\tiny 3}


\put(26000,0){\circle*{200}}
\put(28000,0){\circle*{200}}
\put(30000,0){\circle*{200}}

\put(26000,2000){\circle*{200}}
\put(28000,2000){\circle*{200}}
\put(30000,2000){\circle*{200}}

\put(26000,4000){\circle*{200}}
\put(28000,4000){\circle*{200}}
\put(30000,4000){\circle*{200}}

\put(26000,6000){\circle*{200}}
\put(28000,6000){\circle*{200}}
\put(30000,6000){\circle*{200}}

\put(26000,8000){\circle*{200}}
\put(28000,8000){\circle*{200}}
\put(30000,8000){\circle*{200}}

\put(26000,10000){\circle*{200}}
\put(28000,10000){\circle*{200}}
\put(30000,10000){\circle*{200}}

\put(26000,12000){\circle*{200}}
\put(28000,12000){\circle*{200}}
\put(30000,12000){\circle*{200}}

\put(26000,14000){\circle*{200}}
\put(28000,14000){\circle*{200}}
\put(30000,14000){\circle*{200}}

\put(26000,16000){\circle*{200}}
\put(28000,16000){\circle*{200}}
\put(30000,16000){\circle*{200}}

\put(26000,18000){\circle*{200}}
\put(28000,18000){\circle*{200}}
\put(30000,18000){\circle*{200}}

\put(26000,20000){\circle*{200}}
\put(28000,20000){\circle*{200}}
\put(30000,20000){\circle*{200}}

\put(26000,22000){\circle*{200}}
\put(28000,22000){\circle*{200}}
\put(30000,22000){\circle*{200}}

\end{picture}

\end{minipage}
\end{center}

\caption{The graphical representation of the scalar product (\ref{sl3-sp-def}), in the case of the monodromy matrix (\ref{sl3-mon}). In the top half of the lattice, the top $m$ horizontal lines represent $B^{(2)}(\mu^B_i)$ operators, and the following $\ell$ lines represent $B^{(1)}(\lambda^B_i)$ operators. In the bottom half, the first $\ell$ lines represent $C^{(1)}(\lambda^C_i)$ operators, and the following $m$ horizontal lines represent $C^{(2)}(\mu^C_i)$ operators. The state variables at the ends of the vertical lines correspond with the pseudo-vacua (\ref{sl3-psvac}).}

\label{fig-sl3-sp}

\end{figure}
Unlike in the $SU(2)$ case of Appendix {\bf\ref{app:1}}, the $SU(3)$ scalar product in Figure {\bf\ref{fig-sl3-sp}} does not {\it a priori} factorize into recognizable objects. It is necessary to specialize $\{w\}$ and $\{v\}$ to the values in (\ref{spec-sl3}), which causes the lattice to simplify. We indicate this simplification in Figures {\bf\ref{fig-sl3-sp2}}--{\bf\ref{fig-sl3-sp3}}. 
\begin{figure}[H]

\begin{center}
\begin{minipage}{4.3in}

\setlength{\unitlength}{0.00028cm}
\begin{picture}(40000,30000)(4000,-4000)


\put(25000,22700){$\left\{\phantom{\frac{\prod}{\prod}}\right.$}

\put(23500,23000){\tiny $\mbii$} 

\put(11500,15000){\tiny $\lbii$}
\put(13000,14700){$\left\{\phantom{\frac{\prod}{\prod}}\right.$}

\put(11500,12000){\tiny $\lbi$}

\path(26000,24000)(32000,24000)
 \put(26000,24500){\tiny 2}
 \put(39500,24500){\tiny 3}

\path(26000,22000)(32000,22000)
 \put(26000,22500){\tiny 2}
 \put(39500,22500){\tiny 3}

\path(14000,16000)(32000,16000)
 \put(14000,16500){\tiny 1}

\path(14000,14000)(32000,14000)
\put(14000,14500){\tiny 1}

\path(14000,12000)(19000,12000)
\put(19000,13000){\arc{2000}{0}{1.5708}}
\put(21000,11000){\arc{2000}{3.142}{4.712}}
\path(21000,12000)(32000,12000)
 \put(14000,12500){\tiny 1}


\path(32000,24000)(40000,24000)
\path(32000,22000)(40000,22000)

\path(32000,16000)(33000,16000)
\put(33000,17000){\arc{2000}{0}{1.5708}}
\path(34000,17000)(34000,26000)
\put(34000,26500){\tiny 2}

\path(32000,14000)(35000,14000)
\put(35000,15000){\arc{2000}{0}{1.5708}}
\path(36000,15000)(36000,26000)
\put(36000,26500){\tiny 2}

\path(32000,12000)(37000,12000)
\put(37000,13000){\arc{2000}{0}{1.5708}}
\path(38000,13000)(38000,26000)
\put(38000,26500){\tiny 2}


\put(3500,0000){\tiny $\mci$}

\path(14000,10000)(17000,10000)
\put(17000,11000){\arc{2000}{0}{1.5708}}
\put(19000,9000){\arc{2000}{3.142}{4.712}}
\path(19000,10000)(32000,10000)
 \put(31500,10500){\tiny 1}
 
\path(14000,8000)(15000,8000)
\put(15000,9000){\arc{2000}{0}{1.5708}}
\put(17000,7000){\arc{2000}{3.142}{4.712}}
\path(17000,8000)(32000,8000)
 \put(31500,8500){\tiny 1}

\path(14000,6000)(32000,6000)
 \put(31500,6500){\tiny 1}

\path(14000,000)(28000,000)
\put(6000,500){\tiny 3}
\put(27500,500){\tiny 2}


\path(9000,10000)(14000,10000)
\put(9000,9000){\arc{2000}{3.142}{4.712}}
\path(8000,9000)(8000,-2000)
\put(8000,-3000){\tiny 2}

\path(11000,8000)(14000,8000)
\put(11000,7000){\arc{2000}{3.142}{4.712}}
\path(10000,7000)(10000,-2000)
\put(10000,-3000){\tiny 2}

\path(13000,6000)(14000,6000)
\put(13000,5000){\arc{2000}{3.142}{4.712}}
\path(12000,5000)(12000,-2000)
\put(12000,-3000){\tiny 2}

\path(6000,000)(14000,000)

\put(8300,-4600){\tiny $\lcii$}
\put(7950,-3000){$\underbrace{\phantom{...}}$}

\put(11700,-4600){\tiny $\lci$}


\path(16000,4000)(16000,7000)
\path(16000,9000)(16000,18000)
\put(16000,3000){\tiny 1}
\put(16000,18500){\tiny 1}

\path(18000,4000)(18000,9000)
\path(18000,11000)(18000,18000)
\put(18000,3000){\tiny 1}
\put(18000,18500){\tiny 1}

\path(20000,4000)(20000,11000)
\path(20000,13000)(20000,18000)
\put(20000,3000){\tiny 1}
\put(20000,18500){\tiny 1}

\put(16300,1400){\tiny $\lcii$}
\put(15950,3000){$\underbrace{\phantom{...}}$}

\put(19700,1400){\tiny $\lbi$}


\path(26000,-2000)(26000,18000)
\put(25500,-4600){\tiny $\mbi$}
\put(26000,-3000){\tiny 3}
\put(26000,18500){\tiny 2}

\path(28000,4000)(28000,26000)
\put(28000,3000){\tiny 2}
\put(28000,26500){\tiny 3}

\path(30000,4000)(30000,26000)
\put(30000,3000){\tiny 2}
\put(30000,26500){\tiny 3}

\put(28300,1400){\tiny $\mcii$}
\put(27950,3000){$\underbrace{\phantom{...}}$}


\put(26000,0000){\circle*{200}}

\put(26000,6000){\circle*{200}}
\put(28000,6000){\circle*{200}}
\put(30000,6000){\circle*{200}}

\put(26000,8000){\circle*{200}}
\put(28000,8000){\circle*{200}}
\put(30000,8000){\circle*{200}}

\put(26000,10000){\circle*{200}}
\put(28000,10000){\circle*{200}}
\put(30000,10000){\circle*{200}}

\put(26000,12000){\circle*{200}}
\put(28000,12000){\circle*{200}}
\put(30000,12000){\circle*{200}}

\put(26000,14000){\circle*{200}}
\put(28000,14000){\circle*{200}}
\put(30000,14000){\circle*{200}}

\put(26000,16000){\circle*{200}}
\put(28000,16000){\circle*{200}}
\put(30000,16000){\circle*{200}}

\put(28000,22000){\circle*{200}}
\put(30000,22000){\circle*{200}}

\put(28000,24000){\circle*{200}}
\put(30000,24000){\circle*{200}}

\end{picture}

\end{minipage}
\end{center}

\caption{Up to normalization of some of the vertices (which we suppress for simplicity, though it is obviously essential to the final answer) $S( \{\mc\}, \{\lc\} | \{\lb\}, \{\mb\} )$ simplifies to the lattice shown when $\{w\} \rightarrow \{\lcii\} \cup \{\lbi\}$ and $\{v\} \rightarrow \{\mcii\} \cup \{\mbi\}$. Due to the form of the renormalized Boltzmann weights, any vertex with the same horizontal and vertical variable is forbidden to be in certain configurations, while all its allowed configurations have weight 1. This causes the splitting of $\ell+m$ vertices in the lattice, $\ell$ of which are shown above (the remaining $m$ splittings were amongst the dotted vertices, but are no longer explicitly visible).}

\label{fig-sl3-sp2}

\end{figure}
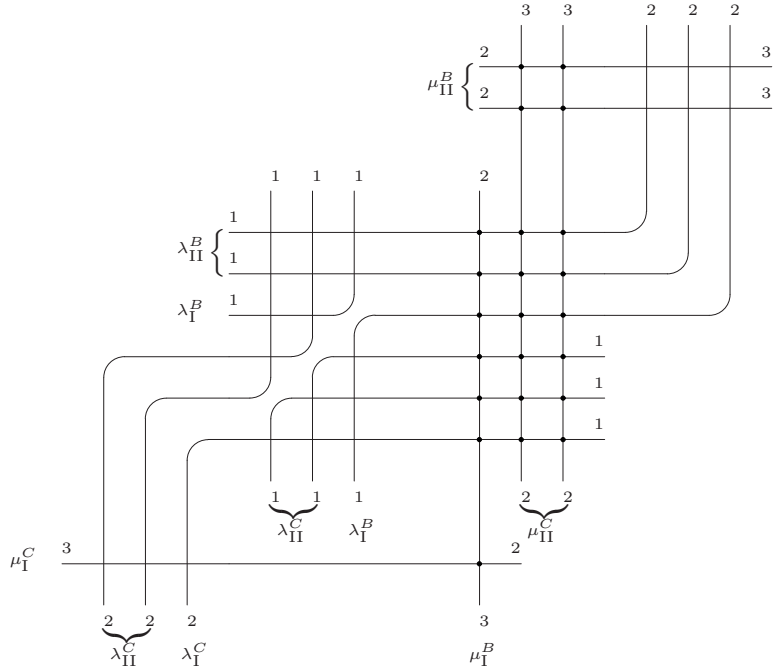
\begin{figure}[H]

\begin{center}
\begin{minipage}{4.3in}

\setlength{\unitlength}{0.00028cm}
\begin{picture}(40000,30000)(4000,-4000)


\put(25000,22700){$\left\{\phantom{\frac{\prod}{\prod}}\right.$}

\put(23500,23000){\tiny $\mbii$} 

\put(3500,15000){\tiny $\lbii$}
\put(5000,14700){$\left\{\phantom{\frac{\prod}{\prod}}\right.$}

\path(26000,24000)(32000,24000)
 \put(26000,24500){\tiny 2}
 \put(39500,24500){\tiny 3}

\path(26000,22000)(32000,22000)
 \put(26000,22500){\tiny 2}
 \put(39500,22500){\tiny 3}

\path(6000,16000)(32000,16000)
 \put(6000,16500){\tiny 1}

\path(6000,14000)(32000,14000)
\put(6000,14500){\tiny 1}


\path(32000,24000)(40000,24000)
\path(32000,22000)(40000,22000)

\path(32000,16000)(33000,16000)
\put(33000,17000){\arc{2000}{0}{1.5708}}
\path(34000,17000)(34000,26000)
\put(34000,26500){\tiny 2}

\path(32000,14000)(35000,14000)
\put(35000,15000){\arc{2000}{0}{1.5708}}
\path(36000,15000)(36000,26000)
\put(36000,26500){\tiny 2}

\path(38000,4000)(38000,26000)
\put(38000,3000){\tiny 1}
\put(38000,26500){\tiny 2}

\put(37700,1400){\tiny $\lbi$}


\put(3500,0000){\tiny $\mci$}

\path(14000,6000)(40000,6000)
 \put(39500,6500){\tiny 1}

\path(14000,000)(28000,000)
\put(6000,500){\tiny 3}
\put(27500,500){\tiny 2}


\path(8000,18000)(8000,-2000)
\put(8000,18500){\tiny 1}
\put(8000,-3000){\tiny 2}

\path(10000,18000)(10000,-2000)
\put(10000,18500){\tiny 1}
\put(10000,-3000){\tiny 2}

\path(13000,6000)(14000,6000)
\put(13000,5000){\arc{2000}{3.142}{4.712}}
\path(12000,5000)(12000,-2000)
\put(12000,-3000){\tiny 2}

\path(6000,000)(14000,000)

\put(8300,-4600){\tiny $\lcii$}
\put(7950,-3000){$\underbrace{\phantom{...}}$}

\put(11700,-4600){\tiny $\lci$}



\path(26000,-2000)(26000,18000)
\put(25500,-4600){\tiny $\mbi$}
\put(26000,-3000){\tiny 3}
\put(26000,18500){\tiny 2}

\path(28000,4000)(28000,26000)
\put(28000,3000){\tiny 2}
\put(28000,26500){\tiny 3}

\path(30000,4000)(30000,26000)
\put(30000,3000){\tiny 2}
\put(30000,26500){\tiny 3}

\put(28300,1400){\tiny $\mcii$}
\put(27950,3000){$\underbrace{\phantom{...}}$}


\put(26000,0000){\circle*{200}}

\put(26000,6000){\circle*{200}}
\put(28000,6000){\circle*{200}}
\put(30000,6000){\circle*{200}}

\put(26000,14000){\circle*{200}}
\put(28000,14000){\circle*{200}}
\put(30000,14000){\circle*{200}}

\put(26000,16000){\circle*{200}}
\put(28000,16000){\circle*{200}}
\put(30000,16000){\circle*{200}}

\put(28000,22000){\circle*{200}}
\put(30000,22000){\circle*{200}}

\put(28000,24000){\circle*{200}}
\put(30000,24000){\circle*{200}}

\end{picture}

\end{minipage}
\end{center}

\caption{The lattice shown in Figure {\bf\ref{fig-sl3-sp2}} can be simplified through successive uses of the Yang-Baxter equations (\ref{yb1}) and (\ref{yb2}) to the form shown above. Once again, this is correct up to an overall normalization, which we suppress here. Further applications of (\ref{yb2}) to this lattice cause it to manifestly factorize into two of the partition functions shown in Figure {\bf\ref{resh-pf}}, where one of these partition functions is rotated by 180\degree.}

\label{fig-sl3-sp3}

\end{figure}
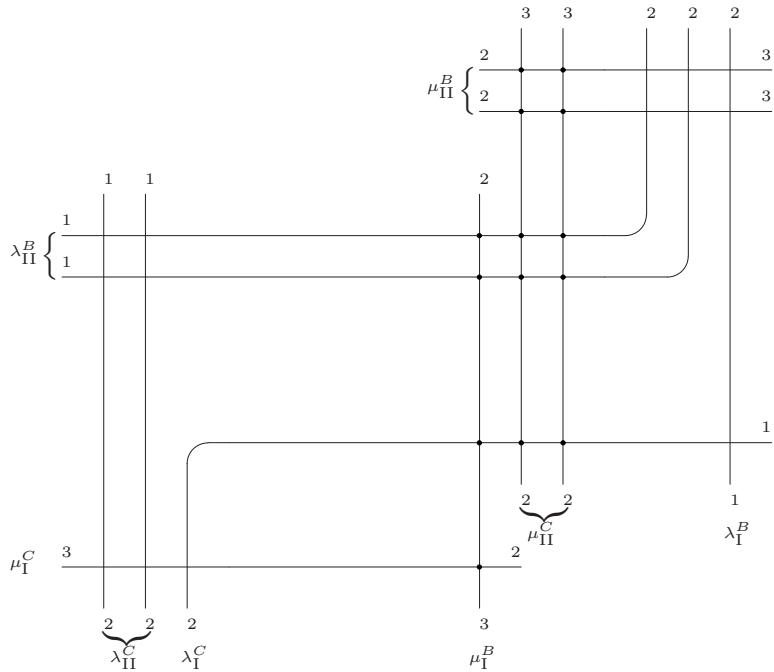
It is a slightly intricate process to prove the equality of the lattices shown in these figures, involving repeated use of the Yang-Baxter equations (\ref{yb1}) and (\ref{yb2}). This calculation is not difficult but we omit its full details since, as was remarked in \cite{res}, it is awkward and would require a lengthy digression. The final result of the calculation (keeping track of normalizations and all multiplicative terms acquired in the process) is the factorization
\begin{multline}
S( \{\mc\}, \{\lc\} | \{\lb\}, \{\mb\} )
\Big|_{\substack{\{w\} \rightarrow \{\lcii\}\cup\{\lbi\} \\ \{v\} \rightarrow \{\mcii\}\cup\{\mbi\}}}
=
\Big(
f(\lci,\lbi) f(\lbii,\lcii) f(\mbi,\mci) f(\mcii,\mbii)
\Big)^{-1}
\\
\times
f(\mci,\lci) f(\mbii,\lbii)
Z(\{\lbii\},\{\mci\} | \{\lcii\},\{\mbi\}) Z(\{\lci\},\{\mbii\} | \{\lbi\},\{\mcii\})
\label{139} 
\end{multline}
where $Z$ denotes the partition function of Subsection {\bf\ref{pf-resh}}, or a 180\degree\ rotation thereof, since all vertices are invariant under such a transformation. Comparing equation (\ref{138}) and (\ref{139}), we find that
\begin{multline}
K\left( 
\begin{array}{c|c||c|c} 
\{\lci\} & \{\lcii\} & \{\mci\} & \{\mcii\} \\ \{\lbi\} & \{\lbii\} & \{\mbi\} & \{\mbii\} 
\end{array} 
\right)
=
f(\lbii,\lbi) f(\lci,\lcii) 
f(\mbi,\mbii) f(\mcii,\mci)
\\
\times
f(\mci,\lci) f(\mbii,\lbii)
Z(\{\lbii\},\{\mci\} | \{\lcii\},\{\mbi\}) Z(\{\lci\},\{\mbii\} | \{\lbi\},\{\mcii\}) .
\end{multline}
Combining this with (\ref{sl3-sp-sum1}) completes the proof of (\ref{sum-sl3}).


\begin{thebibliography}{99}
%
\bibitem{kor}
V~E~Korepin,
\textit{Calculation of norms of Bethe wave functions,}
Comm. Math. Phys. {\bf 86} (1982),
391--418
%
\bibitem{ik}
A~G~Izergin, V~E~Korepin,
\textit{The Quantum Inverse Scattering Method approach
to correlation functions,} 
Comm. Math. Phys. {\bf 94} (1984), 
67
%
\bibitem{kbi}
V~E~Korepin, N~M~Bogoliubov, A~G~Izergin,
\textit{Quantum Inverse Scattering Method and Correlation Functions,}
Cambridge University Press (1993)
%
\bibitem{gau}
M~Gaudin,
\textit{La Fonction d'Onde de Bethe,}
Paris: Masson, (1983)
%
\bibitem{sla}
N~A~Slavnov,
\textit{Calculation of scalar products of wave functions and form factors in the framework of the algebraic Bethe Ansatz,}
Theor. Math. Phys. {\bf 79} (1989),
502--508
%
\bibitem{kmt}
N~Kitanine, J~M~Maillet, V~Terras,
{\it Form factors of the XXZ Heisenberg spin-$\frac{1}{2}$ finite chain,}
Nucl. Phys. B {\bf 554} (1999),
647--678,
{\tt arXiv:math-ph/9807020}
%
\bibitem{kmst}
N~Kitanine, J~M~Maillet, N~A~Slavnov, V~Terras,
\textit{On the algebraic Bethe Ansatz approach to the correlation functions of the XXZ spin-1/2 Heisenberg chain,}
{\tt arXiv:hep-th/0505006}
%
\bibitem{bpr}
S~Belliard, S~Pakuliak, E~Ragoucy,
\textit{Bethe Ansatz and Bethe vectors scalar products,}
SIGMA {\bf 6} (2010),
094,
{\tt arXiv:1012.1455}
%
\bibitem{tv2}
V~O~Tarasov, A~Varchenko,
\textit{Asymptotic solutions to the quantized Knizhnik--Zamolodchikov equation and Bethe vectors,}
{\tt arXiv:hep-th/9406060}
%
\bibitem{egsv}
J~Escobedo, N~Gromov, A~Sever, P~Vieira,
\textit{Tailoring three-point functions and integrability,}
JHEP Vol. 2011 No. 9 (2011),
28,
{\tt arXiv:1012.2475}
%
\bibitem{res}
N~Yu~Reshetikhin,
\textit{Calculation of the norm of Bethe vectors in models with $SU(3)$-symmetry,}
Zap. Nauchn. Sem. {\bf 150} (1986),
196--213
%
\bibitem{bprs2}
S~Belliard, S~Pakuliak, E~Ragoucy, N~A~Slavnov,
\textit{Algebraic Bethe Ansatz for scalar products in $SU(3)$-invariant integrable models,}
J. Stat. Mech. (2012) P10017,
{\tt arXiv:1207.0956} 
%
\bibitem{cae}
J~Caetano,
unpublished.
%
\bibitem{ize}
A~G~Izergin,
\textit{Partition function of the six-vertex model in a finite volume,}
Sov. Phys. Dokl. {\bf 32} (1987),
878--879
%
\bibitem{kos1}
I~Kostov,
\textit{Classical limit of the three-point function of $\mathcal{N}=4$ supersymmetric Yang--Mills theory from integrability,}
Phys. Rev. Lett. {\bf 108} (2012),
261604,
{\tt arXiv:1203.6180}
%
\bibitem{kos2}
I~Kostov,
\textit{Three-point function of semiclassical states at weak coupling,}
J. Phys. A {\bf 45} (2012),
494018,
{\tt arXiv:1205.4412} 
%
\bibitem{fw}
O~Foda, M~Wheeler,
\textit{Partial domain wall partition functions,}
JHEP Vol. 2012 No. 7 (2012),
186,
{\tt arXiv:1205.4400}
%
\bibitem{fst}
L~D~Faddeev, E~K~Sklyanin, L~A~Takhtajan,
\textit{Quantum inverse problem method. I,}
Theor. Math. Phys. {\bf 40} (1979),
688--706
%
\bibitem{fad}
L~D~Faddeev,
\textit{How algebraic Bethe Ansatz works for integrable models,}
Les-Houches lecture notes,
{\tt arXiv:hep-th/9605187} 
%
\bibitem{jim2}
M~Jimbo,
{\it Quantum $R$-matrix for the generalized Toda system,}
Comm. Math. Phys. {\bf 102} (1986),
537--547
%
\bibitem{bax1}
R~J~Baxter,
\textit{Exactly Solved Models in Statistical Mechanics,}
Dover (2008)
%
\bibitem{kor2}
V~E~Korepin,
\textit{Analysis of the bilinear relation of the six-vertex model,}
Dokl. Akad. Nauk. SSSR {\bf 265} (1982),
1361--1364
%
\bibitem{tar1}
V~O~Tarasov,
\textit{Structure of quantum $L$ operators for the $R$ matrix of the XXZ model,}
Theor. Math. Phys. {\bf 61} (1984),
1065--1072
%
\bibitem{tar2}
V~O~Tarasov,
\textit{Irreducible monodromy matrices for the $R$ matrix of the XXZ model and local lattice quantum Hamiltonians,}
Theor. Math. Phys. {\bf 63} (1985),
440--454
%
\bibitem{ft}
L~D~Faddeev, L~A~Takhtajan,
\textit{Spectrum and scattering of excitations in the one-dimensional isotropic Heisenberg model,}
J. Sov. Math. {\bf 24} (1984),
241--267
%
\bibitem{kr}
P~P~Kulish, N~Yu~Reshetikhin,
{\it Diagonalization of $gl(n)$ invariant transfer matrices and quantum $N$-wave system (Lee model),}
J. Phys. {\bf A 16} (1983),
L591--L596
%
\bibitem{br}
S~Belliard, E~Ragoucy,
{\it Nested Bethe Ansatz for ``all'' closed spin chains,}
J. Phys. {\bf A 41} (2008),
295202,
{\tt arXiv:0804.2822}
%
\bibitem{tv1}
V~O~Tarasov, A~Varchenko,
\textit{Combinatorial formulae for nested Bethe vectors,}
{\tt arXiv:math/0702277}
%
\bibitem{dri}
V~G~Drinfel'd,
\textit{A new realization of Yangians and quantized affine algebras,}
Sov. Math. Dokl. {\bf 36} (1988),
212
%
\bibitem{mol}
A~Molev,
\textit{Finite-dimensional irreducible representations of twisted Yangians,}
J. Math. Phys. {\bf 39} (1998),
5559--5600,
{\tt arXiv:q-alg/9711022}
%
\bibitem{kkmst}
N~Kitanine, K~Kozlowski, J~M~Maillet, N~A~Slavnov, V~Terras,
\textit{On correlation functions of integrable models associated to the six-vertex $R$-matrix,}
J. Stat. Mech. 0701:P01022 (2007),
{\tt arXiv:hep-th/0611142}
%
\bibitem{bprs1}
S~Belliard, S~Pakuliak, E~Ragoucy, N~A~Slavnov,
\textit{Highest coefficient of scalar products in $SU(3)$-invariant integrable models,}
J. Stat. Mech. (2012) P09003,
{\tt arXiv:1206.4931}
%
\bibitem{fjks}
O~Foda, Y~Jiang, I~Kostov, D~Serban,
\textit{A tree-level 3-point function in the $su(3)$-sector of planar $\mathcal{N}=4$ SYM,}
{\tt arXiv:1302.3539}
%
\bibitem{whe}
M~Wheeler,
\textit{An Izergin--Korepin procedure for calculating scalar products in the six-vertex model,}
Nucl. Phys. B {\bf 852} (2011),
468--507,
{\tt arXiv:1104.2113}
%
\bibitem{fw2}
O~Foda, M~Wheeler,
\textit{Colour-independent partition functions in coloured vertex models,}
Nucl. Phys. B {\bf 871} (2013), 
330--361, 
{\tt arXiv:1301.5158}
%
\end{thebibliography}
\end{document}